\documentclass[12pt,reqno]{amsart}

\usepackage{lscape}
\usepackage{amsmath,amssymb,amsthm}
\usepackage[numbers]{natbib}
\usepackage[dvips]{graphicx}
\usepackage{float}
\usepackage{epsfig,array,multirow}
\usepackage[usenames]{color}
\usepackage{geometry}
\usepackage{setspace}
\allowdisplaybreaks[4]

\theoremstyle{plain}
\newtheorem{theorem}{Theorem}[section]	
\newtheorem{lemma}{Lemma}[section]

\newtheorem{assumption}{Assumption}[section]

\theoremstyle{definition}

\newtheorem{remark}{Remark}[section]

\renewcommand{\qed}{\hfill{\tiny \ensuremath{\blacksquare} }}%

\renewcommand{\Pr}{{\mathrm{P}}}

\begin{document}

\title[Spectral and Post-Spectral Estimators for Grouped Panel Data Models]{Spectral and Post-Spectral Estimators for Grouped Panel Data Models}
\thanks{We are grateful to Tim Armstrong, Stephane Bonhomme, Victor Chernozhukov, Andrew Chesher, Jin Hahn, Yusuke Narita, Martin Weidner, Daniel Wilhelm, Andrei Zeleneev, participants at many seminars, and especially to Arturas Juodis, Andres Santos, and our discussant at Chamberlain's seminar Roger Moon for useful comments. We also thank Martin Weidner and Roger Moon for providing the code for their estimator.}

\author[Chetverikov]{Denis Chetverikov}
\author[Manresa]{Elena Manresa}

\address[D. Chetverikov]{
Department of Economics, UCLA, Bunche Hall, 8283, 315 Portola Plaza, Los Angeles, CA 90095, USA.
}
\email{chetverikov@econ.ucla.edu}

\address[E. Manresa]{
Department of Economics, New York University, 19 West 4th Street, New York, NY 10003, USA.
}
\email{em1849@nyu.edu}

\date{\today.}

\begin{abstract}
In this paper, we develop spectral and post-spectral estimators for grouped panel data models. Both estimators are consistent in the asymptotics where the number of observations $N$ and the number of time periods $T$ simultaneously grow large. In addition, the post-spectral estimator is $\sqrt{NT}$-consistent and asymptotically normal with mean zero under the assumption of well-separated groups even if $T$ is growing much slower than $N$. The post-spectral estimator has, therefore, theoretical properties that are comparable to those of the grouped fixed-effect estimator developed by Bonhomme and Manresa in \cite{BM15}. In contrast to the grouped fixed-effect estimator, however, our post-spectral estimator is computationally straightforward. 
\end{abstract}


\maketitle

\section{Introduction}
Consider a grouped panel data model
\begin{equation}\label{eq: model}
y_{it} = x_{it}'\beta + \alpha_{g_i t} + v_{it},\quad\text{for all }i=1,\dots,N, \ t = 1,\dots,T,
\end{equation}
where $i$ denotes cross-sectional units, $t$ denotes time periods, $y_{it}\in\mathbb R$ is an observable dependent variable, $x_{it} = (x_{it1},\dots,x_{itd})'\in\mathbb R^d$ is a corresponding vector of observable covariates, $g_i\in\{1,\dots,G\}$ is an unobservable group-membership variable, $v_{i t}\in\mathbb R$ is an unobservable zero-mean noise random variable, $\beta = (\beta_1,\dots,\beta_d)'\in\mathbb R^d$ is a vector of parameters of interest, and $(\alpha_{1 t},\dots,\alpha_{G t})' \in\mathbb R^G$ is a vector of unobservable group-specific time effects. Here, we assume that the noise and covariates are uncorrelated, 
\begin{equation}\label{eq: mean independence}
\mathbb E[v_{it}x_{it}] = 0_d,\quad\text{for all }i=1,\dots,N, \ t = 1,\dots,T,
\end{equation}
where $0_d = (0,\dots,0)'\in\mathbb R^d$, but group-specific time effects and group-membership variables can be arbitrarily correlated with covariates. Also, throughout the paper, we assume that the number of groups $G$ is known (consistently estimated).

The model \eqref{eq: model} was originally introduced by Bonhomme and Manresa in \cite{BM15}, who also developed a so-called grouped fixed-effect estimator of the vector of parameters $\beta$ in this model.
This estimator has attractive theoretical properties but is computationally difficult. It is therefore of interest to see if there exist alternative estimators that would be easier to compute. In this paper, we answer this question affirmatively, under certain additional assumptions to be specified in the next paragraph, and propose an estimator of $\beta$, which we call the {\em post-spectral estimator}, that also has nice theoretical properties but, in contrast to the grouped fixed-effect estimator, is computationally simple.

Like in the previous papers on grouped panel data models, we consider large $(N,T)$-asymptotics, i.e. we assume that $T\to\infty$, potentially very slowly, as $N\to\infty$, since otherwise $\beta$ is in general not identified. In contrast to the previous papers, however, we impose a special structure on the data-generating process for the covariates $x_{it}$. In particular, we assume that for some $M\geq 1$,
\begin{equation}\label{eq: x process general}
x_{it} = \sum_{m=1}^{M} \rho_{i m}\alpha_{g_it}^{m} + z_{it},\quad\text{for all }i=1,\dots,N, \ t= 1,\dots,T,
\end{equation}
where $(\alpha_{1t}^m,\dots,\alpha_{Gt}^m)'\in\mathbb R^G$ for $m=1,\dots,M$ are group-specific time effects, $\rho_{im}\in \mathbb R^d$ for $m=1,\dots,M$ are individual-specific vectors of coefficients, and $z_{it}$ is a zero-mean component of $x_{it}$ that is independent of group-specific time effects, group-membership variable $g_{i}$, and vectors of coefficients $\rho_{i1},\dots\rho_{iM}$. Here, no quantity on the right-hand side of \eqref{eq: x process general} is observed, except for the number of time effects $M$, which we assume to be known (consistently estimated). Also, without loss of generality, we assume that $\alpha_{\gamma t}^1 = \alpha_{\gamma t}$ for all $\gamma = 1,\dots,G$ and $t = 1,\dots,T$. We believe that this factor-analytic model for the covariates $x_{it}$ is rather flexible as it allows for individual-specific correlations between covariates and group-specific time effects.
In Appendix \ref{sec: motivating example}, we also provide an example in terms of agricultural production functions and environmental economics to motivate equation (\ref{eq: x process general}).

Our post-spectral estimator consists of three steps. In the first step, we carry out preliminary estimation of $\beta$. To do so, we prove that as long as the data-generating process is given by equations \eqref{eq: model} and \eqref{eq: x process general}, there exists a convex quadratic function $f\colon \mathbb R^d\to\mathbb R$ such that (i) its unique minimum is achieved by $\beta$ and (ii) for each value of $b\in\mathbb R^d$, the value of $f(b)$ can be consistently estimated by the sum of $2G M+2$ largest in absolute value eigenvalues of a certain matrix. We then demonstrate that this function and its consistent estimator can be used to construct an estimator of $\beta$ that is both consistent and computationally simple. This estimator, which we call the {\em spectral estimator}, may have slow rate of convergence if $T$ is growing rather slowly, and so we proceed to the second and the third steps. In the second step, using the preliminary spectral estimator of $\beta$ obtained in the first step, we carry out classification of units $i=1,\dots,N$ into groups $\gamma = 1,\dots,G$. Importantly, our classification algorithm, which is a version of spectral clustering method (\cite{VW04, M18, V18, LZZ21}), is fast and does not require solving any non-convex optimization problems. 
In the third step, we obtain the {\em post-spectral estimator} of $\beta$ by performing OLS-type estimation pooling all units within the same group together. 

We prove that our post-spectral estimator is generally consistent and has particularly attractive properties under the assumption of well-separated groups, which means that the vectors $(\alpha_{\gamma 1},\dots,\alpha_{\gamma T})'$, $\gamma = 1,\dots,G$, are not too close to each other, and which was also used in \cite{BM15}.\footnote{We make no assumptions on the distance between vectors $(\alpha_{\gamma 1}^m,\dots,\alpha_{\gamma T}^m)'$ for $m\geq 2$.} Specifically, we show that under this assumption, the classifier constructed in the second step is consistent in the sense that with probability approaching one, any two units are getting classified into the same group if and only if they belong to the same group, and so the post-spectral estimator of $\beta$ is asymptotically equivalent to the pooled-OLS estimator with known group memberships (i.e., the oracle estimator). In turn, the latter is $\sqrt{NT}$ consistent and admits the standard OLS inference. Under the assumption of well-separated groups, our post-spectral estimator thus can be used for testing hypotheses and for constructing confidence intervals for $\beta$ using standard panel-data methods, ignoring the preliminary estimation and classification steps. Inference without the assumption of well-separated groups, however, remains an open (and challenging) question for future work. Theoretical properties of our post-spectral estimator are thus comparable to those of the grouped fixed-effect estimator, with the caveat that we impose the special structure on the data-generating process for the covariates $x_{it}$ given in \eqref{eq: x process general}.


To compare computational properties of the post-spectral estimator to those of the grouped fixed-effect estimator, we note that the latter, as well as many related estimators (\cite{AB16,AB17,CLW19,CSS19,ZWZ17}), are built around the K-means optimization problem. This optimization problem is known to be NP-hard (see \cite{ADHP09}), and so it is rather unlikely that there exists a fast algorithm for finding its solutions. Any proposed fast implementation of the aforementioned estimators therefore is likely to fail occasionally. For example, the grouped fixed-effect estimator is defined as the (global) minimizer of the sum of squared residuals over all parameter values and over all partitions of units into groups, and the main algorithm to calculate this minimizer in \cite{BM15} proceeds by initializing randomly selected values of parameters $\beta$ and $\{\alpha_{\gamma t}\}_{\gamma,t=1}^{G,T}$ and then alternating between two steps: (1) optimization over the values of group memberships $g_1,\dots,g_N$ given the values of parameters $\beta$ and $\{\alpha_{\gamma t}\}_{\gamma,t=1}^{G,T}$, and (2) optimization over the values of parameters $\beta$ and $\{\alpha_{\gamma t}\}_{\gamma,t=1}^{G,T}$ given the values of group memberships $g_1,\dots,g_N$. Since this procedure at best converges to a local minimum, it is repeated over many different initial values of parameters to find the global minimum, which corresponds to the grouped fixed-effect estimator. As \cite{BM15} notes, however, ``a prohibitive number of initial values may be needed to obtain reliable solutions.'' In addition, it seems never possible to say whether the global minimum has been found, as this in general would require minimizing the sum of squared residuals over parameter values for {\em each} partition of units into groups, and the number of these partitions, $G^N$, is tremendously large even in small samples. \cite{BM15} also proposed a few other algorithms to calculate the grouped fixed-effect estimator that tend to perform better in simulations but they are all subject to the same critique: if they are fast, they must fail occasionally. In contrast, our post-spectral estimator is easy to compute and does not suffer from the potential failure problem.

There is also growing literature on non-linear panel data models with group structure of individual-level parameters (\cite{HM10, LN12, SSP16, BML17, SJ18, WPS18, GV19, GS20}) originated by Hahn and Moon in \cite{HM10}. This literature is conceptually related to the grouped panel data model \eqref{eq: model} but estimation techniques developed in this literature are very different from those considered here because all aforementioned papers assume that the individual-level parameters are time-independent, and so preliminary consistent estimation of these parameters is possible by performing estimation separately for each unit. The latter is not possible in the model \eqref{eq: model} because individual effects $\alpha_{g_it}$ are varying over time, which creates one of the key challenges in estimating this model. 


We note also that the grouped panel data model \eqref{eq: model} is a special case of a panel data model with interactive fixed effects, corresponding to factor loadings with finite support in the latter model. The methods developed for estimating panel data models with interactive fixed effects can therefore be used to estimate $\beta$ in \eqref{eq: model} as well. To the best of our knowledge, however, most of these methods are either computationally difficult or require conditions that are substantially different from those used in our paper. For example, the estimator in \cite{B09} is based on a solution to a non-convex optimization problem, and the estimators in \cite{HNR88, ALS01, ALS13} require certain IV-type conditions. Like in our approach, the estimators in \cite{P06} also require restricting the data-generating process for the covariates $x_{it}$ but the nature of imposed restrictions is very different. In particular, \cite{P06} either imposes a certain rank condition, which can only be satisfied if the dimensionality of $x_{it}$ is sufficiently large, or requires the factor loadings in the equation for covariates to be independent of the factor loadings in the equation for the dependent variable, which in our model would correspond to assuming that $g_i$ in \eqref{eq: x process general} is different and independent of $g_i$ in \eqref{eq: model},\footnote{\cite{P06} claims that his estimators are consistent without requiring independence of the factors but a counter-example is given in \cite{WU13}, which proves that if the rank condition is not satisfied, then independence is essentially a necessary condition for consistency of the estimators in \cite{P06}.} thus leading to a random, rather than fixed, effect model. In fact, the only exception in this literature we are aware of is \cite{MW19}, who developed a computationally relatively simple method and used conditions that are similar to (actually somewhat weaker than) those used in our paper. Their estimator can potentially replace our preliminary spectral estimator. However, the convergence rate of their estimator is only $(T\wedge N)^{-1/2}$, which can be particularly slow if $T$ is growing much slower than $N$, a case of special interest in the grouped panel data model. In contrast, the rate of convergence of our spectral estimator is $(T\wedge N)^{-1}$, which seems much more acceptable for the preliminary estimation. Indeed, we find via simulations reported below that our spectral estimator leads to much better results in reasonably large samples. Finally, \cite{AWZ22} has recently developed a method for debiasing estimators with the slow convergence rate $(T\wedge N)^{-1/2}$ yielding estimators with the fast convergence rate $(T\wedge N)^{-1}$. An advantage of our estimator here is that we obtain the fast rate in just one step, instead of two steps, which again may be preferable when $T$ is relatively small, so that the original estimator to be debiased is rather imprecise and is hard to debias.

In addition, we study three extensions of the model \eqref{eq: model}. First, we consider a dynamic version of the model, where lagged values of $y_{it}$ appear on the right-hand side of \eqref{eq: model}. We demonstrate that our spectral and post-spectral estimators work for this model too, as long as the number of factors $M$ is appropriately modified. We allow for both pre-determined and exogenous covariates in this model. Second, we consider a high-dimensional version of the model \eqref{eq: model}, where the number of covariates $d$ is large, potentially much larger than $NT$, but the vector of coefficients $\beta$ is sparse in the sense that it has relatively few non-zero components. We demonstrate how to modify the spectral estimator via $\ell^1$-penalization to obtain a computationally simple and consistent estimator of $\beta$ in this case. Third, we consider an interactive fixed-effect panel data model, where $\alpha_{g_it}$ in \eqref{eq: model} is replaced by $\kappa_i'\phi_t$, with $\phi_t$ being a vector of factors and $\kappa_i$ being a vector of factor loadings. We demonstrate that the spectral estimator, with appropriately modified parameters $G$ and $M$, is consistent in this model with the convergence rate $(T\wedge N)^{-1}$. Being computationally simple, our spectral estimator thus can serve as an alternative to existing estimators in the literature on interactive fixed-effect panel data models. Note, however, that in all three extensions, we maintain a version of \eqref{eq: x process general}.

The rest of the paper is organized as follows. In the next section, we discuss details of implementation of our spectral and post-spectral estimators. In Section \ref{sec: asymptotic theory}, we state their asymptotic properties. In Section \ref{sec: extensions}, we provide the extensions of the baseline model. In Section \ref{sec: monte carlo}, we discuss results of a small-scale Monte Carlo simulation study that shed some light on finite-sample properties of our estimators. In Section \ref{sec: remaining proofs}, we present main proofs. In Appendix \ref{sec: technical lemmas}, we collect some technical lemmas that are useful for the proofs of our main results. In Appendix \ref{sec: proofs for extensions}, we present remaining proofs. In Appendix \ref{sec: randomized algorithm}, we describe a method for calculating eigenvalues of large matrices, which may be needed for implementing our estimators. In Appendix \ref{sec: additional details on dynamic model}, we provide some details on the assumptions of the dynamic model extension. In Appendix \ref{sec: motivating example}, we discuss an example motivating equation \eqref{eq: x process general}.

 
 
 

\section{Estimation}\label{sec: estimator}
Our proposed estimation procedure consists of three steps. The first step is preliminary consistent estimation of $\beta$, which is based on the spectral analysis of certain matrices and gives the spectral estimator. The second step is classification of units into groups. The third step is pooled-OLS estimation of $\beta$ on classified units, which gives the post-spectral estimator. Under the assumption of well-separated groups, the post-spectral estimator will be $\sqrt{NT}$-consistent and asymptotically normal with mean zero, making inference based on this estimator straightforward.

\subsection{Spectral Estimator}\label{sub: preliminary estimation definition}

For all $b\in\mathbb R^d$, let $A^b$ be an $N\times N$ matrix whose $(i,j)$-th element is
\begin{equation}\label{eq: matrix a}
A^b_{ij} = \frac{1}{NT}\sum_{t=1}^T\Big\{(y_{it} - x_{it}'b) - (y_{jt} - x_{jt}'b) \Big\}^2,\quad \text{for all }i,j = 1,\dots,N,
\end{equation}
Since $A^b$ is an $N\times N$ symmetric matrix, it has $N$ real eigenvalues. Let $\lambda^b_1,\dots,\lambda^b_{2GM+2}$ be its $2G M+2$ largest in absolute value eigenvalues. We will show below that under mild conditions,
\begin{equation}\label{eq: prelim prob limit}
\lambda^b_1 + \dots + \lambda^b_{2GM+2} = b'\Sigma b + S'b + L + o_P(1),\quad\text{for all }b\in\mathbb R^d,
\end{equation}
where $\Sigma$ is a $d\times d$ symmetric positive definite matrix, $S$ is a $d\times 1$ vector, $L$ is a scalar, and, importantly, $\beta$ is the unique minimizer of the function $b\mapsto f(b) = b'\Sigma b + S'b + L$. We therefore define our spectral estimator as
\begin{equation}\label{eq: spectral estimator definition}
\tilde\beta = (\tilde\beta_1,\dots,\tilde\beta_d)' = \arg\min_{b\in\mathbb R^d} \Big\{b'\hat\Sigma b + \hat S'b + \hat L\Big\},
\end{equation}
where $\hat\Sigma$, $\hat S$, and $\hat L$ are estimators of $\Sigma$, $S$, and $L$, respectively, to be constructed below. By the first-order conditions, the estimator $\tilde\beta$ can be equivalently defined as
$$
\tilde\beta = -\hat\Sigma^{-1}\hat S/2.
$$
We will prove in the next section that $\tilde\beta\to_P \beta$.

Next, we discuss estimators $\hat\Sigma$, $\hat S$, and $\hat L$. For brevity of notation, denote
\begin{equation}\label{eq: f function}
\hat f(b) = \lambda^b_1 + \dots + \lambda^b_{2GM+2},\quad\text{for all }b\in\mathbb R^d,
\end{equation}
so that by \eqref{eq: prelim prob limit},
\begin{equation}\label{eq: prelim 2 prob limit}
\hat f(b) = b'\Sigma b + S'b + L + o_P(1),\quad\text{for all }b\in\mathbb R^d.
\end{equation}
Also, for all $k=1,\dots,d$, let $e_k = (0,\dots,0,1,0,\dots,0)'$ be the $d\times 1$ vector with 1 in the $k$-th position and 0 in all other positions, and let $0_d = (0,\dots,0)'$ be the $d\times 1$ vector with 0 in all positions. Since \eqref{eq: prelim prob limit} implies $\hat f(0_d) = L + o_P(1)$, we set
$
\hat L = \hat f(0_d).
$
Further, since $\hat f(e_k) - \hat f(-e_k) \to_P 2S_k$, we set
$$
\hat S_k = \frac{\hat f(e_k) - \hat f(-e_k)}{2},\quad\text{for all }k=1,\dots,d,
$$
and $\hat S = (\hat S_1,\dots,\hat S_d)'$. Finally, since $\hat f(e_k) + \hat f(-e_k) = 2(\Sigma_{k k} + L) + o_P(1)$, we set
$$
\hat \Sigma_{k k} = \frac{\hat f(e_k) + \hat f(-e_k)}{2} - \hat L,\quad\text{for all }k=1,\dots,d,
$$
and since $\hat f(e_{k} + e_{l}) = \Sigma_{k k} + \Sigma_{l l} + 2\Sigma_{k l} + S_{k} + S_{l} + L + o_p(1)$, we set
$$
\hat \Sigma_{k l} = \hat \Sigma_{l k} = \frac{\hat f(e_{k} + e_{l}) - \hat\Sigma_{k k} - \hat\Sigma_{l l} - \hat S_{k} - \hat S_{l} - \hat L}{2},
$$
for all $k,l=1,\dots,d, \ k> l$, and let $\hat \Sigma$ be the matrix whose $(k,l)$-th component is $\hat \Sigma_{k l}$.\footnote{Note that the quality of the estimators $\hat\Sigma$, $\hat S$, and $\hat L$ could potentially be improved by exploiting additional values of the vector $b$ but we leave the question of optimal estimation for future work.}

Under result \eqref{eq: prelim prob limit}, the estimators $\hat\Sigma$ and $\hat S$ are consistent for $\Sigma$ and $S$, respectively, and so $\tilde\beta = -\hat\Sigma^{-1}\hat S/2 \to_P -\Sigma^{-1}S/2 = \beta$, as long as $\Sigma$ is invertible, which is the case under mild conditions. The bulk of our derivations in the next section will thus be related to proving \eqref{eq: prelim prob limit}.

Before we move on, however, we note that the $N\times N$ matrices $A^b$ may be rather large, and the reader might wonder how much time it takes to calculate their eigenvalues $\lambda^b_1,\dots,\lambda^b_{2GM+2}$. Fortunately, there exists a class of fast randomized algorithms that allow to calculate these eigenvalues arbitrarily well; see Appendix \ref{sec: randomized algorithm} for details.

\begin{remark}[Alternative Version of Spectral Estimator]
Given that we have \eqref{eq: prelim prob limit} and that $\beta$ is the unique minimizer of the function $b\mapsto b'\Sigma b + S'b + L$, it seems natural to consider
$$
\check\beta = \arg\min_{b\in\mathbb R^d}(\lambda^b_1 + \dots + \lambda^b_{2GM+2})
$$
as an alternative to the spectral estimator $\tilde\beta$ appearing in \eqref{eq: spectral estimator definition}. The minimization problem here, however, is not necessarily convex, even though the criterion function is asymptotically convex. Computing $\check\beta$ may therefore be difficult. In contrast, our spectral estimator $\tilde\beta$ circumvents this problem by employing the parametric structure of the limit of this criterion function.
\qed
\end{remark}
\begin{remark}[Tuning Parameters for Spectral Estimator]
Implementing the spectral estimator $\tilde\beta$ requires choosing the product $GM$ but does not require knowing $G$ and $M$ separately, which means that we only need one tuning parameter instead of two of them. In addition, the proof of Theorem \ref{thm: consistency spectral estimator} below reveals that consistency of the spectral estimator holds even if we replace $GM$ in the definition of the estimator $\tilde\beta$ by any number that is bigger than $GM$ (as long as it is independent of $N$ and $T$). Thus, to implement the spectral estimator, we actually only need an upper bound on the product $GM$. Moreover, the proof of Theorem \ref{thm: consistency spectral estimator} also shows that for any vector $b\in\mathbb R^d$, the matrix $A^b$ has at most $2GM+2$ eigenvalues that are not asymptotically vanishing. This suggests a method to estimate the product $GM$ by counting the number of eigenvalues of the matrix $A^b$ exceeding certain threshold, which is chosen to slowly converge to zero. This method can underestimate the product $GM$, which happens if the matrix $A^b$ actually has fewer than $2GM+2$ eigenvalues that are not asymptotically vanishing, but whenever this happens, we can lose only asymptotically vanishing eigenvalues in the sum \eqref{eq: f function}, which can not break consistency of the spectral estimator. For brevity of the paper, however, we leave the question of formally deriving results with an estimated product $GM$ to future work.
\qed
\end{remark}


\subsection{Classifier}\label{sub: classification}
To classify units into groups, we will use a version of the spectral clustering method.\footnote{Note that the special structure on the data-generating process for the covariates $x_{it}$ given in \eqref{eq: x process general} was used for the construction of the spectral estimator only. This structure has no role for our classifier and for the post-spectral estimator described below.} For reasons to be explained in Remark \ref{rem: sample splitting} in the next section, we will also rely on sample splitting. To this end, let $h_1,\dots,h_N$ be i.i.d. random variables that are independent of the data and that are taking values $0$ and $1$, each with probability 1/2. We split all cross-sectional units $i=1,\dots,N$ into two subsamples, $\mathcal I_0 = \{i =1,\dots,N\colon h_i=1\}$ and $\mathcal I_1 = \{i =1,\dots,N\colon h_i=0\}$. Further, for $i=1,\dots,N$, denote $y_i = (y_{i1},\dots,y_{iT})'$ and $x_i = (x_{i1},\dots,x_{iT})'$. Also, for $h=0,1$, let $\tilde\beta^h$ be the spectral estimator calculated using the subsample $\mathcal I_h$ and let $\hat B^h$ be a $T\times T$ matrix given by
\begin{equation}\label{eq: b matrix definition}
\hat B^h = \frac{2}{NT}\sum_{i\in \mathcal I_h} (y_{i} - x_{i}\tilde\beta^h)(y_{i} - x_{i}'\tilde\beta^h)'.
\end{equation}
Since $\hat B^h$ is a $T\times T$ symmetric positive definite matrix, it has $T$ non-negative eigenvalues and $T$ corresponding orthonormal eigenvectors. Let $\hat F_h$ be a $T\times G$ matrix whose columns are orthonormal eigenvectors corresponding to $G$ largest eigenvalues of the matrix $\hat B^h$. Moreover, for all $i=1,\dots,N$, let $\hat A_i$ be a $T\times1$ vector defined by
\begin{equation}\label{eq: p vector definition}
\hat A_{i} = \hat F_{h_i}\hat F_{h_i}'(y_{i} - x_{i}\tilde\beta^{h_i}).\footnote{Thus, for all units $i$ with $h_i=1$, we calculate $\hat A_i$ using $\hat F_1$ and $\tilde\beta_1$, which are obtained from the subsample $\mathcal I_1$ consisting of all units $j$ with $h_j=0$ and, vice versa, for all units $i$ with $h_i=0$, we calculate $\hat A_i$ using $\hat F_0$ and $\tilde\beta_0$, which are obtained from the subsample $\mathcal I_0$ consisting of all units $j$ with $h_j=1$.}
\end{equation}
Intuitively, the vectors $\hat A_1,\dots,\hat A_N$ estimate the vectors $\alpha_{g_1},\dots,\alpha_{g_N}$, where we denoted $\alpha_{\gamma} = (\alpha_{\gamma 1},\dots,\alpha_{\gamma T})'$ for all $\gamma=1,\dots,G$.
We therefore classify units $i=1,\dots,N$ into $G$ groups using these vectors. To do so, fix a tuning parameter $\lambda > 0$, to be chosen below, and consider the following algorithm:

\medskip
\noindent
{\bf Classification Algorithm.}

{\em Step 1:} set $\mathcal A_1 = \{1\}$, $m = 1$, and $i = 1$;

{\em Step 2:} replace $i$ by $i + 1$;

{\em Step 3:} if $i = N + 1$, stop the algorithm;

{\em Step 4:} set $\mathcal C_i = \{\gamma=1,\dots,m\colon \|\hat A_i - |\mathcal A_{\gamma}|^{-1}\sum_{l\in\mathcal A_{\gamma}}\hat A_l\| \leq \lambda\}$;

{\em Step 5:} if $\mathcal C_i$ is empty, replace $m$ by $m+1$, set $\mathcal A_m = \{i\}$, and go to Step 2;

{\em Step 6:} if $\mathcal C_i$ is not empty, replace $\mathcal A_{\gamma}$ by $\mathcal A_{\gamma}\cup \{i\}$ for $\gamma=\min\mathcal C_i$ and go to Step 2.

\medskip
\noindent

This algorithm creates $m$ groups $\mathcal A_1,\dots,\mathcal A_m$, with the number of groups $m$ depending on $\lambda$, so that $m = m(\lambda)$. Clearly, $\lambda\mapsto m(\lambda)$ is a right-continuous  function, and so
$$
\hat\lambda = \min\Big\{\lambda>0\colon m(\lambda)\leq G\Big\}
$$
is well-defined. (In practice, $\hat\lambda$ can be calculated using the values of $m(\lambda)$ on a fine grid.) We classify units $i=1,\dots,N$ into $G$ groups using this algorithm with $\lambda = \hat\lambda$. The result of the algorithm is then $m(\hat\lambda)\leq G$ groups $\mathcal A_1,\dots,\mathcal A_{m(\hat\lambda)}$, and for all $i=1,\dots,N$, there exists a unique $\gamma = \gamma(i)\in\{1,\dots,m(\hat\lambda)\}$ such that $i\in\mathcal A_{\gamma}$. We set
$$
\hat g_i = \gamma(i),\quad\text{for all }i=1,\dots,N,
$$ 
and $\hat g = (\hat g_1,\dots,\hat g_N)'$. Note that this classifier can occasionally lead to less than $G$ groups, which happens when $m(\hat\lambda) < G$, but we will show in the next section that it is consistent in the sense that
\begin{equation}\label{eq: classifier consistency statement}
\mathrm P\Big(\text{for all }i,j=1,\dots,N, \ \  \hat g_{i} = \hat g_{j}\text{ if and only if }g_{i} = g_{j}\Big)\to 1,
\end{equation}
under the assumption of well-separated groups.


\begin{remark}[Covariate-Based Classifiers]
Recall that we assume the group structure in the data-generating process for covariates $x_{it}$, equation \eqref{eq: x process general}. In principle, this structure could be used to classify units into groups as well. This seemingly sensible alternative to our procedures is interesting because it does not require estimating $\beta$ on the first step, and so looks much easier than our procedures. However, a substantial drawback of this procedure is that it may not be consistent if the group structure in the data-generating process for $x_{it}$'s is coarser than the group-structure in the data-generating process for $y_{it}$'s. For example, suppose that $M=2$ and $\rho_{i1}=0$ for all $i=1,\dots,N$. Then equation \eqref{eq: x process general} becomes
$$
x_{it}=\rho_{i2}\alpha_{g_it}^2 + z_{it},\quad \text{for all }i=1,\dots,N, \ t=1,\dots,T.
$$
Now, if we assume that $G=3$ but $\alpha_{1t}^2=\alpha_{2t}^2\neq \alpha_{3t}^2$, there are effectively only two groups in the data-generating process for $x_{it}$'s. Therefore, any reasonable classification based on this equation would merge groups 1 and 2, which would make \eqref{eq: classifier consistency statement} impossible. 
\qed
\end{remark}


\subsection{Post-Spectral Estimator}
Once we have classified units into groups, estimation of $\beta$ is straightforward. In particular, we rely upon a pooled-OLS estimator:
\begin{equation}\label{eq: pooled ols}
(\hat\beta,\hat \alpha) = \arg\min_{b\in\mathcal B,a\in\mathcal A_{G,T}}\sum_{i=1}^N\sum_{t=1}^T\left(y_{it} - x_{it}'b - a_{\hat g_i t}\right)^2,
\end{equation}
where $\mathcal B$ is a parameter space for the vector $\beta$, and $\mathcal A_{G,T}$ is a parameter space for the matrix $\{\alpha_{\gamma t}\}_{\gamma,t=1}^{G,T}$. We refer to $\hat\beta$ as the post-spectral estimator for grouped panel data models. We will show in the next section that under the assumption of well-separated groups, this estimator is asymptotically equivalent to the estimator based on correct classification,
\begin{equation}\label{eq: correct specification}
(\hat\beta^0,\hat \alpha^0) = \arg\min_{b\in\mathcal B,a\in\mathcal A_{G,T}}\sum_{i=1}^N\sum_{t=1}^T\left( y_{it} -  x_{it}'b - a_{g_it}\right)^2,
\end{equation}
and thus to the grouped fixed-effect estimator of \cite{BM15}. Hence, the standard OLS inference ignoring group classification applies.
\begin{remark}[Estimating $\beta$ by OLS of $y_{it}$ on estimated $z_{it}$]
Observe that our data-generating process for covariates $x_{it}$ in \eqref{eq: x process general} is given by a factor-analytic model; namely, it can be written as
$$
x_{it} = \sum_{m=1}^M\sum_{\gamma=1}^G \rho_{im}1\{g_i = \gamma\}\alpha_{\gamma t}^m + z_{it} =  \omega_i'\phi_t + z_{it},\quad\text{for all } i=1,\dots,N, \ t=1,\dots,T,
$$
where $\phi_t = (\alpha_{1t}^1,\dots,\alpha_{Gt}^M)'$ is a $GM\times 1$ vector of factors and $\omega_i = (\rho_{i1}1\{g_i=1\},\dots,\rho_{iM}1\{g_i = G\})'$ is the $GM\times d$ matrix of factor loadings. Here, factors $\phi_t$ and the factor loadings $\omega_i$ can be estimated by the method of asymptotic principle components as in Section 3 of \cite{BN02}; see also \cite{CK86,CK88,SW99,FHLR00}. Denoting these estimators $\hat\phi_t$ and $\hat\omega_i$ and letting $\hat z_{it} = x_{it} - \hat\omega_i'\hat\phi_t$, we are then able to obtain an estimator of $\beta$ by simply running OLS of $y_{it}$ or $\hat z_{it}$. This estimator is easy to compute and is consistent under weak conditions since $z_{it}$ is uncorrelated with both $\alpha_{g_it}$ and $x_{it} - z_{it}$. However, it performs poorly in the case of weak factors, i.e. when the factor loadings $\omega_i$ are close to zero, as factors $\phi_t$ can not be consistently estimated in this case; see \cite{CPT11, O12} for details. In particular, our simulation experience confirms that the post-spectral estimator substantially outperforms this simple estimator in the case of weak factors.
\end{remark}
\begin{remark}[Tuning Parameters for Post-Spectral Estimator]
Implementing the post-spectral estimator $\hat\beta$ requires choosing the number of groups $G$ but does not require to specify the number of time effects $M$ in the equation for covariates. Thus, like in the case of the spectral estimator, we only need one tuning parameter instead of two to implement the post-spectral estimator. In turn, estimating the number of groups $G$ is relatively easy. In particular, we can employ penalization techniques as developed in \cite{BN02}, in the same fashion as discussed in \cite{BM15}. However, for brevity of the paper, we leave the question of formally deriving results with an estimated number of groups $G$ to future work.
\qed
\end{remark}

\section{Asymptotic Theory}\label{sec: asymptotic theory}
In this section, we derive asymptotic properties of the procedures described above. For convenience, we do so in three separate subsections: spectral estimator, classifier, and post-spectral estimator. 

Throughout the rest of the paper, we assume that membership variables $g_i$, group-specific time effects $\alpha_{\gamma t}^m$ and individual specific vectors of coefficients $\rho_{im}$ are non-stochastic, i.e. our analysis is conditional on these random quantities. Also, given that we set $\alpha_{\gamma t}^1 = \alpha_{\gamma t}$, group-specific time effects $\alpha_{\gamma t}$ are non-stochastic as well. Moreover, we assume that the units $i$ are independent.

\subsection{Spectral Estimator}

Let $\mathcal S^T$ denote the unit sphere in $\mathbb R^T$, i.e. $\mathcal S^T = \{u\in\mathbb R^T\colon \|u\|=1\}$. Also, for any random variable $w$, let $\|w\|_{\psi_2}$ denote the sub-Gaussian norm of $w$, i.e.
$$
\|w\|_{\psi_2} = \inf\Big\{\epsilon > 0\colon \mathbb E[\exp(w^2/\epsilon^2)]\leq 2\Big\};
$$
see Section 2.5.2 in \cite{V18} on properties of the sub-Gaussian norm.\footnote{Sub-Gaussian norm also often appears in the literature under the name of Orlicz norm.} Intuitively, a random variable has a finite sub-Gaussian norm if the tails of its distribution are not heavier than tails of the Gaussian distribution. For example, every bounded random variable has a finite sub-Gaussian norm. To prove consistency and to derive the rate of convergence of the spectral estimator $\tilde\beta$, we will use the following assumptions.

\begin{assumption}\label{as: spectral v}
(i) For some constant $C_1>0$, we have $\|\sum_{t=1}^T u_t v_{it}\|_{\psi_2}\leq C_1$ for all $i=1,\dots,N$ and $u = (u_1,\dots,u_T)'\in\mathcal S^T$. (ii) In addition, for some constant $C_2>0$, we have $\|\sum_{t=1}^T u_t z_{itk}\|_{\psi_2}\leq C_2$ for all $i=1,\dots,N$, $u=(u_1,\dots,u_T)'\in\mathcal S^T$, and $k=1,\dots,d$.
\end{assumption}
By Hoeffding's inequality (Proposition 2.6.1 in \cite{V18}), Assumption \ref{as: spectral v}(i) holds if the random variables $v_{it}$ have finite sub-Gaussian norm and are independent across $t$. More generally, due to numerous versions of Hoeffding's inequality for time series data (e.g., see \cite{D09,V02}), Assumption \ref{as: spectral v}(i) holds as long as the dependence of the random variables $v_{it}$ across $t$ is not too strong. Assumption \ref{as: spectral v}(ii) is similar to Assumption \ref{as: spectral v}(i) but imposes the integrability and time series dependence restrictions on $z_{it}$ instead of $v_{it}$. We admit that the assumption of random variables having finite sub-Gaussian norm may be somewhat strong but we emphasize that a version of Theorem \ref{thm: consistency spectral estimator} below, with slower rates, can be derived under weaker integrability assumptions. We have chosen to work with Assumption \ref{as: spectral v} in order to minimize technicalities of our analysis.

\begin{assumption}\label{as: spectral vz}
(i) We have $\|(NT)^{-1}\sum_{i=1}^N\sum_{t=1}^T v_{it}z_{it}\| = O_P(1/\sqrt{NT})$. (ii) In addition, $(NT)^{-1}\sum_{i=1}^N\sum_{t=1}^T z_{it}z_{it}' = \Sigma/2 + O_P(1/\sqrt{NT})$, where $\Sigma$ is a positive definite $d\times d$ matrix.
\end{assumption}

Since $\mathbb E[v_{it}] = 0$, $\mathbb E[v_{it} x_{it}] = 0_d$, and we assume that $\alpha_{g_it}^m$ and $\rho_{im}$ are non-stochastic, it follows from \eqref{eq: x process general} that $\mathbb E[v_{it}z_{it}] = 0_d$. Hence, Assumption \ref{as: spectral vz}(i) is a quantitative law of large numbers for the random vector $(NT)^{-1}\sum_{i=1}^N\sum_{t=1}^T v_{it}z_{it}$. Similarly, Assumption \ref{as: spectral vz}(ii) is a quantitative law of large numbers for the random matrix $(NT)^{-1}\sum_{i=1}^N\sum_{t=1}^T z_{it}z_{it}'$. Assumption \ref{as: spectral vz}(ii) also imposes the constraint that the probability limit of this matrix is positive-definite, which is an identification condition.

\begin{theorem}[Rate of Convergence of Spectral Estimator $\tilde\beta$]\label{thm: consistency spectral estimator}
Under Assumptions \ref{as: spectral v} and \ref{as: spectral vz},
\begin{equation}\label{eq: convergence rate spectral estimator}
\tilde\beta = \beta + O_P\left(\frac{1}{T\wedge N}\right).
\end{equation}
\end{theorem}

\begin{remark}[Relaxing Data-Generating Process for Covariates]
Inspecting the proof of Theorem \ref{thm: consistency spectral estimator} reveals that the theorem continues to hold even if we allow for a substantially larger class of data-generating processes for covariates instead of that specified in \eqref{eq: x process general}. Indeed, if we simply assume that $x_{it} =\varsigma_{it} + z_{it}$ for all $i=1,\dots,N$ and $t =1,\dots,T$ and some $N\times T$ matrix $\varsigma$ of rank $M$, then Theorem \ref{thm: consistency spectral estimator} holds as long as the spectral estimator $\hat\beta$ uses $2(G+M+1)$ instead of $2GM+2$ eigenvalues of the matrices $A^b$. Throughout the paper, however, we prefer to work with \eqref{eq: x process general} as this seems to be the most natural assumption on the data-generating process for covariates.\footnote{The representation $x_{it} = \varsigma_{it} + z_{it}$, however, emphasizes the fact that our methods are able to deal with the case when the equation for covariates $x_{it}$ have more groups than the equation for the dependent variable $y_{it}$.}
\qed
\end{remark}

\subsection{Classifier}
For all $\gamma = 1,\dots,G$, let $N_\gamma = 1\{g_i = \gamma\}$ be the number of units $i$ within group $\gamma$. To prove consistency of the classifier $\hat g$, we will use the following assumptions.

\begin{assumption}\label{as: bounded rho}
(i) For some constant $C_3>0$, we have $\|\rho_{im}\|\leq C_3$ for all $m=1,\dots,M$ and $i=1,\dots,N$. (ii) In addition, for some constant $C_4>0$, we have $|\alpha^m_{\gamma t}|\leq C_4$ for all $m=1,\dots,M$, $\gamma=1,\dots,G$, and $t=1,\dots,T$.
\end{assumption}


\begin{assumption}\label{as: well separated groups}
For some constant $c_1>0$, we have $T^{-1}\sum_{t=1}^T(\alpha_{\gamma_1 t} - \alpha_{\gamma_2 t})^2\geq c_1$ for all $\gamma_1,\gamma_2=1,\dots,G$ such that $\gamma_1\neq\gamma_2$.
\end{assumption}

Assumption \ref{as: bounded rho} is self-explanatory. Assumption \ref{as: well separated groups} means that the groups are well-separated in the sense that the vectors of group-specific time effects, $(\alpha_{\gamma 1},\dots,\alpha_{\gamma T})'$ for $\gamma = 1,\dots,G$, are not too close to each other. We will use this assumption to prove consistency of the classifier $\hat g$ and to derive asymptotic normality of the post-spectral estimator $\hat\beta$ but we will not use it to prove consistency of the post-spectral estimator $\hat\beta$. Note, however, that if groups are not well-separated, it is possible that the spectral estimator actually outperforms the post-spectral one. Also, note that Assumption \ref{as: well separated groups} is sufficient for consistency of the classifier $\hat g$ but by no means necessary. In particular, using more sophisticated arguments as in \cite{LZZ21} and stronger conditions on the noise variables $v_{it}$ (i.e. isotropic Gaussianity, which means that the random variables $v_{it}$ are i.i.d. centered Gaussian), one can replace Assumption \ref{as: well separated groups} by a much weaker condition $\sum_{t=1}^T(\alpha_{\gamma_1t} - \alpha_{\gamma_2t})^2\geq c \log N$ for all $\gamma_1,\gamma_2=1,\dots,G$ such that $\gamma_1\neq\gamma_2$ and a suitable constant $c>0$.

\begin{assumption}\label{as: group sizes new}
For some constant $c_2>0$, we have $N_{\gamma} \geq c_2 N$ for all $\gamma = 1,\dots,G$.
\end{assumption}

Assumption \ref{as: group sizes new} requires that each group $\gamma = 1,\dots,G$ constitutes a non-trivial fraction of all units. If we were to assume random group assignment, where each unit is assigned to group $\gamma$ with probability $\mathrm{p}_{\gamma}>0$, so that $\sum_{\gamma=1}^G\mathrm p_{\gamma}=1$, and units are assigned independently, this assumption would be satisfied with probability approaching one as $N\to\infty$. 

\begin{assumption}\label{as: N and T}
We have $\log N = o(T)$ and $\log T = o(N)$.
\end{assumption}

Assumption \ref{as: N and T} specifies how fast $T$ and $N$ are required to grow relative to each other in the asymptotics. The most important observation here is that we allow $T$ to be much smaller than $N$, which is the main case of interest for grouped panel data models; see \cite{BM15}.

\begin{theorem}[Consistency of Classifier $\hat g$]\label{thm: classifier consistency new}
Under Assumptions \ref{as: spectral v}--\ref{as: N and T}, we have
$$
\mathrm P\Big(\text{for all }i,j=1,\dots,N, \ \text{we have } \hat g_i = \hat g_j\text{ if and only if }g_i = g_j\Big)\to 1,
$$
as $N\to\infty$.
\end{theorem}

\begin{remark}[On the Role of Sample Splitting in Theorem \ref{thm: classifier consistency new}]\label{rem: sample splitting}
Using sample splitting to construct the vectors $\hat A_1,\dots,\hat A_N$, which are in turn used in the Classification Algorithm to obtain the classifier $\hat g$, is important for our analysis. Specifically, sample splitting allows us to avoid some restrictive assumptions on the geometry of group-specific time effects. Indeed, suppose that we do full-sample estimation, i.e. we set
$
\hat A_{i} = \hat F\hat F'(y_{i} - x_{i}\tilde\beta)
$
for all $i=1,\dots,N$, where $\tilde \beta$ is the full-sample spectral estimator and $\hat F$ is the $T\times G$ matrix consisting of orthonormal eigenvectors corresponding to $G$ largest eigenvalues of the matrix
$$
\hat B = \frac{1}{NT}\sum_{i=1}^N (y_i - x_i'\tilde\beta)(y_i - x_i'\tilde\beta)',
$$
and consider the following example. Let $G=2$ and $\alpha_{1} = 2\alpha_{2}$, where $\|\alpha_2\|\geq c\sqrt T$ for some constant $c>0$. In this example, the assumption of well-separated groups (Assumption \ref{as: well separated groups}) is satisfied but the $T\times T$ matrix $B = N^{-1}\sum_{i=1}^N\alpha_{g_i}\alpha_{g_i}$ has only one non-zero eigenvalue. Therefore, given that $\hat B$ consistently estimate $B$, the matrix $\hat F$ may not have a probability limit. As a result, $T^{-1/2}\hat F'(v_{i1},\dots,v_{iT})'$ may not converge to zero in probability (which is guaranteed in the construction based on sample splitting), and the vectors $\hat A_1,\dots,\hat A_N$ may turn out poor estimators of the vectors $\alpha_{g_1},\dots,\alpha_{g_N}$, leading to inconsistency of the classifier $\hat g$. More generally, with full-sample estimation, we would have to impose in Theorem \ref{thm: classifier consistency new} an extra assumption that the matrix $B = N^{-1}\sum_{i=1}^N \alpha_{g_i}\alpha_{g_i}'$ has $G$ eigenvalues bounded away from zero, which seems difficult to justify. See, however, \cite{LZZ21}, who are able to avoid such conditions without using sample splitting under the isotropic Gaussianity condition mentioned above.\footnote{As a side note, we also observe that \cite{VW04} do not use sample splitting to estimate vectors $\alpha_{g_1},\dots,\alpha_{g_N}$ but some parts of their derivations are difficult to verify. In particular, in Section 5, they use an observation that the projection of a Gaussian random vector on any subspace remains Gaussian but in fact the subspace in their construction is random, in which case the projection may not be Gaussian.}
\qed
\end{remark}

\subsection{Post-Spectral Estimator}
In this subsection, we present two results on our post-spectral estimator $\hat\beta$. First, we show that this estimator is generally consistent. Second, we show that under the assumption of well-separated groups, Assumption \ref{as: well separated groups}, this estimator is $\sqrt{NT}$-consistent and has simple asymptotic distribution.

For all $\nu = (\nu_1,\dots,\nu_N)'\in\{1,\dots,G\}^N$ and $\gamma_1,\gamma_2 = 1,\dots,G$, denote
$$
\mathcal I(\nu,\gamma_1,\gamma_2) = \Big\{i=1,\dots,N\colon \nu_i =\gamma_1\text{ and }g_i = \gamma_2\Big\}
$$
and
$$
\bar x_{\nu,\gamma_1,\gamma_2,t} = \frac{1}{|\mathcal I(\nu,\gamma_1,\gamma_2)|}\sum_{i\in\mathcal I(\nu,\gamma_1,\gamma_2)}x_{it},\quad\text{for all }t=1,\dots,T.
$$
To derive consistency of the post-spectral estimator $\hat\beta$, we will use the following conditions:
\begin{assumption}\label{as: eigenvalues}
For some constant $c_3>0$, the minimal eigenvalue of the matrix
$$
\frac{1}{NT}\sum_{i=1}^N\sum_{t=1}^T(x_{it} - \bar x_{\nu,\nu_i,g_i,t})(x_{it} - \bar x_{\nu,\nu_i,g_i,t})'
$$
is bounded from below by $c_3$ for all $\nu\in \{1,\dots,G\}^N$ with probability $1 - o(1)$.
\end{assumption}

\begin{assumption}\label{as: compact}
(i) The set $\mathcal B$ is compact and (ii) for some constant $C_5>0$, all elements $\{a_{\gamma t}\}_{\gamma,t=1}^{G,T}$ of the set $\mathcal A_{G,T}$ satisfy $|a_{\gamma t}|\leq C_5$ for all $\gamma = 1,\dots,G$ and $t=1,\dots,T$.
\end{assumption}

Assumption \ref{as: eigenvalues} requires that covariates $x_{it}$ have sufficient within-group variation over time and across units. This assumption was used in \cite{BM15} as well. Assumption \ref{as: compact} is a standard compactness condition used in the statistical analysis of non-linear models.

\begin{theorem}[Consistency of Post-Spectral Estimator $\hat\beta$]\label{thm: main consistency}
Under Assumptions \ref{as: spectral v}--\ref{as: bounded rho} and \ref{as: group sizes new}--\ref{as: compact}, we have
$\hat\beta\to_P\beta$.
\end{theorem}

Finally, we prove $\sqrt{NT}$-consistency and asymptotic normality of the post-spectral estimator $\hat\beta$. To do so,
denote
$$
\bar x^{\gamma, t} = \frac{1}{N_{\gamma}}\sum_{i\colon g_i = \gamma}x_{it},\quad\text{for all }\gamma = 1,\dots,G, \ t = 1,\dots, T
$$
and
$$
\check x_{it} = x_{it} - \bar x^{g_i, t},\quad\text{for all }i=1,\dots,N, \ t = 1,\dots,T.
$$
We will use the following condition:
\begin{assumption}\label{as: distribution}
We have (i) $(NT)^{-1}\sum_{i=1}^N\sum_{t=1}^T\check x_{it}\check x_{it}'\to_P\check\Sigma$ for some positive-definite $d\times d$ matrix $\check\Sigma$ and (ii) $(NT)^{-1/2}\sum_{i=1}^N\sum_{t=1}^T v_{it}\check x_{it}\to_D N(0,\Omega)$ for some symmetric $d\times d$ matrix $\Omega$.
\end{assumption}
This assumption is similar to the corresponding assumptions in \cite{BM15}.
\begin{theorem}[Asymptotic Distribution of Post-Spectral Estimator $\hat\beta$]\label{thm: asymptotic distribution}
Under Assumptions \ref{as: spectral v}--\ref{as: distribution}, we have
$$
\sqrt{NT}(\hat\beta - \beta)\to N(0_d,\check\Sigma^{-1}\Omega\check\Sigma^{-1}).
$$
\end{theorem}
\begin{remark}[Variance-Covariance Matrix Estimation]
Theorem \ref{thm: asymptotic distribution} leads to standard inference on the vector of parameters $\beta$ as long as we can consistently estimate the asymptotic variance-covariance matrix $\check\Sigma^{-1}\Omega\check\Sigma^{-1}$. In turn, the latter is simple. Indeed, Assumption \ref{as: distribution}(i) implies that we can consistently estimate $\check\Sigma$ by $(NT)^{-1}\sum_{i=1}^N\sum_{t=1}^T\check x_{it}\check x_{it}'$. Also, to estimate $\Omega$, we can use a formula from \cite{A87}: $(NT)^{-1}\sum_{i=1}^N\sum_{t_1 = 1}^T\sum_{t_2=1}^T\hat v_{it_1}\hat v_{it_2} \check x_{it_1}\check x_{it_2}'$, where $\hat v_{it} = y_{it} -x_{it}'\hat\beta - \hat\alpha_{\hat g_i t}$ for all $i=1,\dots,N$ and $t=1,\dots,T$. Conditions for consistency of this formula are proven in \cite{H07}. For more detailed discussion of variance-covariance matrix estimation, please refer to \cite{BM15}.
\qed
\end{remark}
\begin{remark}[On Assumptions of Theorems \ref{thm: consistency spectral estimator}-\ref{thm: asymptotic distribution}]\label{rem: noncentered noise}
Recall that we stated in the Introduction that $\mathbb E[v_{it}] = 0$ and $\mathbb E[z_{it}] = \mathbb E[v_{it}x_{it}]=0_d$. These identities are quite intuitive and help motivate Assumptions \ref{as: spectral v} and \ref{as: spectral vz}(i). However, we emphasize that these identities are actually not used in the proofs of Theorems \ref{thm: consistency spectral estimator}-\ref{thm: asymptotic distribution}: the results of the theorems hinge only on assumptions that are explicitly mentioned in the statements of the theorems, as well as \eqref{eq: x process general}. This observation will be helpful below, when we discuss dynamic grouped panel data models.
\qed
\end{remark}

\section{Extensions}\label{sec: extensions}
In this section, we consider three extensions of the model we studied above. The first extension is concerned with a dynamic version of the model, i.e. a model that allows for lagged values of $y_{it}$ on the right-hand side of equation \eqref{eq: model}. The second extension is concerned with a high-dimensional version of the model, i.e. a model with high-dimensional $\beta$. The third extension is concerned with an interactive fixed-effect model.

\subsection{Dynamic Model}
Consider a dynamic grouped panel data model
\begin{equation}\label{eq: dynamic panel}
y_{it} = \theta y_{it-1} + x_{it}'\beta + \alpha_{g_it}+v_{it},\quad\text{for all }i=1,\dots,N, \ t=1,\dots,T,
\end{equation}
where $x_{it}$ is a $d\times 1$ vector of pre-determined covariates. As before, assume also that \eqref{eq: x process general} is satisfied. This model is different from the model we studied above as we now allow for the lagged dependent variable on the right-hand side of \eqref{eq: dynamic panel}. In this section, we explain what changes one has to carry out in the spectral and post-spectral estimators to estimate parameters $\theta$ and $\beta$ of this model. We will assume throughout that $|\theta|<1$ since otherwise the extension seems difficult. The results below can be easily extended to allow for additional lagged values of $y_{it}$ on the right-hand side of \eqref{eq: dynamic panel}, i.e. $y_{it-2}$, $y_{it-3}$, etc.

To motivate our approach, note that iterating \eqref{eq: dynamic panel} and substituting \eqref{eq: x process general} yields
\begin{align*}
y_{it-1} 
& = \theta^{t-1} y_{i0} + \sum_{r=0}^{t-2} \theta^r(x_{it-r-1}'\beta + \alpha_{g_it-r-1} + v_{it-r-1}) \\
& = \sum_{m=1}^M \rho_{im}^y\alpha_{g_it}^{m,y} + z_{it}^y,\quad\text{for all }i=1,\dots,N, \ t=1,\dots,T,
\end{align*}
where we denoted
$$
\rho_{im}^y = 1\{m=1\} + \rho_{im}'\beta,\quad\text{for all }i=1,\dots,N , m=1,\dots,M,
$$
$$
\alpha_{\gamma t}^{m,y} = \sum_{r=0}^{t-2}\theta^r \alpha_{\gamma t-r-1}^m,\quad\text{for all }t=1,\dots,T, \ \gamma = 1,\dots,G, \ m=1,\dots,M,
$$
$$
z_{it}^y = \theta^{t-1}y_{i0} + \sum_{r=0}^{t-2}\theta^r(z_{it-r-1}'\beta + v_{it-r-1}),\quad\text{for all }i=1,\dots,N, \ t=1,\dots,T,
$$
where the sum $\sum_{r=0}^{-1}$ is treated as zero. Hence, we can write
$$
y_{it} = \mathring{x}_{it}'\mathring\beta + \alpha_{g_it}+v_{it},
$$
where we denoted $\mathring x_{it} = (y_{it-1},x_{it}')'$ and $\mathring\beta = (\theta,\beta')'$ and the vector of covariates $\mathring x_{it}$ satisfies
$$
\mathring x_{it} = \sum_{m=1}^{2M}\mathring\rho_{im}\mathring\alpha_{g_it}^m + \mathring z_{it},\quad\text{for all }i=1,\dots,N, \ t=1,\dots,T,
$$
where
$$
\mathring\rho_{im} = (0,\rho_{im}')'1\{m\leq M\} + (\rho_{im-M}^y,0_d')'1\{m>M\},
$$
$$
\mathring \alpha_{g_it}^m = \alpha_{g_it}^m1\{m\leq M\} + \alpha_{g_it}^{m-M,y}1\{m>M\},
\text{ and }\mathring z_{it} = (z_{it}^y,z_{it}')'.
$$
Thus, the dynamic model considered here reduces to the model studied in Sections \ref{sec: estimator} and \ref{sec: asymptotic theory} with $x_{it}$, $\beta$, $\rho_{im}$, $\alpha_{\gamma t}^m$, $z_{it}$, and $M$ replaced by $\mathring x_{it}$, $\mathring \beta$, $\mathring\rho_{im}$, $\mathring\alpha_{\gamma t}^m$, $\mathring z_{it}$, and $2M$, respectively. Therefore, the parameters of the dynamic model, $\mathring\beta = (\theta,\beta')'$, can be estimated by the spectral and post-spectral estimators as in Section \ref{sec: estimator} with $M$, $x_{it}$, $d$, and $\mathcal B$ replaced by $2M$, $\mathring x_{it}$, $d+1$, and $\mathring{\mathcal B}$, respectively, where $\mathring{\mathcal B}$ is a parameter space for $\mathring\beta$, e.g. $\mathring{\mathcal B} = [-1+\delta,1-\delta]\times\mathcal B$ for some $\delta\in(0,1)$.


To state the formal results, let Assumptions 4.1--4.12 be the same as Assumptions 3.1--3.12 with $z_{it}$, $x_{it}$, $d$, $\Sigma$, and $\mathcal B$ replaced by $\mathring z_{it}$, $\mathring x_{it}$, $d+1$, $\mathring\Sigma$, and $\mathring{\mathcal B}$, respectively.
\begin{theorem}[Dynamic Model]\label{thm: dynamic model}
In the setting of this section, the statements of Theorems \ref{thm: consistency spectral estimator}-\ref{thm: asymptotic distribution} continue to hold, with $\beta$ replaced by $\mathring\beta$, if Assumptions \ref{as: spectral v}--\ref{as: distribution} are replaced by Assumptions 4.1--4.12.
\end{theorem}
\begin{remark}[Relation between Assumptions 3.1--3.12 and Assumptions 4.1--4.12]
Assumptions 4.1--4.12 are stronger than Assumptions 3.1--3.12 in the sense that they require Assumptions 3.1--3.12 to be supplemented by some extra conditions. We provide a detailed analysis of these conditions in Appendix \ref{sec: additional details on dynamic model} but we note here that one of these extra conditions is that the noise variables $v_{it}$ satisfy
\begin{equation}\label{eq: model with predetermined v's}
\mathbb E[v_{it}|y_{it-1},\dots,y_{i0},x_{it},\dots,x_{i1}] = 0,\quad\text{for all }i=1,\dots,N, \ t=1,\dots,T,
\end{equation}
which is rather standard in the literature and allows for predetermined covariates $x_{it}$. In addition, we note that these conditions do not require the random variables $y_{i0}$ to have mean zero. The latter means that the random vectors $\mathring z_{it}$ may not be centered but this does not contradict the results of Theorem \ref{thm: dynamic model} per our discussion in Remark \ref{rem: noncentered noise}\qed
\end{remark}


\subsection{High-Dimensional Model}
Consider a high-dimensional grouped panel data model
$$
y_{it} = x_{it}'\beta + \alpha_{g_it} +v_{it},\quad\text{for all }i=1,\dots,N, \ t=1,\dots,T,
$$
where $x_{it}$ is a $d\times1$ vector of covariates and $d$ can be large, potentially much larger than $NT$, but the vector of coefficients $\beta$ is sparse, i.e. $s = \|\beta\|_0  = \sum_{k=1}^d 1\{\beta_k\neq 0\} $ is relatively small (in the sense to be made precise later). As before, assume also that \eqref{eq: x process general} is satisfied. Estimating this model requires introducing penalized versions of the spectral and post-spectral estimators. For brevity, however, we focus here on the penalized spectral estimator only, as deriving results for the penalized post-spectral estimator requires taking care of a lot of technicalities but does not seem to bring any new insight.\footnote{Under consistent classification \eqref{eq: classifier consistency statement}, the $\ell^1$-penalized post-spectral estimator will be similar to the usual Lasso estimator, with comparable properties.}

To define the penalized spectral estimator, let $\hat S$ and $\hat\Sigma$ be the same $d\times1$ vector and $d\times d$ matrix as those appearing in Section \ref{sub: preliminary estimation definition}. Note that calculating these quantities requires only $O(d^2)$ operations, and thus is computationally rather simple. We then define the penalized spectral estimator as
$$
\hat\beta_{\lambda} = \arg\min_{b\in\mathbb R^d}\Big\{ b'\hat\Sigma b + \hat S'b +\lambda\|b\|_1 \Big\},
$$
where $\lambda>0$ is a penalty parameter and $\|b\|_1 = \sum_{k=1}^d |b_k|$ denotes the $\ell^1$-norm of the vector $b=(b_1,\dots,b_d)'$. The optimization problem here is convex and can be carried out using standard software.

To analyze this estimator, we are going to rely on the triangular array asymptotics, where the model, as well as the dimension $d$ of the vector of covariates, are allowed to depend on $N$ and $T$ but for brevity of notation, we keep this dependence implicit. Also, we now have to modify Assumption \ref{as: spectral vz}. Indeed, Assumption \ref{as: spectral vz}(i) is impossible to satisfy if $d$ is growing together with $N$ and $T$ because of the $\ell^2$-norm appearing in the assumption. Assumption \ref{as: spectral vz}(ii) also has to be modified as the concept $O_P(\cdot)$ is not well-defined when applied to matrices of growing dimensions. However, since the required modification are only used in this subsection, we spell them directly in the statement of the theorem. In addition, to state the formal result, for any matrix $A$, we will use $\|A\|_{\infty}$ to denote its $\ell^{\infty}$-norm, i.e. the maximum of absolute values of components of $A$.

\begin{theorem}[High-Dimensional Model]\label{thm: hd model}
In the setting of this section, suppose that Assumption \ref{as: spectral v} is satisfied. In addition, suppose that 
\begin{equation}\label{eq: assumption 3.3 hd case}
\max_{1\leq k\leq d}\left|\frac{1}{NT}\sum_{i=1}^N\sum_{t=1}^T v_{it}z_{itk}\right|=O_P\left(\sqrt{\frac{\log d}{NT}}\right)
\end{equation}
and 
\begin{equation}\label{eq: assumption 3.4 hd case}
\max_{1\leq k,l\leq d}\left|\frac{1}{NT}\sum_{i=1}^N\sum_{t=1}^Tz_{itk}z_{itl} - \frac{\Sigma_{kl}}{2}\right|=O_P\left(\sqrt{\frac{\log d}{NT}}\right)
\end{equation}
for some positive definite $d\times d$ matrix $\Sigma$. Also, suppose that for some constant $C_z>0$, we have $\|\sum_{t=1}^Tu_tz_{it}'\beta\|_{\psi_2}\leq C_z$ for all $i=1,\dots,N$ and $u=(u_1,\dots,u_T)'\in\mathcal S^T$. Moreover, suppose that $\log d = o(N)$ and $\|\beta\|_1\leq C_{\beta}$ for some constant $C_{\beta}>0$. Finally, let $c_{\Sigma}>0$ be the minimal eigenvalue of the matrix $\Sigma$, $c_{\lambda}>1$ be some constant, and $S = -2\Sigma\beta$ be a $d\times 1$ vector. Then
\begin{equation}\label{eq: main convergence in hd model}
\|\hat S - S\|_{\infty}\vee\|\hat\Sigma - \Sigma\|_{\infty} = O_P\left(\frac{1}{T\wedge N} + \sqrt{\frac{\log d}{NT}}\right).
\end{equation}
Moreover, on the intersection of events
\begin{equation}\label{eq: lambda dominance}
\lambda \geq c_{\lambda}\Big( \|\hat S-S\|_{\infty} + 2C_{\beta}\|\hat\Sigma - \Sigma\|_{\infty} \Big)
\end{equation}
and
\begin{equation}\label{eq: second event dominance}
s\|\hat\Sigma - \Sigma\|_{\infty} \leq \frac{c_{\Sigma}}{2(1+\bar c_{\lambda})^2},
\end{equation}
we have
\begin{equation}\label{eq: implication l1 norm hd case}
\|\hat\beta_{\lambda} - \beta\|_1 \leq \frac{2(1+c_{\lambda})(1+\bar c_{\lambda})s\lambda}{c_{\Sigma}c_{\lambda}}\quad\text{and}\quad \|\hat\beta_{\lambda}-\beta\| \leq \frac{2(1+c_{\lambda})\sqrt s\lambda}{c_{\Sigma}c_{\lambda}}
\end{equation}
where $\bar c_{\lambda} = (c_{\lambda} + 1)/(c_{\lambda}-1)$.
\end{theorem}
\begin{remark}[Main Implication of Theorem \ref{thm: hd model}]
By the result \eqref{eq: main convergence in hd model}, it follows that selecting a constant $C>0$ large enough and setting
\begin{equation}\label{eq: lambda choice}
\lambda = C\left(\frac{1}{T\wedge N} + \sqrt{\frac{\log d}{NT}}\right)
\end{equation}
will ensure that the event \eqref{eq: lambda dominance} holds with probability arbitrarily close to one. Also, the result \eqref{eq: main convergence in hd model} ensures that the event \eqref{eq: second event dominance} holds with probability approaching one as long as
\begin{equation}\label{eq: sparsity convergence}
s\left(\frac{1}{T\wedge N} + \sqrt{\frac{\log d}{NT}}\right) \to 0.
\end{equation}
Thus, setting $\lambda$ according to \eqref{eq: lambda choice} and assuming \eqref{eq: sparsity convergence} ensures that
$$
\|\hat\beta_{\lambda} - \beta\|_1 \leq \bar Cs\left(\frac{1}{T\wedge N} + \sqrt{\frac{\log d}{NT}}\right)\ \text{and}\  \|\hat\beta_{\lambda} - \beta\|_2 \leq \bar C\sqrt s\left(\frac{1}{T\wedge N} + \sqrt{\frac{\log d}{NT}}\right)
$$
with probability arbitrarily close to one, where $\bar C$ is some constant. Here, \eqref{eq: sparsity convergence} explains how small $s$ has to be in order for our results to go through. 

The choice of the penalty parameter $\lambda$ in \eqref{eq: lambda choice} is not practical, as it does not specify the constant $C>0$. From the asymptotic point of view, we can of course set $C = \log\log n$, in which case the bounds on $\|\hat\beta_{\lambda} - \beta\|_1$ and $\|\hat\beta_{\lambda} - \beta\|_2$ presented above hold with probability approaching one, with $\bar C$ replaced by $\bar C\log\log n$, but this choice may of course perform poorly in finite samples. However, the question of practical choices of $\lambda$ is beyond the scope of this paper and we leave it for future work.
\qed
\end{remark}
\begin{remark}[Conditions of Theorems \ref{thm: hd model}]
Conditions \eqref{eq: assumption 3.3 hd case} and \eqref{eq: assumption 3.4 hd case} used in Theorem \ref{thm: hd model} are a natural extension of conditions used in the literature on high-dimensional models and can be derived from appropriate maximal inequalities. For example, assuming that components of $z_{it}$ are bounded and $v_{it}$ is sub-Gaussian, both \eqref{eq: assumption 3.3 hd case} and \eqref{eq: assumption 3.4 hd case} follow from a combination of the union bound and a time-series version of Hoeffding's inequality as long as the dependence of random vectors $(v_{it},z_{it}')'$ across $t$ is not too strong. Moreover, a version of Theorem \ref{thm: hd model} can be derived if conditions \eqref{eq: assumption 3.3 hd case} and \eqref{eq: assumption 3.4 hd case} are relaxed, in which case we would simply have to correspondingly modify the bound \eqref{eq: main convergence in hd model}.
\qed
\end{remark}

\subsection{Interactive Fixed-Effect Model}
Consider an interactive fixed-effect panel data model
\begin{equation}\label{eq: ifepd}
y_{it} = x_{it}'\beta + \kappa_i'\phi_t + v_{it},\quad\text{for all }i=1,\dots,N, \ t=1,\dots,T,
\end{equation}
where $\phi_t$ is a $J\times 1$ vector of factors and $\kappa_i$ is a $J\times 1$ vector of factor loadings. \cite{B09} developed an OLS-type interactive fixed-effect estimator of $\beta$ in this model. The proposed estimator, however, is a solution to a non-convex optimization problem, and like in the case of the grouped fixed-effect estimator, it may be difficult to find (ensure that we found) the global solution to this optimization problem. To fix this issue, \cite{MW19} developed an estimator of $\beta$ based on the nuclear-norm minimization. Their estimator solves a convex optimization problem and so is computationally rather simple. However, the convergence rate of their estimator is only $(T\wedge N)^{-1/2}$, which can be rather slow, especially if $T\ll N$ or $N\ll T$. It is therefore of interest to look for alternative estimators of $\beta$ in this model.

In this subsection, we demonstrate that a trivial modification of our spectral estimator can be used to estimate $\beta$ in this model with the rate $(T\wedge N)^{-1}$ as long as we assume that the covariates $x_{it}$ satisfy the factor model as well, namely
\begin{equation}\label{eq: ifepd covariates}
x_{it} = \omega_i'\phi_t + z_{it},\quad\text{for all }i=1,\dots,N, \ t=1,\dots,T,
\end{equation}
where $\phi_t$ is a $J\times 1$ vector of factors and $\omega_i$ is a $J\times d$ matrix of factor loadings. Note here that it is without loss of generality to assume that the same factors $\phi_t$ appear both in \eqref{eq: ifepd} and in \eqref{eq: ifepd covariates} as we can always merge the factors from two equations.

Let $\tilde\beta_J$ be the same spectral estimator $\tilde\beta$ as that defined in Section \ref{sub: preliminary estimation definition} with $J$ replacing $GM$. We then have the following result.
\begin{theorem}[Interactive Fixed-Effect Model]\label{thm: interactive model}
In the setting of this section, suppose that Assumptions \ref{as: spectral v} and \ref{as: spectral vz} are satisfied. Then
$$
\tilde\beta_J = \beta + O_P\left(\frac{1}{T\wedge N}\right).
$$
\end{theorem}
\begin{remark}[Combining $\tilde\beta_J$ and Interactive Fixed-Effect Estimator]
Although the convergence rate of our estimator $\tilde\beta_J$ is faster than the convergence rate of the estimator proposed in \cite{MW19}, it is still slower than the convergence rate of the interactive fixed-effect estimator proposed in \cite{B09}, which is $(NT)^{-1/2}$. Like in \cite{MW19}, we therefore consider our estimator as a starting point in the problem of finding a local minimum of the optimization problem used to define the interactive fixed-effect estimator. The resulting estimator will be asymptotically equivalent to the interactive fixed-effect estimator under certain conditions.
\qed
\end{remark}

\section{Monte Carlo Simulation Study}\label{sec: monte carlo}
In this section, we present results of a Monte Carlo simulation study. We compare the performance in finite samples of the spectral (S) and post-spectral (P-S) estimators to several natural alternatives. Specifically, we consider the grouped fixed effect (GFE) estimator proposed in \cite{BM15}, the least squares (LS) estimator proposed in \cite{B09}, and the penalized nuclear norm (Pen NN) estimator proposed in \cite{MW19}. The comparison with LS and Pen NN estimators is pertinent because equation (\ref{eq: model}) can be seen as a particular case of an interactive fixed effects model where the factors are the group trends, $\alpha_{gt}$, and the loadings are restricted to be of the form $(0,\ldots,1,\ldots,0)$, with 1 in the $g_i^{th}$ position. 

Computation of the GFE estimator is not trivial. Even with $N=100$ and $G=2$, it is hard to be sure that any given algorithm produces the global minimum of the GFE optimization problem since this would require separately considering $G^N$ partitions of units $i=1,\dots,N$ into $G$ groups. Having in mind this issue, we calculate the GFE estimator using Algorithm 1 in \cite{BM15} with 500 random initial values, which is in a nutshell the Lloyd algorithm with covariates. Given a  classification, this algorithm obtains the regression coefficients and the group-specific trends by minimizing OLS, and given the parameters, units are classified into groups according to the smallest individual-specific residual. We generate the initial values from the standard normal distribution in $\mathbb R^{2 + G\times T}$, where the first two and the last $G\times T$ components correspond to $\beta = (\beta_1,\beta_2)'$ and $\{\alpha_{\gamma t}\}_{\gamma,t=1}^{G,T}$, respectively. 

We also consider an infeasible GFE estimator, which uses only one initial value but this value is chosen to be the oracle estimator, i.e. the pooled OLS estimator based on the true partition of units into groups. Using this value substantially increases the chances to find the global minimum of the GFE optimization problem as we expect the solution to be near the oracle estimator. In what follows, we refer to this infeasible GFE estimator as I-GFE.

Computation of the LS estimator can also be problematic, as its objective function is prone to local minima as well. To deal with this problem, like in the case of the GFE estimator, we use 500 random initial values and choose the one that gives the best value of the LS criterion function.


Computation of the Pen-NN estimator is relatively straightforward: it minimizes a convex objective function. To make the comparison with the other estimators fair, however, we choose a penalty parameter for this estimator so that the total number of generated groups is equal to the true value, $G$, instead of using the data-driven procedure proposed by the authors. Also, instead of presenting results for the plain Pen-NN estimator, following the original paper \cite{MW19}, we actually present results for the LS estimator that uses the Pen-NN estimator as initial value, as adding an extra step of least squares minimization is supposed to make the estimator more precise. With some abuse of notation, we still refer to this estimator as Pen-NN.\footnote{\cite{MW19} proposed several other estimators as well but we have chosen to present results for the Pen-NN estimator only, as it dominates the other estimators in our simulations.} 


Finally, for all data-generating processes, we also report results for the oracle estimator, which is the pooled OLS estimator based on the true partition of units into groups. Although we do not discuss these results, an interested reader can use them as a benchmark for the results on the other estimators.



Next, we discuss the data-generating processes. We generate the data $(x_{it}',y_{it})'$, $i=1,\dots,N$ and $t=1,\dots,T$, according to equations \eqref{eq: model} and \eqref{eq: x process general}, where we set $d=2$, so that the model contains two covariates. Depending on the experiment, we set $N=100$, $200$, or $400$ and $T=20$, $50$, or $100$. We also set $G=2$ or $7$ and $M=1$ or $2$. The random variables $v_{it}$ and components of the random vectors $z_{it}$  are generated from a truncated standard normal distribution, $Z1\{|Z|\leq C\}$, where $Z\sim N(0,1)$ and $C=20$, independently of each other and across indices. The random variables $\alpha_{\gamma t}^m$ are generated from a truncated normal distribution as well but with different variance values, $Z1\{|Z|\leq C\}$, where $Z\sim N(0,\sigma^2)$, independently across indices and of all other random variables. Depending on the experiment, we set $\sigma^2$ equal to either 1 or 4. We have checked that results are robust with respect to different specifications of the truncation constant $C$. Further, we set the number of observations $i$ within each group $\gamma$ to be the same, except for the last group, which contains more observations if $G$ does not divide $N$. We have also checked with uneven units across groups and results are unchanged.

Throughout all experiments, we set $\beta = (\beta_1,\beta_2)'$, where $\beta_1=-1$ and $\beta_2 = .8$ but we note that the particular choice of these values does not seem to matter much for the simulation results. Also, when $M=1$, we set $\rho_{i1} = (\rho_{i11},\rho_{i12})'$ with $\rho_{i11} = \varrho + Z_{i1}1\{|Z_{i1}|\leq C\}$ and $\rho_{i12} = \varrho + Z_{i2}1\{|Z_{i2}|\leq C\}$, where $Z_{i1}$ and $Z_{i2}$ are independent of each other, across $i$, and of all other random variables in the data-generating process, and where $\varrho = 3$. Here, $\varrho$ can be understood as a measure of endogeneity in the model as it govern the correlation between covariates and grouped-specific time effects. When $M=2$, we set $\rho_{i11} = \varrho + Z_{i11}1\{|Z_{i11}|\leq C\}$, $\rho_{i21} = 1 + Z_{i21}1\{|Z_{i21}|\leq C\}$, and $\rho_{i12} = \varrho + Z_{i12}1\{|Z_{i12}|\leq C\}$ and $\rho_{i22} = Z_{i22}1\{|Z_{i22}|\leq C\}$, where random variables $Z_{imj}$ for $m=1,2$ have standard normal distribution and are independent of each other, across indices, and of all other random variables in the data-generating process.

We present results in Tables 1--4 at the end of Supplementary Materials. For each combination of $N$ and $T$, these tables give the mean absolute error (MAE) for S, P-S, LS, Pen-NN, I-GFE, GFE, and oracle estimators. In addition, the tables give the fraction of misclassified units based on our classification algorithm (which is based on the spectral estimator, and which is the second step in constructing our post-spectral estimator) and the fraction of misclassified units in the construction of the GFE estimator. Each table presents results in the upper panel for $G =2$ and in the lower panel for $G = 7$. Tables 1 and 2 correspond to $M=1$ and Tables 3 and 4 correspond to $M=2$. In addition, Tables 1 and 3 correspond to $\sigma^2=1$, and Tables 2 and 4 correspond to $\sigma^2=4$. When increasing the variance $\sigma^2$, we are effectively making the groups more separated, and hence group classification improves. Both the MAE and the fraction of misclassified units are calculated as the average over 50 simulations. 

We now describe the results. We start with the behavior of the S and P-S estimators. In most cases, the P-S estimator outperforms the S estimator. Hence, it seems that doing the "post" step delivers a better estimator. Also, the P-S estimator tends to be comparable to the oracle estimator across all tables as long as $N$ and $T$ increase sufficiently. In addition, while the P-S estimator seems to have more trouble when the separation of the groups is less, i.e. in Tables 1 and 3, it still quickly achieves the oracle for moderate values of $N$ and $T$. Moreover, when $G=2$, the P-S estimator seems to achieve the oracle faster than when $G = 7$, but even for small values of $N$ and $T$, the P-S estimator is a robust reliable estimator. 

Let us now turn to the comparison with the grouped fixed effect estimators. Not surprisingly, the infeasible I-GFE  estimator is often comparable to the oracle. However, the performance of the feasible GFE estimator varies substantially depending on the number and separation of the groups. For instance, in the upper panel of Table 1 ($G=2)$, where separation of groups is moderate, and for moderate values of $T$, the GFE estimator is close to the I-GFE estimator, indicating that the minimum might have been achieved. However, in the lower panel of Table 1 ($G = 7$), it is unlikely that 500 initial values cover a significant part of the space of all possible partitions of $N$ units into $G$ groups, and the performance of the GFE estimator deteriorates.  In fact, missclassification for the GFE estimator is on the level of $70\%$ for moderate values of $T$ when $G = 7$. In this case, the P-S estimator has a comparative advantage relative to the GFE estimator. This same pattern is reproduced in Table 2. However, since the groups are more separated, and the r-squared of the model is higher, the GFE estimator improves. Interestingly, the behavior of the GFE estimator when $M=2$ improves relative to $M=1$.


As far the LS estimator is concerned, when the number of groups increases, the finite sample performance of this estimator deteriorates, although it improves as both $N$ and $T$ increase. However, it is only for large values of $T$ that the behavior is comparable to the oracle. At least in this exercise, the P-S estimator outperforms the LS estimator.

Finally, the pen NN estimator is, in general, dominated by the other estimators for our data-generating processes. However, like with the LS estimator, its performance improves as we increase $N$ and $T$. For example, in Table 4, when $G=2$, the estimator is actually comparable to the oracle for large values of $N$ and $T$.

\section{Proofs of Theorems \ref{thm: consistency spectral estimator}, \ref{thm: classifier consistency new}, and \ref{thm: main consistency}}\label{sec: remaining proofs}
\begin{proof}[Proof of Theorem \ref{thm: consistency spectral estimator}]
Throughout this proof, we use $c$ and $C$ to denote strictly positive constants that can change from place to place but can be chosen to depend on $C_1$ and $C_2$ only.

We will prove that 
\begin{equation}\label{eq: key convergence proof}
\lambda^b_1 + \dots + \lambda^b_{2GM+2} = b'\Sigma b + S'b + L + O_P\left(\frac{1}{T\wedge N}\right),\quad\text{for all }b\in\mathbb R^d,
\end{equation}
where $\Sigma$ is a positive definite matrix appearing in Assumption \ref{as: spectral vz}(ii), $S = -2\Sigma \beta$, and $L = L_N = \beta'\Sigma\beta + 2(NT)^{-1}\sum_{i=1}^N\sum_{t=1}^T v_{it}^2$. Then  $\beta = -\Sigma^{-1}S/2$, and so $\tilde\beta = -\hat\Sigma^{-1}\hat S/2$ satisfies \eqref{eq: convergence rate spectral estimator} by the delta method since \eqref{eq: key convergence proof} implies that 
$$
\|\hat S - S\|\vee\|\hat\Sigma - \Sigma\| = O_P\left(\frac{1}{T\wedge N}\right)
$$
by construction in Section \ref{sub: preliminary estimation definition}.

To prove \eqref{eq: key convergence proof}, fix any $b\in\mathbb R^d$. For brevity of notation, we will write $A$ instead of $A^b$. Then
\[
A_{ij}=\frac{1}{NT}\sum_{t=1}^{T}\Big\{ f_{it}-f_{jt}+z_{it}'(\beta-b)-z_{jt}'(\beta-b)+v_{it}-v_{jt}\Big\} ^{2},\quad i,j=1,\dots,N,
\]
where
$$
f_{it} = \alpha_{g_it}^1(1+\rho_{i1}'(\beta - b)) + \sum_{m=2}^M \alpha_{g_it}^m\rho_{im}'(\beta - b),\quad i=1,\dots,N,\,t=1,\dots,T.
$$
Also, let $R$ be an $N\times N$ matrix whose $(i,j)$-th element is
$$
R_{ij}=-\frac{2}{NT}\sum_{t=1}^{T}\Big\{(\beta-b)'z_{it}z_{jt}(\beta-b)+v_{it}v_{jt}+v_{it}z_{jt}'(\beta-b)+v_{jt}z_{it}'(\beta-b)\Big\},\ i,j=1,\dots,N.
$$
As we will see below, this matrix represents the asymptotically negligible component of
$A$ in the sense that
\begin{equation}
\|R\|=O_{P}\left(\frac{1}{T\wedge N}\right).\label{eq: rate}
\end{equation}
On the other hand,
\begin{align*}
\textrm{tr}(R) & =-\frac{2}{NT}\sum_{i=1}^{N}\sum_{t=1}^{T}\Big\{(\beta-b)'z_{it}z_{it}(\beta-b)+v_{it}^{2}+2v_{it}z_{it}'(\beta-b)\Big\}\\
 & =-\frac{2}{NT}\sum_{i=1}^{N}\sum_{t=1}^{T}\Big\{(\beta-b)'z_{it}z_{it}(\beta-b)+v_{it}^{2}\Big\}+O_{P}\left(\frac{1}{\sqrt{NT}}\right)
\end{align*}
by Assumption \ref{as: spectral vz}(i). Thus, given that $\textrm{tr}(A)=0$, the matrix $A_0=A-R$
satisfies
\begin{align*}
\textrm{tr}(A_0)&=\frac{2}{NT}\sum_{i=1}^{N}\sum_{t=1}^{T}\Big\{(\beta-b)'z_{it}z_{it}(\beta-b)+v_{it}^{2}\Big\}+O_{P}\left(\frac{1}{\sqrt{NT}}\right)\\
& = b'\Sigma b + S'b + L + O_{P}\left(\frac{1}{\sqrt{NT}}\right)
\end{align*}
by Assumption \ref{as: spectral vz}(ii). In addition, we will show below that $A_0$ has at most $2GM+2$ non-zero eigenvalues. Hence,
$$
\lambda_{1}+\dots+\lambda_{2GM+2}=b'\Sigma b + S'b + L+O_P\left(\frac{1}{\sqrt{NT}}\right)+O_{P}\left(\frac{1}{T\wedge N}\right)
$$
by Lemma \ref{lem: eigenvalue comparison} and \eqref{eq: rate}. This gives \eqref{eq: key convergence proof} since $\sqrt{NT}\geq T\wedge N$. Therefore, it remains to prove (\ref{eq: rate}) and to show that $A_0$ has at most $2GM+2$ non-zero eigenvalues. 

To derive the bound on the number of non-zero eigenvalues of $A_0$, let
us first introduce some notation. For all $t=1,\dots,T$, denote
$$
F_{t}=(f_{1t},\dots,f_{Nt})',\ V_{t}=(v_{1t},\dots,v_{Nt})', \ Z_{t}=(z_{1t}'(\beta-b),\dots,z_{Nt}'(\beta-b))'.
$$
Also, denote $1_{N}=(1,\dots,1)'\in\mathbb{R}^{N}$. In addition, for any
vectors $a=(a_{1},\dots,a_{N})'$ and $b=(b_{1},\dots,b_{N})'$, write
$(ab)=(a_{1}b_{1},\dots,a_{N}b_{N})'$ and $a^{2}=(a_{1}^{2},\dots,a_{N}^{2})'$.
Then
\begin{align*}
A & =\frac{1}{NT}\sum_{t=1}^{T}\Big(F_{t}^{2}1_{N}'+1_{N}(F_{t}^{2})'+Z_{t}^{2}1_{N}'+1_{N}(Z_{t}^{2})'+V_{t}^{2}1_{N}'+1_{N}(V_{t}^{2})'\Big)\\
 & +\frac{2}{NT}\sum_{t=1}^{T}\Big((F_{t}Z_{t})1_{N}'+1_{N}(F_{t}Z_{t})'+(F_{t}V_{t})1_{N}'+1_{N}(F_{t}V_{t})'+(Z_{t}V_{t})1_{N}'+1_{N}(Z_{t}V_{t})'\Big)\\
 & -\frac{2}{NT}\sum_{t=1}^{T}\Big(F_{t}Z_{t}'+Z_{t}F_{t}'+F_{t}V_{t}'+V_{t}F_{t}'+Z_{t}V_{t}'+V_{t}Z_{t}' +F_tF_t' +Z_tZ_t' + V_tV_t'\Big).
\end{align*}
Also,
$$
R=-\frac{2}{NT}\sum_{t=1}^{T}\Big(Z_{t}Z_{t}'+V_{t}V_{t}'+V_{t}Z_{t}'+Z_{t}V_{t}'\Big).
$$
Hence, denoting
\[
Q=\frac{1}{NT}\sum_{t=1}^{T}\Big(F_{t}^{2}+Z_{t}^{2}+V_{t}^{2}+2(F_{t}Z_{t})+2(F_{t}V_{t})+2(Z_{t}V_{t})\Big)
\]
and recalling $A_0 = A - R$, we have
\begin{align*}
A_0&=Q1_{N}'+1_{N}Q'-\frac{2}{NT}\sum_{t=1}^{T}\Big(F_{t}F_{t}'+F_{t}Z_{t}'+Z_{t}F_{t}'+F_{t}V_{t}'+V_{t}F_{t}'\Big)\\
&=Q1_{N}'+1_{N}Q'-\frac{2}{NT}\sum_{t=1}^{T}F_{t}(F_{t}+Z_{t}+V_{t})'-\frac{2}{NT}\sum_{t=1}^{T}(Z_{t}+V_{t})F_{t}'.
\end{align*}
The last expression makes bounding the number of non-zero eigenvalues of the matrix $A_0$ straightforward. Indeed, $\textrm{rank}(Q1_{N}')=\textrm{rank}(1_{N}Q')=1$.
Also, for each $m=1,\dots,M$ and $\gamma=1,\dots,G$, define an $N\times 1$ vector $\varrho_{m\gamma}$ whose $i$-th element is 
\[
\varrho_{m\gamma i}=\begin{cases}
\mathbb I\{m=1\}+\rho_{im}'(\beta - b) & \text{if }g_{i}=\gamma,\\
0 & \text{if }g_{i}\neq\gamma,
\end{cases}\quad i=1,\dots,N.
\]
Then
$
F_{t}=\sum_{m=1}^{M}\sum_{\gamma=1}^{G}\alpha_{\gamma t}^{m}\varrho_{m\gamma}
$
for all $t=1,\dots,T$, and so, for any vector $a=(a_{1},\dots,a_{N})'$, the vector
\begin{align*}
\sum_{t=1}^{T}F_{t}(F_{t}+Z_{t}+V_{t})'a&=\sum_{m=1}^{M}\sum_{\gamma=1}^{G}\sum_{t=1}^{T}\alpha_{\gamma t}^{m}\varrho_{m\gamma}(F_{t}+Z_{t}+V_{t})'a\\
&=\sum_{m=1}^{M}\sum_{\gamma=1}^{G}\varrho_{m\gamma}\sum_{t=1}^{T}\alpha_{\gamma t}^{m}(F_{t}+Z_{t}+V_{t})'a
\end{align*}
belongs to the linear subspace of $\mathbb R^N$ spanned by the vectors $\varrho_{11},\dots,\varrho_{MG}$. Hence,
\[
\textrm{rank}\left(\sum_{t=1}^{T}F_{t}(F_{t}+Z_{t}+V_{t})'\right)\leq MG.
\]
Similarly, using the fact that the row rank is equal to the column
rank,
\[
\textrm{rank}\left(\sum_{t=1}^{T}(Z_{t}+V_{t})F_{t}'\right)\leq MG.
\]
Thus, given that the rank operator is sub-additive, we conclude that
\[
\textrm{rank}(A_0)\leq2MG+2,
\]
and so $A_0$ has at most $2MG+2$ non-zero eigenvalues, as claimed.

It remains to prove (\ref{eq: rate}). To do so, we will show that
\begin{equation}
\left\Vert \frac{1}{NT}\sum_{t=1}^{T}V_{t}V_{t}'\right\Vert =O_{P}\left(\frac{1}{T\wedge N}\right)\label{eq: remains to prove}
\end{equation}
and note that $\|(NT)^{-1}\sum_{t=1}^T Z_t Z_t'\| = O_P(1/(T\wedge N))$ follows from the exactly same argument (the former is obtained from Assumption \ref{as: spectral v}(i) whereas the latter is from Assumption \ref{as: spectral v}(ii)). Then
$$
\left\|\frac{1}{NT}\sum_{t=1}^T V_t Z_t'\right\| \leq \sqrt{\left\|\frac{1}{NT}\sum_{t=1}^T V_t V_t'\right\|}\sqrt{\left\|\frac{1}{NT}\sum_{t=1}^T Z_t Z_t'\right\|} = O_P\left(\frac{1}{T\wedge N}\right)
$$
by Lemma \ref{lem: cs matrices}. This gives (\ref{eq: rate}).

To prove \eqref{eq: remains to prove}, we proceed by appropriately modifying the proof of Theorem 4.7.1
in \cite{V18}. Denote $\mathcal{V}=(v_{1},\dots,v_{N})'$, where $v_{i}=(v_{i1},\dots,v_{iT})'$
for all $i=1,\dots,N$. Then uniformly over $u\in\mathcal{S}^{T}$,
\[
u'\mathbb{E}\left[\frac{\mathcal{V}'\mathcal{V}}{N}\right]u=\frac{1}{N}\sum_{i=1}^{N}\mathbb{E}\left[u'v_{i}v_{i}'u\right]=\frac{1}{N}\sum_{i=1}^{N}\mathbb{E}\left[(v_{i}'u)^{2}\right]\leq C
\]
by Assumption \ref{as: spectral v}(i). Thus,
\begin{equation}
\left\Vert \mathbb{E}\left[\frac{\mathcal{V}'\mathcal{V}}{N}\right]\right\Vert \leq C.\label{eq: expectation}
\end{equation}
Further, by Exercise 4.4.3, part 2, in \cite{V18},
\begin{equation}
\left\Vert \frac{\mathcal{V}'\mathcal{V}}{N}-\mathbb{E}\left[\frac{\mathcal{V}'\mathcal{V}}{N}\right]\right\Vert \leq2\max_{u\in\mathcal{N}}\left|u'\left(\frac{\mathcal{V}'\mathcal{V}}{N}-\mathbb{E}\left[\frac{\mathcal{V}'\mathcal{V}}{N}\right]\right)u\right|,\label{eq: discretization}
\end{equation}
where $\mathcal{N}$ is any $(1/4)$-net in $\mathcal{S}^{T}$. Moreover, by Corollary 4.2.13 in \cite{V18}, the net $\mathcal N$ can be chosen so that $|\mathcal{N}_{\epsilon}|\leq9^{T}$, which we are going to use below.

Next, fix any $u\in\mathcal{N}$. Then by Assumption \ref{as: spectral v}(i) and Lemma 2.7.6
in \cite{V18}, for all $i=1,\dots,N$, we have $\|(u'v_{i})^{2}\|_{\psi_{1}}\leq C$, and so by Exercise
2.7.10 in \cite{V18}, $\|(u'v_{i})^{2}-\mathbb{E}[(u'v_{i})^{2}]\|_{\psi_{1}}\leq C$.
Hence, by Bernstein's inequality (Corollary 2.8.3 in \cite{V18}), for
any $\epsilon>0$,
\begin{align*}
&\mathrm{P}\left(\left|u'\left(\frac{\mathcal{V}'\mathcal{V}}{N}-\mathbb{E}\left[\frac{\mathcal{V}'\mathcal{V}}{N}\right]\right)u\right|\geq\epsilon\right)\\
&\qquad =\mathrm{P}\left(\left|\frac{1}{N}\sum_{i=1}^N \Big((v_i'u)^2 - \mathbb E[(v_i'u)^2]\Big)\right|\geq \epsilon\right) \leq2\exp\left[-c(\epsilon\wedge\epsilon^{2})N\right],
\end{align*}
and so by the union bound,
\[
\mathrm{P}\left(\max_{u\in\mathcal{N}}\left|u'\left(\frac{\mathcal{V}'\mathcal{V}}{N}-\mathbb{E}\left[\frac{\mathcal{V}'\mathcal{V}}{N}\right]\right)u\right|\geq\epsilon\right)\leq2\times9^{T}\times\exp\left[-c(\epsilon\wedge\epsilon^{2})N\right].
\]
Therefore, setting $\epsilon=c^{-1}(\log9)(T/N)+1$, we obtain
\[
\mathrm{P}\left(\max_{u\in\mathcal{N}}\left|u'\left(\frac{\mathcal{V}'\mathcal{V}}{N}-\mathbb{E}\left[\frac{\mathcal{V}'\mathcal{V}}{N}\right]\right)u\right|\geq \frac{CT}{N}+1\right)\leq2\exp(-cN).
\]
Combining this bound with (\ref{eq: expectation}) and (\ref{eq: discretization})
gives
\begin{equation}\label{eq: useful bound from proof of spectral consistency}
\left\Vert \frac{\mathcal{V}'\mathcal{V}}{N}\right\Vert =O_{P}\left(\frac{T}{N}+1\right).
\end{equation}
Hence,
$
\|\mathcal{V}\|=O_{P}(\sqrt{T}+\sqrt{N}),
$
and so
\[
\left\Vert \frac{1}{T}\sum_{t=1}^{T}V_{t}V_{t}'\right\Vert =\left\Vert \frac{\mathcal{V}\mathcal{V}'}{T}\right\Vert =\frac{\|\mathcal{V}\|^{2}}{T}=O_{P}\left(1+\frac{N}{T}\right).
\]
Conclude that
\[
\left\Vert \frac{1}{NT}\sum_{t=1}^{T}V_{t}V_{t}'\right\Vert =O_{P}\left(\frac{1}{T}+\frac{1}{N}\right) = O_{P}\left(\frac{1}{T\wedge N}\right),
\]
which gives (\ref{eq: remains to prove}) and completes the proof.
\end{proof}

\begin{proof}[Proof of Theorem \ref{thm: classifier consistency new}]
First, we introduce some notations. For all $i=1,\dots,N$, denote $v_i = (v_{i1},\dots,v_{iT})'$ and for all $\gamma=1,\dots,G$, denote $\alpha_{\gamma} = (\alpha_{\gamma_1},\dots,\alpha_{\gamma T})'$. Then the model \eqref{eq: model} can be written as
$$
y_i = x_i\beta+ \alpha_{g_i} + v_i,\quad i=1,\dots,N.
$$
Further, let $F$ be a $T\times G$ matrix whose columns are orthonormal and span the vectors $\alpha_{1},\dots,\alpha_{G}$. Then for each $\gamma=1,\dots,G$, there exists a $G\times1$ vector $p_{\gamma}$ such that $\alpha_{\gamma}=\sqrt TF p_{\gamma}$. Here, without loss of generality, we assume that $\Lambda = N^{-1}\sum_{i=1}^N p_{g_i}p_{g_i}'$ is a diagonal matrix $\mathrm{diag}(\lambda_1,\dots,\lambda_G)$ with $\lambda_1\geq\dots\geq\lambda_G\geq 0$ since otherwise we can consider the eigenvalue decomposition $S\Lambda S'$ of the matrix $N^{-1}\sum_{i=1}^N p_{g_i}p_{g_i}'$ and replace $F$ and $p_{1},\dots,p_{G}$ by $FS$ and $S^{-1}p_1,\dots,S^{-1}p_{G}$, respectively. Also, throughout this proof, we use $c$ and $C$ to denote strictly positive constants that can change from place to place but can be chosen to depend on $C_1$, $C_2$, $C_3$, $C_4$, $c_1$, and $c_2$ only.

We will show below that
\begin{equation}\label{eq: A estimator bound}
\|\hat A_i - \alpha_{g_i}\| = o_P(\sqrt T)\quad\text{uniformly over }i=1,\dots,N.
\end{equation}
Therefore, by Assumption \ref{as: well separated groups}, there exists a constant $\bar C>0$ such that with probability approaching one,
\begin{equation}\label{eq: good separation 1}
\|\hat A_i - \hat A_j\| \leq \bar C\sqrt T\quad\text{for all $i,j=1,\dots,N$ such that $g_i = g_j$}
\end{equation}
and
\begin{equation}\label{eq: good separation 2}
\|\hat A_i - \hat A_j\| > 2\bar C\sqrt T\quad\text{for all $i,j=1,\dots,N$ such that $g_i \neq g_j$}.
\end{equation}
Now, assume that \eqref{eq: good separation 1} and \eqref{eq: good separation 2} are satisfied and let $\mathcal A_1(\lambda),\dots,\mathcal A_{m(\lambda)}(\lambda)$ be the partition of units generated by the Classification Algorithm from Section \ref{sub: classification} for any given value of the tuning parameter $\lambda$ and let $\tilde g(\lambda) = (\tilde g_1(\lambda),\dots,\tilde g_N(\lambda))'\in\mathbb \{1,\dots,m(\lambda)\}^N$ be the corresponding vector of group assignments, so that $\hat g = \tilde g(\hat\lambda)$. 
Then observe that for any $\lambda\leq \bar C\sqrt T$, the vector $\tilde g(\lambda)$ has the following property:
\begin{equation}\label{eq: key property classification}
\text{for all }i,j=1,\dots,N,\text{ we have }\tilde g_i(\lambda)\neq\tilde g_j(\lambda)\text{ if }g_i\neq g_j,
\end{equation}
i.e. units from different groups can not be classified into the same group. To see why this is so, suppose to the contrary that for some $\lambda\leq \bar C\sqrt T$, there exist $i,j=1,\dots,N$ such that $g_i\neq g_j$ but $\tilde g_i(\lambda)=\tilde g_j(\lambda)$. In this case, we can consider the first misclassified unit, say unit $i_0$, in the Classification Algorithm. This unit satisfies the following inequality for some $\mathcal A\subset\{1,\dots,N\}$:
\begin{equation}\label{eq: simple implication classification}
\left\|\hat A_{i_0} - \frac{1}{|\mathcal A|}\sum_{l\in\mathcal A}\hat A_l\right\|\leq \lambda,
\end{equation}
where  all units $l\in\mathcal A$ are coming from the same group and $i_0$ is coming from another group, i.e. $g_{i_0}\neq \gamma = g_l$ for all $l\in\mathcal A$. But then, for any $j\in \mathcal A$,
\begin{align*}
\|\hat A_{i_0} - \hat A_j\|
& \leq \left\|\hat A_{i_0} - \frac{1}{|\mathcal A|}\sum_{l\in\mathcal A}\hat A_l \right\|  + \left\|\frac{1}{|\mathcal A|}\sum_{l\in\mathcal A}\hat A_l - \hat A_j\right\|\\
& \leq \lambda + \frac{1}{\mathcal A}\sum_{l\in\mathcal A}\|\hat A_l - \hat A_j\| \leq 2\bar C\sqrt T,
\end{align*}
where we used the triangle inequality as well as \eqref{eq: good separation 1} and \eqref{eq: simple implication classification}. This contradicts \eqref{eq: good separation 2}, and so \eqref{eq: key property classification} indeed holds for all $\lambda \leq \bar C\sqrt T$.

In turn, \eqref{eq: key property classification} implies that if $\tilde g_i(\lambda)\neq \tilde g_j(\lambda)$ for some $i,j=1,\dots,N$ with $g_i=g_j$, then $m(\lambda) > G$. Therefore, given than $m(\hat\lambda) = G$, it follows that if $\hat\lambda\leq \bar C\sqrt T$, then the vector $\hat g = \tilde g(\hat \lambda)$ has the following property: for all $i,j=1,\dots,N$, we have $\hat g_i=\hat g_j$ if and only if $g_i=g_j$. But we claim that $\hat\lambda$ indeed satisfies the inequality $\hat\lambda\leq\bar C\sqrt T$. To see why this is so, observe that $m(\bar C\sqrt T)\geq G$ by \eqref{eq: key property classification} and $m(\bar C\sqrt T)\leq G$ by \eqref{eq: good separation 1}. Hence, $m(\bar C\sqrt T)= G$, and so $\hat\lambda\leq \bar C\sqrt T$, as required. Thus, the asserted claim of the theorem follows since  \eqref{eq: good separation 1} and \eqref{eq: good separation 2} hold with probability approaching one. It thus remains to prove \eqref{eq: A estimator bound}. We do so in eight steps.

\smallskip
{\bf Step 1.} Here we show that 
$$
\|\hat F_{h_i}'v_i\| = o_P(\sqrt T)\quad\text{uniformly over }i=1,\dots,N.
$$
To do so, note that for all $i=1,\dots,N$, the random vectors $v_i$ and $\hat F_{h_i}$ are independent conditional on $h_1,\dots,h_N$.  Therefore, by Assumption \ref{as: spectral v}(i) and (2.14) in \cite{V18}, for all $\epsilon > 0$, $i=1,\dots,N$, and $\gamma = 1,\dots,G$,
$$
\mathrm{P}(|\hat F_{h_i \gamma}'v_i|>\epsilon)\leq 2\exp(-c\epsilon^2),
$$
where $\hat F_{h_i \gamma}$ denotes the $\gamma$-th column of the matrix $\hat F_{h_i}$. Hence, by the union bound,
$$
\mathrm{P}\left(\max_{1\leq i\leq N}|\hat F_{h_i \gamma }'v_i| > \epsilon\right)\leq 2\exp(\log N - c\epsilon ^2).
$$
Combining this bound with Assumption \ref{as: N and T} gives the asserted claim of this step.

\smallskip
{\bf Step 2.} Here we show that
$$
\|\hat F_{h_i}'x_i(\tilde\beta^{h_i}-\beta)\| = o_P(\sqrt T)\quad\text{uniformly over }i=1,\dots,N.
$$
To do so, for all $i = 1,\dots,N$, denote $z_i = (z_{i1},\dots,z_{iT})'$ and 
$$
\bar x_i = \left(\sum_{m=1}^M\rho_{im}\alpha_{g_i1}^m,\dots,\sum_{m=1}^M\rho_{im}\alpha_{g_iT}^m\right)',
$$
so that $x_i = \bar x_i +z_i$. Then observe that
$$
\|\hat F_{h_i}'z_i\| = o_P(\sqrt T)\quad\text{uniformly over }i=1,\dots,N
$$
by the same argument as that in Step 1, with Assumption \ref{as: spectral v}(ii) playing the role of Assumption \ref{as: spectral v}(i). Also, 
$$
\|\hat F_{h_i}'\bar x_{i}\| = O(\sqrt T)\quad\text{uniformly over }i=1,\dots,N
$$
by Assumption \ref{as: bounded rho}. Moreover, $\|\tilde\beta^0 - \beta\|\vee \|\tilde\beta^1-\beta\| = o_P(1)$ by Theorem \ref{thm: consistency spectral estimator}. Combining these bounds gives the asserted claim of this step.

\smallskip
{\bf Step 3.} Here we show that
$$
\left(\frac{1}{N}\sum_{i=1}^N v_{it}^2\right)\vee\left(\frac{1}{N}\sum_{i=1}^N \|z_{it}\|^2\right) = O_P(1)\quad\text{uniformly over }t=1,\dots,T.
$$
To do so, note that by Assumption \ref{as: spectral v}(i) and Lemma 2.7.6 in \cite{V18}, for all $i=1,\dots,N$ and $t=1,\dots,T$, we have $\|v_{it}^2\|_{\psi_1}\leq C$, and so by Exercise 2.7.10 in \cite{V18}, $\|v_{it}^2-\mathbb E[v_{it}^2]\|_{\psi_1}\leq C$. Hence, by Bernstein's inequality (Corollary 2.8.3 in \cite{V18}), for any $\epsilon > 0$,
$$
\Pr\left(\frac{1}{N}\sum_{i=1}^N (v_{it}^2 - \mathbb E[v_{it}^2]) > \epsilon \right)\leq \exp[-c(\epsilon\wedge\epsilon^2)N],
$$
and so by the union bound,
$$
\Pr\left(\max_{1\leq t\leq T}\frac{1}{N}\sum_{i=1}^N (v_{it}^2 - \mathbb E[v_{it}^2]) > \epsilon\right) \leq T\exp[-c(\epsilon\wedge\epsilon^2)N].
$$
Hence, given that $\mathbb E[v_{it}^2]\leq C$ for all $i=1,\dots,N$ and $t=1,\dots,T$ by Assumption \ref{as: spectral v}(i), it follows from Assumption \ref{as: N and T} that
$$
\frac{1}{N}\sum_{i=1}^N v_{it}^2 = O_P(1)\quad\text{uniformly over }t=1,\dots,T.
$$
Thus, given that
$$
\frac{1}{N}\sum_{i=1}^N \|z_{it}^2\| = O_P(1)\quad\text{uniformly over }t=1,\dots,T
$$
can be proven using the same argument, with Assumption \ref{as: spectral v}(i) replaced by \ref{as: spectral v}(ii), the asserted claim of this step follows.

\smallskip
{\bf Step 4.} Here we show that
$$
\|\hat B^0 - \bar B^0\|\vee\|\hat B^1-\bar B^1\| =o_P(1),
$$
where
$$
\bar B^h = \frac{2}{NT}\sum_{i\in\mathcal I_h}(y_i - x_i\beta)(y_i - x_i\beta)',\quad h=0,1.
$$
To do so, fix $h=0,1$. Then
\begin{align*}
\hat B^h - \bar B^h
& = \frac{2}{NT}\sum_{i\in\mathcal I_h}x_i(\tilde\beta^h - \beta)(\tilde\beta^h - \beta)'x_i'\\
&\quad -\frac{2}{NT}\sum_{i\in\mathcal I_h}(y_i - x_i\beta)(\tilde\beta^h - \beta)'x_i'\\
&\quad -\frac{2}{NT}\sum_{i\in\mathcal I_h}x_i(\tilde\beta^h - \beta)(y_i - x_i\beta)'.
\end{align*}
All three terms here have the spectral norm $o_P(1)$ but for brevity, we only consider the first term, i.e. we prove that
\begin{equation}\label{eq: three terms consider one}
\left\| \frac{1}{NT}\sum_{i\in\mathcal I_h}x_i(\tilde\beta^h - \beta)(\tilde\beta^h - \beta)'x_i' \right\| = o_P(1)
\end{equation}
and note that the other two terms can be bounded similarly. To prove \eqref{eq: three terms consider one}, recall that $x_i = \bar x_i + z_i$ as in Step 2. Therefore, by Lemma \ref{lem: cs matrices}, the left-hand side of \eqref{eq: three terms consider one} is bounded from above by
$$
2\left\| \frac{1}{NT}\sum_{i\in\mathcal I_h}\bar x_i(\tilde\beta^h - \beta)(\tilde\beta^h - \beta)'\bar x_i' \right\|
+ 2\left\| \frac{1}{NT}\sum_{i\in\mathcal I_h}z_i(\tilde\beta^h - \beta)(\tilde\beta^h - \beta)'z_i' \right\|.
$$
Here,
$$
\left\| \frac{1}{NT}\sum_{i\in\mathcal I_h}\bar x_i(\tilde\beta^h - \beta)(\tilde\beta^h - \beta)'\bar x_i' \right\|
\leq \frac{1}{NT}\sum_{i\in\mathcal I_h}\|\bar x_i\|^2\|\tilde\beta^h - \beta\|^2 = o_P(1)
$$
by Theorem \ref{thm: consistency spectral estimator} and Assumption \ref{as: bounded rho}. Also, given that $(NT)^{-1}\sum_{i\in\mathcal I_h}z_i(\tilde\beta^h-\beta)(\tilde\beta^h-\beta)'z_i'$ is a positive-definite matrix,
\begin{align*}
&\left\| \frac{1}{NT}\sum_{i\in\mathcal I_h}z_i(\tilde\beta^h - \beta)(\tilde\beta^h - \beta)'z_i' \right\|\\
&\qquad \leq \textrm{tr}\left\{\frac{1}{NT}\sum_{i\in\mathcal I_h}z_i(\tilde\beta^h - \beta)(\tilde\beta^h - \beta)'z_i'   \right\}
 = \frac{1}{NT}\sum_{i\in\mathcal I_h}\textrm{tr}\left\{ (\tilde\beta^h - \beta)'z_i'z_i(\tilde\beta^h - \beta) \right\}\\
&\qquad = \frac{1}{NT}\sum_{i\in\mathcal I_h}\sum_{t=1}^T |z_{it}'(\tilde\beta^h - \beta)|^2
 \leq \frac{1}{NT}\sum_{i\in\mathcal I_h}\sum_{t=1}^T \|z_{it}\|^2\|\tilde\beta^h - \beta\|^2 = o_P(1)
\end{align*}
by Step 3 and Theorem \ref{thm: consistency spectral estimator}. Combining presented bounds, we obtain the asserted claim of this step.

\smallskip
{\bf Step 5.} Here we show that
$$
\|\hat B^0 - F\Lambda F'\|\vee \|\hat B^1 - F\Lambda F'\|=o_P(1),
$$
To do so, fix $h=0,1$. Then, denoting
$$
\check B = \frac{1}{NT}\sum_{i=1}^N(y_i - x_i\beta)(y_i - x_i\beta)',
$$
we have
\begin{equation}\label{eq: one more approximation classification argument}
\|\bar B^h - \check B\| = o_P(1),
\end{equation}
where the matrix $\bar B^h$ is defined in the previous step. To see why this is so, observe that
$$
\bar B^h - \check B = \frac{2}{NT}\sum_{i=1}^N \Big(1\{h_i = 1-h\} - \mathbb E[1\{h_i = 1-h\}] \Big)(y_i - x_i\beta)(y_i - x_i\beta)'.
$$
Thus, recalling that $y_i = x_i\beta + \alpha_{g_i} + v_i$, we have
\begin{align*}
\|\bar B^h-\check B\| & \leq \left\|\frac{2}{NT}\sum_{i=1}^N \Big(1\{h_i = 1-h\} - \mathbb E[1\{h_i = 1-h\}] \Big) \alpha_{g_i}\alpha_{g_i}' \right\|\\
&\quad + \left\|\frac{2}{NT}\sum_{i=1}^N \Big(1\{h_i = 1-h\} - \mathbb E[1\{h_i = 1-h\}] \Big) v_{i}v_{i}' \right\|\\
&\quad + \left\|\frac{4}{NT}\sum_{i=1}^N \Big(1\{h_i = 1-h\} - \mathbb E[1\{h_i = 1-h\}] \Big) \alpha_{g_i}v_{i}' \right\|.
\end{align*}
Here, $\|\alpha_{g_i}\alpha_{g_i}'\| = \|\alpha_{g_i}\|^2 \leq CT$ for all $i=1,\dots,N$, and so applying the expectation version of Bernstein's matrix inequality (Exercise 5.4.11 in \cite{V18}) conditionally on $(x_1,y_1),\dots,(x_N,y_N)$,
\begin{multline*}
\mathbb E\left[ \left\| \frac{1}{NT}\sum_{i=1}^N \Big(1\{h_i = 1-h\} - \mathbb E[1\{h_i = 1-h\}] \Big) \alpha_{g_i}\alpha_{g_i}'  \right\| \right] \\
\leq C\left(\frac{\sqrt{NT^2\log T}}{NT} + \frac{T\log T}{NT}\right) = o(1),
\end{multline*}
where the last bound follows from Assumption \ref{as: N and T}. Hence,
$$
\left\| \frac{1}{NT}\sum_{i=1}^N \Big(1\{h_i = 1-h\} - \mathbb E[1\{h_i = 1-h\}] \Big) \alpha_{g_i}\alpha_{g_i}'  \right\| = o_P(1)
$$
by Markov's inequality. In addition,
$$
\left\|\frac{1}{NT}\sum_{i=1}^N \Big(1\{h_i = 1-h\} - \mathbb E[1\{h_i = 1-h\}] \Big) v_{i}v_{i}' \right\|
\leq \left\|\frac{1}{NT}\sum_{i=1}^N v_i v_i'\right\| = o_P(1)
$$
by Lemma \ref{lem: cs matrices} and \eqref{eq: useful bound from proof of spectral consistency} in the proof of Theorem \ref{thm: consistency spectral estimator}. Moreover, by Lemma \ref{lem: cs matrices},
\begin{multline*}
\left\|\frac{1}{NT}\sum_{i=1}^N \Big(1\{h_i = 1-h\} - \mathbb E[1\{h_i = 1-h\}] \Big) \alpha_{g_i}v_{i}' \right\|\\
\leq \sqrt{\left\|\frac{1}{NT}\sum_{i=1}^N \alpha_{g_i}\alpha_{g_i}'\right\|}\sqrt{\left\|\frac{1}{NT}\sum_{i=1}^N v_{i}v_{i}'\right\|} = o_P(1).
\end{multline*}
Combining these bounds, we obtain \eqref{eq: one more approximation classification argument}.

Now, given that $\check B = (NT)^{-1}\sum_{i=1}^N (\alpha_{g_i}+v_i)(\alpha_{g_i}+v_i)'$, $\alpha_{g_i} =\sqrt T Fp_{g_i}$, and $N^{-1}\sum_{i=1}^N p_{g_i}p_{g_i}' = \Lambda$, we have
$$
\|\check B - F\Lambda F'\| \leq \left\|\frac{2}{NT}\sum_{i=1}^N \alpha_{g_i}v_i'\right\| + \left\|\frac{1}{NT}\sum_{i=1}^N v_iv_i'\right\| = o_P(1)
$$
by the arguments above. Combining this bound with \eqref{eq: one more approximation classification argument} and Step 4 and using the triangle inequality gives the asserted claim of this step.

\smallskip
{\bf Step 6.} Here we show that there exist orthogonal $G\times G$ matrices $O_0$ and $O_1$ such that for all $\gamma=1,\dots,G$ and $h=0,1$,
$$
\|(\hat F_{h}O_h-F)'F p_{\gamma}\| = o_P(1).
$$
To do so, fix $h=0,1$ and note that by Step 5, there exists a sequence $\{\psi_n\}_{n\geq 1}$ of positive numbers such that $\psi_n\to0$ as $n\to\infty$ and
\begin{equation}\label{eq: matrix estimation psi}
\mathrm P(\|\hat B^h - F\Lambda F'\|>\psi_n)\leq \psi_n.
\end{equation}
Also, recall that $\Lambda = \textrm{diag}(\lambda_1,\dots,\lambda_G)$, where $\lambda_1\geq\dots\lambda_G\geq 0$. In addition, set $\lambda_0 = +\infty$ and $\lambda_{G+1} = 0$ and let $\bar G$ be the largest integer $\gamma =0,\dots,G$ such that $\lambda_{\gamma} - \lambda_{\gamma +1} > \sqrt{\psi_n}$. Then $\lambda_{\bar G+1}\leq G\sqrt{\psi_n}$. In addition, by the Davis-Kahan theorem (see Theorem 2 in \cite{YWS15}), there exists an orthogonal $G\times G$ matrix $O_h$ such that, denoting $\tilde F = \hat F_h O_h$ and letting $\tilde F_{\gamma}$ and $F_{\gamma}$ denote the $\gamma$-th columns of $\tilde F$ and $F$, respectively, we have
$$
\max_{1\leq \gamma\leq \bar G}\|\tilde F_{\gamma} - F_{\gamma}\| \leq 4\sqrt G\|\hat B^h - F\Lambda F'\| / \sqrt{\psi_n},
$$
and so by \eqref{eq: matrix estimation psi},
$$
\max_{1\leq \gamma\leq \bar G}\|\tilde F_{\gamma} - F_{\gamma}\| \leq 4\sqrt G \sqrt{\psi_n}
$$
with probability at least $1 - \psi_n$. Therefore,
\begin{align*}
&\frac{1}{N}\sum_{i=1}^N \|(\tilde F-F)'F p_{g_i}\|^2
 = \mathrm{tr}\left(\frac{1}{N}\sum_{i=1}^N(\tilde F-F)'F p_{g_i}p_{g_i}'F' (\tilde F-F)\right)\\
&\qquad  = \mathrm{tr}\left( (\tilde F-F)'F\Lambda F' (\tilde F-F) \right)
 = \textrm{tr}\left( (\tilde F - F)'\sum_{\gamma_1=1}^G \lambda_{\gamma_1}F_{\gamma_1}F_{\gamma_1}' (\tilde F - F) \right) \\
 &\qquad  =  \sum_{\gamma_1,\gamma_2=1}^G \lambda_{\gamma_1}\{F_{\gamma_1}'(\tilde F_{\gamma_2} - F_{\gamma_2})\}^2 \leq \sum_{\gamma_1=1}^{\bar G}\sum_{\gamma_2=1}^G \lambda_{\gamma_1}\{F_{\gamma_1}'(\tilde F_{\gamma_2} - F_{\gamma_2})\}^2 + G^3\sqrt{\psi_n}\\
&\qquad \leq \sum_{\gamma_1=1}^{\bar G} \lambda_{\gamma_1}\{F_{\gamma_1}'(\tilde F_{\gamma_1} - F_{\gamma_1})\}^2 + \sum_{\gamma_1=1}^{\bar G}\sum_{\gamma_2\neq \gamma_1}\lambda_{\gamma_1}\{F_{\gamma_1}'\tilde F_{\gamma_2} \}^2 + G^3\sqrt{\psi_n} = o_P(1),
\end{align*}
where the middle term is bounded by Lemma \ref{lem: orthogonality of two bases} and all $\lambda_{\gamma}$ are bounded by
$$
\max_{1\leq\gamma\leq G}\lambda_{\gamma} = \|\Lambda\| = \|F\Lambda F'\| =\left\|\frac{1}{NT}\sum_{i=1}^N \alpha_{g_i}\alpha_{g_i}'\right\|\leq C,
$$
where the last inequality follows from Assumption \ref{as: bounded rho}(ii). Combining this bound with Assumption \ref{as: group sizes new}, we conclude that for all $\gamma=1,\dots,G$,
\begin{align*}
\|(\hat F_{h}O_h-F)'F p_{\gamma}\|^2 & \leq \frac{C}{N}\sum_{i=1}^N \|(\hat F_{h}O_h-F)'F p_{g_i}\|^2\\
& = \frac{C}{N}\sum_{i=1}^N \|(\tilde F-F)'F p_{g_i}\|^2 = o_P(1),
\end{align*}
which gives the asserted claim of this step.

\smallskip
{\bf Step 7.} Here we show that for all $\gamma = 1,\dots,G$ and $h=0,1$,
$$
\|(\hat F_h O_h - F)p_{\gamma}\| = o_P(1).
$$
To do so, fix $h=0,1$ and note that
\begin{align*}
\frac{1}{N}\sum_{i=1}^N \|(\hat F_h O_h - F)p_{g_i}\|^2
& = \mathrm{tr}\left(\frac{1}{N}\sum_{i=1}^N (\hat F_h O_h - F)p_{g_i}p_{g_i}' (\hat F_h O_h - F)' \right)\\
& = \mathrm{tr}\Big( (\hat F_h O_h - F)\Lambda (\hat F_h O_h - F)' \Big)\\
& = \mathrm{tr}\Big( \Lambda (\hat F_h O_h - F)'(\hat F_h O_h - F)  \Big)\\
& = \sum_{\gamma=1}^G \lambda_{\gamma} \|(\hat F_h O_h - F)_{\gamma}\|^2 = o_P(1)
\end{align*}
by the same argument as that used in Step 6. Combining this bound with Assumption \ref{as: group sizes new}, we conclude that for all $\gamma=1,\dots,G$,
$$
\|(\hat F_h O_h - F)p_{\gamma}\|^2 \leq \frac{C}{N}\sum_{i=1}^N \|(\hat F_{h}O_h-F) p_{g_i}\|^2 = o_P(1),
$$
which gives the asserted claim of this step.

\smallskip
{\bf Step 8.} Here we complete the proof. To do so, we have
\begin{align*}
\|\hat A_i - \alpha_{g_i}\|
& = \| \hat F_{h_i}\hat F_{h_i}'(y_i -x_i\tilde\beta^{h_i}) -\sqrt T Fp_{g_i} \|\\
& = \sqrt T \| \hat F_{h_i}\hat F_{h_i}'Fp_{g_i} - Fp_{g_i} \| + o_P(\sqrt T)\\
& = \sqrt T \| \hat F_{h_i}O_{h_i} (\hat F_{h_i} O_{h_i})'Fp_{g_i} - Fp_{g_i} \| + o_P(\sqrt T)\\
& = \sqrt T \| \hat F_{h_i}O_{h_i} p_{g_i} - F p_{g_i} \| + o_P(\sqrt T) = o_P(\sqrt T)
\end{align*}
uniformly over $i=1,\dots,N$, where the second line follows from Steps 1 and 2, the third from $O_{h_i}$ being orthogonal, and the fourth from Steps 6 and 7. This gives the asserted claim of this step and completes the proof of the theorem.
\end{proof}

\begin{proof}[Proof of Theorem \ref{thm: main consistency}]
Throughout the proof, for any $\mathrm a\in\mathbb R^T$ and $R>0$, we use $B(\mathrm a,R)$ to denote the ball in $\mathbb R^T$ with center $\mathrm a$ and radius $R$, i.e. $B(\mathrm a,R) = \{u\in\mathbb R^T\colon \|u-\mathrm a\|\leq R\}$. Also, we use $c$ and $C$ to denote strictly positive constants that can change from place to place but can be chosen to depend on $C_1$, $C_2$, $C_3$, $C_4$, $C_5$, $c_1$, $c_2$, and $c_3$ only.

We proceed in five steps. The second and third steps follow closely the arguments in the proof of Theorem 1 in \cite{BM15} but we provide all the details for reader's convenience.

\smallskip
{\bf Step 1.} Here we show that
\begin{equation}\label{eq: thm 3 step 0}
\max_{1\leq i\leq N}\mathbb E\left[\sum_{t=1}^T v_{it}^2\right]\leq CT\quad\text{and}\quad\max_{1\leq i\leq N}\max_{j\neq i}\mathbb E\left[\left|\sum_{t=1}^T v_{it}v_{jt}\right|\right]\leq C\sqrt T.
\end{equation}
To prove the first inequality, note that for all $i=1,\dots,N$, we have $\mathbb E[v_{it}^2] \leq C\|v_{it}\|_{\psi_2}^2\leq C$ by (2.15) in \cite{V18} and Assumption \ref{as: spectral v}(i). To prove the second inequality note that for all $i,j=1,\dots,N$ with $i\neq j$, we have
$$
\mathrm P\left(\left|\sum_{t=1}^T v_{it}v_{jt}\right| > \epsilon \mid v_{j1},\dots,v_{jT}\right)\leq 2\exp\left(-\frac{c\epsilon^2}{\sum_{t=1}^T v_{jt}^2}\right)
$$
for all $\epsilon > 0$ by (2.14) in \cite{V18} and Assumption \ref{as: spectral v}(i). Therefore, given that the function $r\mapsto \exp(-c/r)$ is concave on $\mathbb R_{+}$ for any $c>0$, we have
$$
\mathrm P\left(\left|\sum_{t=1}^T v_{it}v_{jt}\right| > \epsilon\right)\leq 2\exp\left(-\frac{c\epsilon^2}{\sum_{t=1}^T \mathbb E[v_{jt}^2]}\right)\leq 2\exp\left(-\frac{c\epsilon^2}{T}\right)
$$
by Jensen's inequality and the first inequality in \eqref{eq: thm 3 step 0}. Hence,
$$
\mathbb E\left[\left|\sum_{t=1}^T v_{it}v_{jt}\right|\right] =\int_0^{\infty}\mathrm P\left(\left|\sum_{t=1}^T v_{it}v_{jt}\right| > \epsilon\right)d\epsilon \leq C\sqrt T,
$$
which gives the second inequality in \eqref{eq: thm 3 step 0}.

\smallskip
{\bf Step 2.} For all $b\in\mathcal B$, $a=\{a_{\gamma t}\}_{\gamma,t=1}^{G,T}\in\mathcal A_{G,T}$, and $\nu=\{\nu_i\}_{i=1}^N\in\{1,\dots,G\}^N$, denote
$$
Q(b,a,\nu) = \frac{1}{NT}\sum_{i=1}^N\sum_{t=1}^T(y_{it} - x_{it}'b - a_{\nu_i t})^2
$$
and
$$
\bar Q(b,a,\nu) = \frac{1}{NT}\sum_{i=1}^N\sum_{t=1}^T(x_{it}'(\beta - b) + \alpha_{g_i t} - a_{\nu_i t})^2 + \frac{1}{NT}\sum_{i=1}^N\sum_{t=1}^T v_{it}^2.
$$
In this step, we show that
$$
Q(b,a,\nu) - \bar Q(b,a,\nu) = o_P(1)
$$
uniformly over $b\in\mathcal B$, $a\in\mathcal A_{G,T}$, and $\nu \in\{1,\dots,G\}^N$. To do so, note that
$$
Q(b,a,\nu) - \bar Q(b,a,\nu) = -\frac{2}{NT}\sum_{i=1}^N\sum_{t=1}^T v_{it}\Big( x_{it}'(b-\beta) + a_{\nu_i t} - \alpha_{g_i t}\Big).
$$
Here, we have
\begin{align*}
\mathbb E\left[\left(\frac{1}{NT}\sum_{i=1}^N\sum_{t=1}^T v_{it}\alpha_{g_it}\right)^2\right] 
&= \frac{1}{(NT)^2}\sum_{i=1}^N\mathbb E\left[\left( \sum_{t=1}^T v_{it}\alpha_{g_it} \right)^2\right]\\
& \leq \frac{C}{(NT)^2} \sum_{i=1}^N\sum_{t=1}^T \alpha_{g_it}^2 \leq \frac{C}{NT}\to0,
\end{align*}
where the first inequality follows from Assumption \ref{as: spectral v}(i) and (2.15) in \cite{V18} and the second from Assumption \ref{as: bounded rho}(ii). Thus, by Markov's inequality,
$$
\frac{1}{NT}\sum_{i=1}^N\sum_{t=1}^T v_{it}\alpha_{g_it} = o_P(1).
$$
Further, denoting $\bar x_{it} = \sum_{m=1}^M \rho_{im}\alpha_{g_it}^m$, so that $x_{it} = \bar x_{it} + z_{it}$, we have
$$
\left\| \frac{1}{NT}\sum_{i=1}^N\sum_{t=1}^T v_{it}\bar x_{it}'(b-\beta) \right\| 
\leq C\left\| \frac{1}{NT}\sum_{i=1}^N\sum_{t=1}^T v_{it}\bar x_{it} \right\|  = o_P(1)
$$
uniformly over $b\in\mathcal B$ by the same argument as that we have just used and Assumptions \ref{as: bounded rho} and \ref{as: compact}. Moreover,
$$
\left\| \frac{1}{NT}\sum_{i=1}^N\sum_{t=1}^T v_{it}z_{it}'(b-\beta) \right\| 
\leq C\left\| \frac{1}{NT}\sum_{i=1}^N\sum_{t=1}^T v_{it}z_{it} \right\|  = o_P(1)
$$
uniformly over $b\in\mathcal B$ by Assumptions \ref{as: spectral vz}(i) and \ref{as: compact}.
In addition,
\begin{align*}
& \left|\frac{1}{NT}\sum_{i=1}^N\sum_{t=1}^T v_{it}a_{\nu_i t}\right|
 = \left|\frac{1}{NT}\sum_{\gamma = 1}^G\sum_{t=1}^Ta_{\gamma t}\sum_{i\colon \nu_i = \gamma} v_{it}\right|\\
& \qquad \leq \frac{1}{NT}\sum_{\gamma = 1}^G\sqrt{\sum_{t=1}^T a_{\gamma t}^2}\sqrt{\sum_{t=1}^T \left(\sum_{i\colon \nu_i = \gamma} v_{it}\right)^2} 
\leq \frac{C}{N\sqrt T}\sum_{\gamma = 1}^G\sqrt{\sum_{t=1}^T\left(\sum_{i\colon \nu_i = \gamma} v_{it}\right)^2} \\
& \qquad \leq \frac{C}{N\sqrt T}\sum_{\gamma = 1}^G\sqrt{\sum_{i,j = 1}^N\left|\sum_{t=1}^T v_{it} v_{jt}\right|},
\end{align*}
where the second line follows from the Cauchy-Schwarz inequality and Assumption \ref{as: compact}. In turn,
$$
\mathbb E\left[\sum_{i,j = 1}^N\left|\sum_{t=1}^T v_{it} v_{jt}\right|\right] \leq C(N^2 \sqrt{T} + NT)
$$
by Step 1. Hence, by Markov's inequality,
$$
\left|\frac{1}{NT}\sum_{i=1}^N\sum_{t=1}^Tv_{it}a_{\nu_i t}\right| = o_P(1)
$$
uniformly over $a\in\mathcal A_{G,T}$ and $\nu\in\{1,\dots,G\}^N$ since $T\to\infty$ as $N\to\infty$. Combining presented bounds gives the asserted claim of this step.

\smallskip
{\bf Step 3.} Here we show that
$$
\bar Q(b,a,\nu) - \bar Q(\beta,\alpha,g) \geq c_3\|b-\beta\|^2
$$ 
for all $b\in\mathcal B$, $a\in\mathcal A_{G,T}$, and $\nu\in\{1,\dots,G\}^N$ with probability $1 - o(1)$, where $\alpha = \{\alpha_{\gamma t}\}_{\gamma,t=1}^{G,T}$. To do so, note that
\begin{align*}
\bar Q(b,a,\nu) - \bar Q(\beta,\alpha,g)
& = \frac{1}{NT}\sum_{i=1}^N\sum_{t=1}^T\Big( x_{it}'(b-\beta) + a_{\nu_i t} - \alpha_{g_it} \Big)^2\\
&\geq \min_{\tilde a\in\mathcal A_{G,T}}\frac{1}{NT}\sum_{i=1}^N\sum_{t=1}^T\Big( x_{it}'(b-\beta) + \tilde a_{\nu_i t} - \alpha_{g_it} \Big)^2\\
&\geq \frac{1}{NT}\sum_{i=1}^N\sum_{t=1}^T\Big(( x_{it} - \bar x_{\nu,\nu_i,g_i,t})'(b-\beta)\Big)^2 \geq c_3\|b-\beta\|^2
\end{align*}
for all $b\in\mathcal B$, $a\in\mathcal A_{G,T}$, and $\nu\in\{1,\dots,G\}^N$ with probability $1 - o(1)$ by Assumption \ref{as: eigenvalues}. The asserted claim of this step follows.

\smallskip
{\bf Step 4.} Denote
$$
\mathcal G_{\gamma} = \Big\{i=1,\dots,N\colon \hat g_i =\gamma \Big\},\quad\text{for all }\gamma = 1,\dots,G
$$
and
$$
\check{\alpha}_{\gamma t} = \frac{1}{|\mathcal G_{\gamma}|}\sum_{i\in\mathcal G_{\gamma}}\alpha_{g_i t},\quad\text{for all }\gamma = 1,\dots,G, \ t = 1,\dots,T.
$$
In this step, we show that
\begin{equation}\label{g-alpha approximation}
\frac{1}{NT}\sum_{i=1}^N\sum_{t=1}^T(\check{\alpha}_{\hat g_i t} - \alpha_{g_i t})^2 = o_P(1).
\end{equation}
To do so, note that by \eqref{eq: A estimator bound} in the proof of Theorem \ref{thm: classifier consistency new}, there exist sequences $\{\lambda_N\}_{N\geq 1}$ and $\{\delta_N\}_{N\geq 1}$ of positive numbers such that $\lambda_N/\sqrt T\to0$ and $\delta_N\to0$ as $N\to\infty$ and the event
\begin{equation}\label{eq: proof thm 3 step 3 start}
\max_{1\leq i\leq N}|\hat A_i - \alpha_{g_i}| \leq \lambda_N
\end{equation}
holds with probability at least $1-\delta_N$ (note that the proof of \eqref{eq: A estimator bound} did not use Assumption \ref{as: well separated groups}, which is not imposed here). We claim that on the event \eqref{eq: proof thm 3 step 3 start}, $\hat\lambda\leq C_G\lambda_N$, where $C_G>0$ is a constant depending only on $G$. To see why this is so, suppose that \eqref{eq: proof thm 3 step 3 start} is satisfied and consider the subset $\mathrm A_0$ of the set $\{\alpha_1,\dots,\alpha_G\}$ and a constant $R_0>0$ defined by the following algorithm:

\medskip
{\em Step 1:} set $\gamma=1$, $\mathrm A_{1} = \{\alpha_1,\dots,\alpha_G\}$, and $R_1 = \lambda_N$;

{\em Step 2:} if
\begin{equation}\label{eq: exploding balls}
\Big\{ \cup_{\mathrm a\in\mathrm A_{\gamma}}\big(B(\mathrm a,8R_{\gamma} + 6\lambda_N)\setminus B(\mathrm a,R_{\gamma})\big) \Big\}\cap \mathrm A_{\gamma}=\emptyset,
\end{equation}

\ \ \ \ \  \ \  \  \  \  \  then set $\mathrm A_0 = \mathrm A_{\gamma}$ and $R_0 = R_{\gamma}$ and stop;

{\em Step 3:} replace $\gamma$ by $\gamma + 1$ and set $R_{\gamma} = 9R_{\gamma-1} + 6\lambda_N$;

{\em Step 4:} let $\mathrm A_{\gamma}$ be the smallest subset of $A_{\gamma-1}$ such that
$$
\{\alpha_1,\dots,\alpha_G\}\subset \cup_{\mathrm a\in\mathrm A_{\gamma}} B(\mathrm a,R_{\gamma});\footnote{If there are several smallest sets, choose one of them at random.}
$$

{\em Step 5:} go to step 2.

\medskip
\noindent
Observe that on Step 4 of this algorithm, we have $|\mathrm A_{\gamma}|\leq |\mathrm A_{\gamma-1}| - 1$. Indeed, if Step 4 is performed for given $\gamma$, it follows from Step 2 that there exist $\mathrm a_1,\mathrm a_2\in\mathrm A_{\gamma-1}$ such that $\mathrm a_2\in B(\mathrm a_1,8R_{\gamma-1} + 6\lambda_N)$, and so $B(\mathrm a_2,R_{\gamma-1})\subset B(\mathrm a_1,9R_{\gamma-1} + 6\lambda_N) = B(\mathrm a_1,R_{\gamma})$, letting us drop $\mathrm a_2$ from $\mathrm A_{\gamma-1}$ while constructing $\mathrm A_{\gamma}$ and yielding $|\mathrm A_{\gamma}|\leq |\mathrm A_{\gamma-1}| - 1$. In turn, the latter implies that the algorithm will stop in a finite number of steps and, in fact, Step 3 will be performed at most $G-1$ times. Hence, $R_0$ satisfies $R_0\leq \bar C_G\lambda_N$, where $\bar C_G>0$ is a constant depending only on $G$. In addition,
\begin{equation}\label{eq: covering balls}
\{\alpha_1,\dots,\alpha_G\}\subset \cup_{\mathrm a\in\mathrm A_{0}} B(\mathrm a,R_{0})
\end{equation}
by construction. 

Further, since \eqref{eq: exploding balls} is satisfied with $\gamma = 0$ by construction, it follows that
\begin{equation}\label{eq: separated balls}
\mathrm a_1,\mathrm a_2\in\mathrm A_0,\text{ implies that }\|\mathrm a_2 - \mathrm a_1\|\leq R_0\text{ or }\|\mathrm a_2 - \mathrm a_1\|>8R_0 + 6\lambda_N. 
 \end{equation}
 This allows us to partition $\mathrm A_0$ into equivalence subclasses as follows. Say that $\mathrm a_1\sim \mathrm a_2$ if $\|\mathrm a_2 - \mathrm a_1\|\leq R_0$. Then for any $\mathrm a_1,\mathrm a_2,\mathrm a_3\in\mathrm A_0$, we have that $\mathrm a_1\sim \mathrm a_2$ and $\mathrm a_2\sim\mathrm a_3$ imply $\mathrm a_1 \sim \mathrm a_3$, meaning that the relation $\sim$ is actually an equivalence relation. Thus, we can partition $\mathrm A_0$ into $k\leq G$ equivalence subclasses $\mathrm A_{0,1},\dots,\mathrm A_{0,k}$ such that for any $\mathrm a_1,\mathrm a_2\in\mathrm A_0$ we have $\mathrm a_1\sim\mathrm a_2$ if and only if $\mathrm a_1$ and $\mathrm a_2$ belong to the same subclass. Then it follows from \eqref{eq: proof thm 3 step 3 start},  \eqref{eq: covering balls}, and \eqref{eq: separated balls} that for each $i=1,\dots,N$, there exists a unique $\gamma(i)\in\{1,\dots,k\}$ such that $\|\hat A_i - \mathrm a\|\leq R_0 + \lambda_N$ for some $\mathrm a \in \mathrm A_{0,\gamma(i)}$. Thus, for any $i_1,i_2=1,\dots,N$, we have that $\|\hat A_{i_1} - \hat A_{i_2}\|\leq 3R_0 + 2\lambda_N$ if $\gamma(i_1) = \gamma(i_2)$ and that $\|\hat A_{i_1} - \hat A_{i_2}\| > 6R_0 + 4\lambda_N$ if $\gamma(i_1) \neq \gamma(i_2)$. In turn, the latter implies that if we run the Classification Algorithm from Section \ref{sec: estimator} with $\lambda = 3R_0 + 2\lambda_N$, we obtain $m(\lambda) = k$ groups $\mathcal A_{1},\dots,\mathcal A_k$ such that any two units $i_1,i_2 = 1,\dots,N$ are classified to the same group if and only if $\gamma(i_1)=\gamma(i_2)$. To see why this is so, suppose that for some $\gamma = 1,\dots,k$, we have two units $i_1,i_2$ that are classified to the same group $\mathcal A_{\gamma}$ but are such that $\gamma(i_1)\neq \gamma(i_2)$. For this $\gamma$, let $i_1\to i_2 \to\dots\to i_{|\mathcal A_{\gamma}|}$ be the order in which units are added to $\mathcal A_{\gamma}$ by the Classification Algorithm, and let $r$ be the smallest number in the set $\{2,\dots,|\mathcal A_{\gamma}|\}$ such that $\gamma(i_r)\neq \gamma(i_{r-1})$. Then $\|\hat A_{i_r} - \hat A_{i_1}\| > 6R_0 + 4\lambda_N$ and $\|\hat A_{i_l} - \hat A_{i_1}\|\leq 3R_0 + 2\lambda_N$ for all $l=1,\dots,r-1$. This implies that
$$
\left\| \hat A_{i_r} - \frac{1}{r-1}\sum_{l=1}^{r-1}\hat A_{i_l} \right\| > (6R_0 + 4\lambda_N) - (3R_0 + 2\lambda_N) = 3R_0 + 2\lambda_N,
$$
yielding a contradiction. Now suppose that there are two units $i_1,i_2$ that are classified to different groups $\mathcal A_{\gamma_1}$ and $\mathcal A_{\gamma_2}$ but are such that $\gamma(i_1)=\gamma(i_2)$. For these $\gamma_1$ and $\gamma_2$, let $i_1\to i_2 \to \dots\to i_{|\mathcal A_{\gamma_1}| + |\mathcal A_{\gamma_2}|}$ be the order in which units are added to $\mathcal A_{\gamma_1}$ and $\mathcal A_{\gamma_2}$ by the Classification Algorithm. Assume, without loss of generality, that the first unit, $i_1$, is added to $\mathcal A_{\gamma_1}$, and let $r$ be the smallest number in the set $\{2,\dots,|\mathcal A_{\gamma_1}| + |\mathcal A_{\gamma_2}|\}$ such that the unit $i_r$ is added to the set $\mathcal A_{\gamma_2}$. Then $\|\hat A_{i_r} - \hat A_{i_l}\|\leq 3R_0 + 2\lambda_N$ for all $l=1,\dots,r-1$, and so
$$
\left\| \hat A_{i_r} - \frac{1}{r-1}\sum_{l=1}^{r-1}\hat A_{i_l} \right\| \leq 3R_0 + 2\lambda_N,
$$
yielding a contradiction, and finishing the proof of our claim that the Classification Algorithm from Section \ref{sec: estimator} with $\lambda = 3R_0 + 2\lambda_N$ yields $m(\lambda) = k$ groups $\mathcal A_{1},\dots,\mathcal A_k$ such that any two units $i_1,i_2 = 1,\dots,N$ are classified to the same group if and only if $\gamma(i_1)=\gamma(i_2)$. The latter than implies that $m(3R_0 +2\lambda_N) = k \leq G$, and so $\hat\lambda \leq 3R_0 + 2\lambda_N \leq C_G\lambda_N$ for some $C_G>0$ depending only on $G$, as desired.

Next, we claim that \eqref{eq: proof thm 3 step 3 start} implies even more: it implies that there exists a constant $\tilde C_G>0$ depending only on $G$ such that for any $\gamma=1,\dots,G$ and any $i_1,i_2\in\mathcal G_{\gamma}$, we have 
\begin{equation}\label{eq: thm 3 step 3 crucial}
\|\alpha_{g_{i_2}} - \alpha_{g_{i_1}}\|\leq \tilde C_G\lambda_N.
\end{equation}
To see why this is so, suppose again that \eqref{eq: proof thm 3 step 3 start} is satisfied. As we have proven above, it is then not possible that the Classification Algorithm with $\lambda = \hat\lambda \leq 3R_0 + 2\lambda_N$  generates groups $\mathcal A_{1},\dots,\mathcal A_{m(\lambda)}$ such that there are two units $i_1,i_2$ the are classified to the same group $\mathcal A_{\gamma}$ but are such that $\gamma(i_1)\neq \gamma(i_2)$. In turn, for any $i_1,i_2$ such that $\gamma(i_1)=\gamma(i_2)$, we have $\|\hat A_{i_2} - \hat A_{i_1}\| \leq 3R_0 + 2\lambda_N$, and so $\|\alpha_{g_{i_2}} - \alpha_{g_{i_1}}\| \leq 3R_0 + 4\lambda_N \leq \tilde C_G\lambda_N$ for some $\tilde C_G>0$, as desired.

We are now ready to finish this step. On \eqref{eq: proof thm 3 step 3 start}, by the triangle inequality and \eqref{eq: thm 3 step 3 crucial}, we have
\begin{align*}
\left(\frac{1}{T}\sum_{t=1}^T(\check\alpha_{\hat g_it} - \alpha_{g_i t})^2\right)^{1/2}
&\leq \frac{1}{|\mathcal G_{\hat g_i}|}\sum_{j\in\mathcal G_{\hat g_i}}\left(\frac{1}{T}\sum_{t=1}^T(\alpha_{g_jt}-\alpha_{g_it})^2\right)^{1/2}\\
& = \frac{1}{\sqrt T|\mathcal G_{\hat g_i}|}\sum_{j\in\mathcal G_{\hat g_i}}\|\alpha_{g_j}-\alpha_{g_i}\|
 \leq \frac{1}{\sqrt T|\mathcal G_{\hat g_i}|}\sum_{j\in\mathcal G_{\hat g_i}}\tilde C_G \lambda_N = o(1)
\end{align*}
uniformly over $i=1,\dots,N$. This gives the asserted claim of this step because \eqref{eq: proof thm 3 step 3 start} holds with probability approaching one.

\smallskip
{\bf Step 5.} Here we finish the proof. We have
\begin{align*}
\bar Q(\hat\beta,\hat\alpha,\hat g) &
= Q(\hat\beta,\hat\alpha,\hat g) + o_P(1) \\
& \leq Q(\beta,\check\alpha,\hat g) + o_P(1) = \bar Q(\beta,\check\alpha,\hat g) + o_P(1),
\end{align*}
where the equalities follow from Step 2 and the inequality from \eqref{eq: pooled ols}. Thus, for some constant $c>0$,
\begin{align*}
c\|\hat\beta - \beta\|^2 & 
\leq \bar Q(\hat\beta,\hat\alpha,\hat g) - \bar Q(\beta,\alpha,g) \\
& = \bar Q(\hat\beta,\hat\alpha,\hat g) - \bar Q(\beta,\check\alpha,\hat g) + \bar Q(\beta,\check\alpha,\hat g) - \bar Q(\beta,\alpha,g) \leq o_P(1)
\end{align*}
by Steps 3 and 4. The asserted claim of the theorem follows.
\end{proof}

\newpage

\appendix

\begin{center}
\large{\textbf{Supplementary Materials for ``Spectral and Post-Spectral Estimators for Grouped Panel Data Models'' by D. Chetverikov and E. Manresa}}
\end{center}


\section{Technical Lemmas}\label{sec: technical lemmas}
\begin{lemma}\label{lem: eigenvalue comparison}
Let $A$ and $B$ be two symmetric $N\times N$ matrices, and suppose that the matrix $B$ has only $K\leq N$ non-zero eigenvalues $\lambda_1^B,\dots,\lambda_K^B$. Further, let $\lambda_1^A,\dots,\lambda_K^A$ be the $K$ largest in absolute value eigenvalues of the matrix $A$. Then
$$
\left| \sum_{k=1}^K \lambda_k^A - \sum_{k=1}^K\lambda_k^B \right| \leq 3K\|A-B\|.
$$
\end{lemma}
\begin{remark}
This lemma does not follow from Weyl's inequality immediately because $\lambda_1^A,\dots,\lambda_K^A$ are the largest {\em in absolute value} eigenvalues of $A$.\qed
\end{remark}
\begin{proof}
Let $\lambda_{K+1}^B,\dots,\lambda_N^B$ be the remaining eigenvalues of $B$, so that 
\begin{equation}\label{eq: remaining eigenvalues B}
\lambda_{K+1}^B = \dots \lambda_N^B =0,
\end{equation}
and let $\lambda_{K+1}^A,\dots,\lambda_N^A$ be the remaining eigenvalues of $A$, so that 
\begin{equation}\label{eq: remaining eigenvalues A}
\min(|\lambda_1^A|,\dots,|\lambda_K^A|)\geq \max(|\lambda_{K+1}^A|,\dots,|\lambda_N^A|).
\end{equation}
By Weyl's inequality (\cite{V18}, Theorem 4.5.3), one can construct a one-to-one function $f\colon\{1,\dots,N\}\to\{1,\dots,N\}$ such that
\begin{equation}\label{eq: weyl}
|\lambda_k^A - \lambda_{f(k)}^B|\leq \|A-B\|,\quad\text{for all }k=1,\dots,N.
\end{equation}
Using this function, define $\mathcal K = \{j = 1\dots,K\colon f^{-1}(j)>K\}$ and note that if this set is non-empty, then there exists $k=1,\dots,K$ such that $f(k)>K$, and so for any $j \in\mathcal K$,
$$
|\lambda_j^B|\leq |\lambda^A_{f^{-1}(j)}| + \|A-B\| \leq |\lambda^A_k| + \|A-B\| \leq |\lambda^B_{f(k)}| + 2\|A-B\| = 2\|A-B\|,
$$
where the first inequality follows from \eqref{eq: weyl} and the triangle inequality, the second from \eqref{eq: remaining eigenvalues A}, the third from \eqref{eq: weyl} and the triangle inequality, and the fourth from \eqref{eq: remaining eigenvalues B}. Hence,
$$
\left|\sum_{k=1}^K \lambda_k^B - \sum_{k=1}^K \lambda_{f(k)}^B\right| = \left| \sum_{j\in\mathcal K} \lambda_j^B \right| \leq 2K\|A-B\|.
$$
In addition,
$$
\left| \sum_{k=1}^K \lambda_k^A - \sum_{k=1}^K\lambda_{f(k)}^B \right| \leq K\|A-B\|
$$
by \eqref{eq: weyl}. Combining the last two inequalities gives the asserted claim.
\end{proof}

\begin{lemma}\label{lem: cs matrices}
Let $\mu_1,\dots,\mu_N$ and $\nu_1,\dots,\nu_N$ be two sequences of vectors in $\mathbb R^T$. Define
$$
A = \sum_{i=1}^N \mu_i\mu_i',\quad B = \sum_{i=1}^N \nu_i\nu_i',\quad C = \sum_{i=1}^N\mu_i\nu_i',\ \ \text{and}\ \ D = \sum_{i=1}^N (\mu_i + \nu_i)(\mu_i + \nu_i)'.
$$
Then
$\|C\| \leq \sqrt{\|A\|}\sqrt{\|B\|}$
and $\|D\| \leq 2(\|A\|+\|B\|)$.
\end{lemma}

\begin{proof}
For any $x\in\mathbb R^T$ such that $\|x\|=1$, 
\begin{align*}
x'Cx & = \sum_{i=1}^N (x'\mu_i)(\nu_i'x) \leq \sqrt{\sum_{i=1}^N (x'\mu_i)^2}\sqrt{\sum_{i=1}^N (\nu_i'x)^2} \\
& =  \sqrt{x'Ax}\sqrt{x'Bx} \leq \sqrt{\|A\|}\sqrt{\|B\|},
\end{align*}
by the Cauchy-Schwarz inequality. Taking the supremum over all $x\in\mathbb R^T$ with $\|x\|=1$ of the left- and right-hand sides of this chain of inequalities gives the first asserted claim.

To prove the second asserted claim, note that
\begin{align*}
\|D\| &\leq \left\|\sum_{i=1}^N \mu_i \mu_i'\right\| + \left\|\sum_{i=1}^N \mu_i \nu_i'\right\| + \left\|\sum_{i=1}^N \nu_i \mu_i'\right\|+ \left\|\sum_{i=1}^N \nu_i \nu_i'\right\|\\
& \leq \|A\| + \sqrt{\|A\|}\sqrt{\|B\|} + \sqrt{\|B\|}\sqrt{\|A\|} + \|B\|\\
& \leq \|A\| + (\|A\|+\|B\|)/2 + (\|B\|+\|A\|)/2 + \|B\| = 2(\|A\|+\|B\|).
\end{align*}
This completes the proof of the lemma.
\end{proof}

\begin{lemma}\label{lem: orthogonality of two bases}
Let $F_1$, $F_2$, and $F_3$ be vectors in $\mathbb R^T$ and suppose that (i) $\|F_2 - F_1\|\leq \sqrt 2$, (ii) $F_2'F_3 = 0$, and (iii) $\|F_1\|=\|F_2\|=\|F_3\|=1$. Then $|F_1'F_3|\leq 3\|F_2 - F_1\|$.
\end{lemma}
\begin{proof}
By (iii),
\begin{equation}\label{eq: usual quadratic expansion lemma}
F_1'F_3 = 1 - \frac{\|F_3 - F_1\|^2}{2}.
\end{equation}
By the triangle inequality,
$$
\|F_3 - F_1\| \geq \|F_3 - F_2\| - \|F_2 - F_1\|, 
$$
and so, given that $\|F_3-F_2\|=\sqrt 2$ by (ii), it follows from (i) that
\begin{equation}\label{eq: new lemma eq 1}
\|F_3-F_1\|^2 \geq \Big(\|F_3 - F_2\| - \|F_2 - F_1\|\Big)^2 \geq 2 - 2\sqrt 2\|F_2-F_1\|.
\end{equation}
Combining \eqref{eq: usual quadratic expansion lemma} and \eqref{eq: new lemma eq 1},
\begin{equation}\label{eq: new lemma eq 2}
F_1'F_3 \leq 1 - 1 + \sqrt 2\|F_2-F_1\| =\sqrt 2\|F_2-F_1\|.
\end{equation}
Also, again by the triangle inequality,
$$
\|F_3-F_1\| \leq \|F_3-F_2\| + \|F_2-F_1\|,
$$
and so, by (i),
\begin{align}
\|F_3-F_1\|^2 & \leq \Big(\|F_3-F_2\| + \|F_2-F_1\|\Big)^2 \nonumber\\
& \leq 2 + 2\sqrt 2\|F_2-F_1\| + \sqrt 2\|F_2-F_1\| = 2+3\sqrt 2\|F_2-F_1\|.\label{eq: new lemma eq 3}
\end{align}
Combining \eqref{eq: usual quadratic expansion lemma} and \eqref{eq: new lemma eq 3},
$
F_1'F_3 \geq -3\|F_2-F_1\|.
$
Combining this inequality with \eqref{eq: new lemma eq 2} gives the asserted claim.
\end{proof}

\begin{lemma}\label{lem: weird progression}
Let $\theta\in(-1,1)$ be a real number and $T\geq 2$ be an integer. Then for all $u = (u_1,\dots,u_T)'\in\mathbb R^T$ such that $\sum_{t=1}^T u_t^2\leq 1$, we have
$$
\sum_{t=1}^{T-1}\left( \sum_{r=t}^{T-1} u_{r+1}\theta^{r-t} \right)^2\leq \frac{1}{(1-\theta)^2}.
$$ 
\end{lemma}
\begin{proof}
We proceed by induction. When $T=2$, we have
$$
\sum_{t=1}^{T-1}\left( \sum_{r=t}^{T-1} u_{r+1}\theta^{r-t} \right)^2 = u_2^2 \leq 1\leq \frac{1}{(1-\theta)^2},
$$
so that the claim holds. Now suppose that the claim holds for all $T=2,\dots,k$. We will prove that the claim holds for $T=k+1$. To do so, fix any $u = (u_1,\dots,u_T)'\in\mathbb R^T$ such that $\sum_{t=1}^T u_t^2\leq 1$ and observe that
\begin{align*}
\sum_{t=1}^{T-1}\left( \sum_{r=t}^{T-1} u_{r+1}\theta^{r-t} \right)^2
& = \sum_{t=1}^{T-2}\left( u_{t+1} + \sum_{r=t+1}^{T-1} u_{r+1}\theta^{r-t} \right)^2 + u_T^2\\
& = \sum_{t=1}^{T-2} \left(u_{t+1}^2 + 2u_{t+1}\sum_{r=t+1}^{T-1} u_{r+1}\theta^{r-t} + \left(\sum_{r=t+1}^{T-1} u_{r+1}\theta^{r-t}\right)^2\right) + u_T^2\\
&\leq 1 + \sum_{t=1}^{T-2}\sum_{r=t+1}^{T-1}(u_{t+1}^2 + u_{r+1}^2)\theta^{r-t} + \theta^2 \sum_{t=1}^{T-2}\left(\sum_{r=t}^{T-2}u_{r+2}\theta^{r-t}\right)^2,
\end{align*}
where the third line follows from $\sum_{t=1}^T u_t^2\leq 1$ and an elementary inequality $2ab\leq a^2 + b^2$. Also, by the induction hypothesis,
$$
\sum_{t=1}^{T-2}\left(\sum_{r=t}^{T-2}u_{r+2}\theta^{r-t}\right)^2 \leq \frac{1}{(1-\theta)^2}.
$$
In addition,
$$
\sum_{t=1}^{T-2}\sum_{r=t+1}^{T-1}(u_{t+1}^2 + u_{r+1}^2)\theta^{r-t} \leq 2\sum_{l=1}^{T-2}\sum_{t=1}^T\theta^l u_t^2 \leq \frac{2\theta}{1-\theta}.
$$
Hence,
$$
\sum_{t=1}^{T-1}\left( \sum_{r=t}^{T-1} u_{r+1}\theta^{r-t} \right)^2 \leq 1 +  \frac{2\theta}{1-\theta} +  \frac{\theta^2}{(1-\theta)^2} = \frac{1}{(1-\theta)^2},
$$
which completes the induction argument and thus gives the asserted claim for all $T\geq 2$.
\end{proof}

\section{Proofs for Remaining Results from Main Text}\label{sec: proofs for extensions}
\begin{proof}[Proof of Theorem \ref{thm: asymptotic distribution}]
By Theorem \ref{thm: classifier consistency new}, $\hat\beta = \hat\beta^0$ with probability $1 - o(1)$ for $\hat\beta^0$ appearing in \eqref{eq: correct specification}. In turn, by the Frisch-Waugh-Lovell theorem,
$$
\hat\beta^0 = \left(\sum_{i=1}^N\sum_{t=1}^T \check x_{it}\check x_{it}'\right)^{-1}\left(\sum_{i=1}^N\sum_{t=1}^T\check x_{it} y_{it}\right),
$$
and so
\begin{align*}
\sqrt{NT}(\hat\beta^0 - \beta) 
& = \left(\frac{1}{NT}\sum_{i=1}^N\sum_{t=1}^T \check x_{it}\check x_{it}'\right)^{-1}\left(\frac{1}{\sqrt{NT}}\sum_{i=1}^N\sum_{t=1}^T\check x_{it}v_{it}\right)\to_D N(\check\Sigma^{-1}\Omega\check\Sigma^{-1})
\end{align*}
by Slutsky's lemma and Assumption \ref{as: distribution}. The asserted claim follows.
\end{proof}

\begin{proof}[Proof of Theorem \ref{thm: dynamic model}]
The asserted claim follows from the same arguments as those in the proofs of Theorems \ref{thm: consistency spectral estimator}--\ref{thm: asymptotic distribution} as long as we can show that there exists a constant $C>0$ such that $\|\mathring \rho_{im}\|\vee |\mathring \alpha_{\gamma t}^m|\leq C$ for all $i=1,\dots,N$, $\gamma=1,\dots,G$, $t=1,\dots,T$, and $m=1,\dots,2M$, which would correspond to Assumption \ref{as: bounded rho} in the context of dynamic model. To do so, observe that
\begin{align*}
\max_{1\leq m\leq 2M}\|\mathring\rho_{im}\| & \leq \max_{1\leq m\leq M}(\|\rho_{im}\| + |\rho_{im}^y|) \\
&\leq  \max_{1\leq m\leq M} (\|\rho_{im}\| + 1 + |\rho_{im}'\beta|) \leq 1 + C_3(1+\|\beta\|)
\end{align*}
for all $i=1,\dots,N$ by Assumption 4.5. Also, 
$$
\max_{1\leq m\leq 2M}|\mathring \alpha_{\gamma t}^m| \leq \max_{1\leq m\leq M} \left(|\alpha_{\gamma t}^m| + \left| \sum_{r=0}^{t-2}\theta^r\alpha_{\gamma t-r-1}^m \right|\right)\leq C_4\left(1 + \frac{1}{1-\theta}\right)
$$
by Assumption 4.6. Therefore, the asserted claim follows if we set
$$
C = 1+C_3(1+\|\beta\|) + C_4\left(1 + \frac{1}{1-\theta}\right).
$$
This completes the proof of the theorem.
\end{proof}

\begin{proof}[Proof of Theorem \ref{thm: hd model}]
We first prove \eqref{eq: main convergence in hd model}. By construction of $\hat\Sigma$ and $\hat S$, it suffices to show that 
$$
\lambda^b_1 + \dots + \lambda^b_{2GM+2} = b'\Sigma b + S'b + L + O_P\left(\frac{1}{T\wedge N}+\sqrt{\frac{\log d}{NT}}\right)
$$
uniformly over $b\in\overline{\mathcal B} = \{0_d\}\cup \{e_k\colon k=1,\dots,d\}\cup \{e_k + e_l\colon k,l=1,\dots,d\}$. To do so, we proceed by appropriately modifying the proof of Theorem \ref{thm: consistency spectral estimator}. Throughout, we use the same notations as in the proof of Theorem \ref{thm: consistency spectral estimator}. We have
\begin{align*}
\textrm{tr}(R) & =-\frac{2}{NT}\sum_{i=1}^{N}\sum_{t=1}^{T}\Big\{(\beta-b)'z_{it}z_{it}(\beta-b)+v_{it}^{2}+2v_{it}z_{it}'(\beta-b)\Big\}\\
 & =-\frac{2}{NT}\sum_{i=1}^{N}\sum_{t=1}^{T}\Big\{(\beta-b)'z_{it}z_{it}(\beta-b)+v_{it}^{2}\Big\}+O_{P}\left(\sqrt{\frac{\log d}{NT}}\right)
\end{align*}
uniformly over $b\in\overline{\mathcal B}$ since
$$
\left|\frac{1}{NT}\sum_{i=1}^N\sum_{t=1}^T v_{it}z_{it}'(\beta - b)\right| \leq \left\| \frac{1}{NT}\sum_{i=1}^N\sum_{t=1}^T v_{it}z_{it} \right\|_{\infty} \|\beta - b\|_1 \leq (C_{\beta} + 2) O_{P}\left(\sqrt{\frac{\log d}{NT}}\right)
$$
uniformly over $b\in\overline{\mathcal B}$. Thus,
\begin{align*}
\textrm{tr}(A_0)&=\frac{2}{NT}\sum_{i=1}^{N}\sum_{t=1}^{T}\Big\{(\beta-b)'z_{it}z_{it}(\beta-b)+v_{it}^{2}\Big\}+O_{P}\left(\sqrt{\frac{\log d}{NT}}\right)\\
& = b'\Sigma b + S'b + L + O_{P}\left(\sqrt{\frac{\log d}{NT}}\right)
\end{align*}
uniformly over $b\in\overline{\mathcal B}$ since
\begin{align*}
&\left| \frac{2}{NT}\sum_{i=1}^N\sum_{t=1}^T(\beta - b)'z_{it}z_{it}'(\beta - b) - (\beta - b)'\Sigma(\beta - b) \right|\\
&\qquad \leq \left\| \frac{1}{NT}\sum_{i=1}^N\sum_{t=1}^T z_{it}z_{it}' - \Sigma \right\|_{\infty}\|\beta - b\|_1^2 \leq (C_{\beta} + 2)^2 O_{P}\left(\sqrt{\frac{\log d}{NT}}\right)
\end{align*}
uniformly over $b\in\overline{\mathcal B}$. Also,
$$
\left\|\frac{1}{NT}\sum_{t=1}^T V_tV_t'\right\| = O_P\left(\frac{1}{T\wedge N}\right)
$$
as in the proof of Theorem \ref{thm: consistency spectral estimator}. Further, to prove that
\begin{equation}\label{eq: z matrix convergence}
\left\|\frac{1}{NT}\sum_{t=1}^T Z_{bt}Z_{bt}'\right\| = O_P\left(\frac{1}{T\wedge N}\right)
\end{equation}
uniformly over $b\in\overline{\mathcal B}$, where we denoted $Z_{bt} = (z_{1t}'(\beta - b),\dots,z_{Nt}'(\beta - b))'$, we have $\|u'z_i(\beta - b)\|_{\psi_2} \leq \|u'z_i\beta\|_{\psi_2} + \|u'z_ib\|_{\psi_2} \leq C_z + 2C_2$ uniformly over $u=(u_1,\dots,u_T)'\in\mathcal S^T$ and $b\in\overline{\mathcal B}$ by our assumptions, where we denoted $z_i = (z_{i1},\dots,z_{iT})'$ for all $i=1,\dots,N$. Thus, denoting $\mathcal Z_b = (z_1(\beta - b),\dots,z_N(\beta - b))'$, we have
$$
\mathrm P\left(\max_{u\in\mathcal N}\left|u'\left( \frac{\mathcal Z_b'\mathcal Z_b}{N} - \mathbb E\left[ \frac{\mathcal Z_b'\mathcal Z_b}{N} \right] \right) u\right| \geq \epsilon \right)\leq 2\times 9^T\times\exp(-c(\epsilon\wedge\epsilon^2)N)
$$
for all $b\in\overline{\mathcal B}$, where $\mathcal N$ is the same as in the proof of Theorem \ref{thm: consistency spectral estimator} and $c$ is a constant depending only on $C_z$ and $C_2$. Hence, by the union bound,
$$
\mathrm P\left( \max_{b\in\overline{\mathcal B}}\max_{u\in\mathcal N}\left|u'\left( \frac{\mathcal Z_b'\mathcal Z_b}{N} - \mathbb E\left[ \frac{\mathcal Z_b'\mathcal Z_b}{N} \right] \right) u\right| \geq \epsilon \right)\leq (1+d+d^2)\times 2\times 9^T\times\exp(-c(\epsilon\wedge\epsilon^2)N).
$$
Setting here $\epsilon = c^{-1}(\log 9)(T/N)+1$, we obtain
\begin{multline*}
\mathrm P\left( \max_{b\in\overline{\mathcal B}}\max_{u\in\mathcal N}\left|u'\left( \frac{\mathcal Z_b'\mathcal Z_b}{N} - \mathbb E\left[ \frac{\mathcal Z_b'\mathcal Z_b}{N} \right] \right) u\right| \geq  c^{-1}(\log 9)(T/N)+1 \right)\\
 \leq 2(1+d+d^2)\exp(-cN)\to 0
\end{multline*}
because $\log d = o(N)$. We thus obtain \eqref{eq: z matrix convergence} by the same arguments as those in the proof of Theorem \ref{thm: consistency spectral estimator}. Repeating the remaining arguments of the proof of Theorem \ref{thm: consistency spectral estimator}, we obtain \eqref{eq: main convergence in hd model}.

Next, we prove that \eqref{eq: lambda dominance} and \eqref{eq: second event dominance} imply \eqref{eq: implication l1 norm hd case}. To do so, we assume for the rest of the proof that both \eqref{eq: lambda dominance} and \eqref{eq: second event dominance} are satisfied. Then, by the definition of $\hat\beta_{\lambda}$,
$$
\hat\beta_{\lambda}'\hat\Sigma\hat\beta_{\lambda} +\hat S' \hat\beta_{\lambda} + \lambda\|\hat\beta_{\lambda}\|_1\leq \beta'\hat\Sigma \beta + \hat S'\beta + \lambda\|\beta\|_1.
$$
Also,
$$
(\hat\beta_{\lambda} - \beta)'\hat \Sigma (\hat\beta_{\lambda} - \beta) = \hat\beta_{\lambda}'\hat\Sigma\hat\beta_{\lambda} - \beta'\hat\Sigma\beta + 2(\beta - \hat\beta_{\lambda})'\hat\Sigma\beta.
$$
Taking the sum of these two displays, we obtain
\begin{align*}
(\hat\beta_{\lambda} - \beta)'\hat \Sigma (\hat\beta_{\lambda} - \beta) 
& = (\hat S + 2\hat\Sigma\beta)'(\beta - \hat\beta_{\lambda}) + \lambda\|\beta\|_1 - \lambda\|\hat\beta_{\lambda}\|_1\\
& \leq \|\hat S + 2\hat\Sigma\beta\|_{\infty}\|\hat\beta_{\lambda} - \beta\|_1 + \lambda\|\beta\|_1 - \lambda\|\hat\beta_{\lambda}\|_1\\
& \leq \Big(\|\hat S - S\|_{\infty} + 2 C_{\beta}\|\hat\Sigma - \Sigma\|_{\infty}\Big)\|\hat\beta_{\lambda} - \beta\|_1 + \lambda\|\beta\|_1 - \lambda\|\hat\beta_{\lambda}\|_1,
\end{align*}
where the third line follows by recalling that $S + 2\Sigma\beta=0$ and $\|\beta\|_1\leq C_{\beta}$. Therefore, by \eqref{eq: lambda dominance}, we have
\begin{equation}\label{eq: key rearrangement}
(\hat\beta_{\lambda} - \beta)'\hat \Sigma (\hat\beta_{\lambda} - \beta) \leq (\lambda / c_{\lambda})\|\hat\beta_{\lambda} - \beta\|_1 + \lambda\|\beta\|_1 - \lambda\|\hat\beta_{\lambda}\|_1.
\end{equation}
Further, denote $\delta = \hat\beta_{\lambda} - \beta$, $\mathcal T = \{k=1,\dots,d\colon \beta_k\neq 0\}$, and $\mathcal T^c = \{1,\dots,d\}\setminus \mathcal T$. Also, let $\delta_{\mathcal T} = (\delta_{\mathcal T 1},\dots,\delta_{\mathcal T d})'$ be a $d\times 1$ vector such that $\delta_{\mathcal T k} = \delta_k 1\{k\in\mathcal T\}$ and, similarly, let $\delta_{\mathcal T^c} = (\delta_{\mathcal T^c 1},\dots,\delta_{\mathcal T^c d})'$ be a $d\times 1$ vector such that $\delta_{\mathcal T^c k} = \delta_k 1\{k\in\mathcal T^c\}$. Then, given that $(\hat\beta_{\lambda} - \beta)'\hat \Sigma (\hat\beta_{\lambda} - \beta)\geq 0$, it follows from \eqref{eq: key rearrangement} that
\begin{align*}
0&\leq (1 / c_{\lambda})\|\hat\beta_{\lambda} - \beta\|_1 + \|\beta\|_1 - \|\hat\beta_{\lambda}\|_1\\
&\leq  (1/c_{\lambda})(\|\delta_{\mathcal T} \|_1 + \|\delta_{\mathcal T^c}\|_1) + \| \delta_{\mathcal T}\|_1 - \|\delta_{\mathcal T^c}\|_1,
\end{align*}
where the second inequality follows from observing that $\beta_k = \beta_k1\{k\in\mathcal T\}$ for all $k=1,\dots,d$. Therefore,
$
(1-1/c_{\lambda})\|\delta_{\mathcal T^c}\|_1 \leq (1+1/c_{\lambda})\|\delta_{\mathcal T}\|_1,
$
and so $\|\delta_{\mathcal T^c}\|_1 \leq \bar c_{\lambda}\|\delta_{\mathcal T}\|_1$. Thus,
\begin{align*}
\delta'\hat\Sigma\delta 
& = \delta'\Sigma\delta - \delta'(\Sigma - \hat\Sigma)\delta \geq c_{\Sigma}\|\delta\|^2 - \|\delta\|_1^2\|\hat\Sigma - \Sigma\|_{\infty} \\
&  \geq c_{\Sigma}\|\delta\|^2 - (1+\bar c_{\lambda})^2\|\delta_{\mathcal T}\|_1^2\|\hat\Sigma - \Sigma\|_{\infty}\\
&\geq c_{\Sigma}\|\delta\|^2 - s(1+\bar c_{\lambda})^2\|\delta\|^2\|\hat\Sigma - \Sigma\|_{\infty}\geq c_{\Sigma}\|\delta\|^2/2
\end{align*}
by \eqref{eq: second event dominance}. Substituting this bound into \eqref{eq: key rearrangement}, we obtain
$$
c_{\Sigma}\|\delta\|^2/2 \leq \lambda(1 + 1/c_{\lambda})\|\delta_{\mathcal T}\|_1\leq \sqrt s\lambda (1 + 1/c_{\lambda})\|\delta\|.
$$
Rearranging this bound gives gives the second inequality in \eqref{eq: implication l1 norm hd case}. To obtain the first inequality in \eqref{eq: implication l1 norm hd case}, observe that
$$
\|\delta\|_1 =\|\delta_{\mathcal T}\|_1 + \|\delta_{\mathcal T^c}\|_1 \leq (1+\bar c_{\lambda})\|\delta_{\mathcal T}\|_1 \leq \sqrt s(1+\bar c_{\lambda})\|\delta\|.
$$
This completes the proof of the theorem.
\end{proof}

\begin{proof}[Proof of Theorem \ref{thm: interactive model}]
The proof is closely related to that of Theorem \ref{thm: consistency spectral estimator}, with the main difference is that we now have
$$
f_{it} = (\kappa_i + \omega_i(\beta - b))'\phi_t,\quad\text{for all }i=1,\dots,N, \ t=1,\dots,T.
$$
This difference in turn requires us to change the calculation of the bound on the number of non-zero eigenvalues of the matrix $A_0$. Recalling that $F_t = (f_{1t},\dots,f_{Nt})'$ for all $t=1,\dots,T$, we now have
$$
\sum_{t=1}^T F_t(F_t + Z_t + V_t)' = \sum_{t=1}^T \mathcal K \phi_t (F_t + Z_t + V_t)' = \mathcal K\sum_{t=1}^T\phi_t(F_t + Z_t + V_t)',
$$
where $\mathcal K = (\kappa_1 + \omega_1(\beta - b),\dots,\kappa_N + \omega_N(\beta - b))'$ and $Z_t$ and $V_t$ are the same as in the proof of Theorem \ref{thm: consistency spectral estimator}. Thus, given that $\mathcal K$ is an $N\times J$ matrix, it follows that the rank of the matrix $\sum_{t=1}^T F_t(F_t + Z_t + V_t)'$ is at most $J$. Similarly, the rank of the matrix $\sum_{t=1}^T (Z_t + V_t)F_t'$ is also at most $J$, as the column rank coincides with the row rank. We have thus replaced $GM$ in the proof of Theorem \ref{thm: consistency spectral estimator} by $J$. The rest of the proof is the same as that of Theorem \ref{thm: consistency spectral estimator}.
\end{proof}

\section{Randomized Algorithm for Calculating Eigenvalues of Large Matrices}\label{sec: randomized algorithm}
To calculate the spectral estimator $\tilde\beta$ in Section \ref{sec: estimator}, we had to calculate $2G M + 2$ largest in absolute value eigenvalues of the $N\times N$ matrix $A^b$. When $N$ is large, calculating these eigenvalues exactly may be difficult. Fortunately, there exists a class of randomized algorithms that allow to calculate these eigenvalues approximately with minimal efforts. In this section, we describe one such algorithm. Our discussion here mostly follows \cite{HMT11}, where an interested reader can find several other related algorithms. 

For brevity of notation, suppose that we have an $N\times N$ symmetric matrix $A$ and we would like to calculate its $k$ largest in absolute value eigenvalues, $\lambda_{(1)},\dots,\lambda_{(k)}$, ordered so that $|\lambda_{(1)}|\geq \dots\geq |\lambda_{(k)}|$. Consider the following algorithm:

\medskip
\noindent
{\bf Randomized Algorithm for Calculating Eigenvalues.}

{\em Step 1:} choose an oversampling parameter $p > 0$, e.g. $p=5$ or $10$;

{\em Step 2:} set a multiplication parameter $q=[\log N]$;

{\em Step 3:} draw an $N\times(k+p)$ random matrix $\Omega = \{\Omega_{ij}\}_{i,j=1}^{N,k+p}\stackrel{iid}{\sim} N(0,1)$;

{\em Step 4:} compute the $N\times(k+p)$ matrix $Y = A^{q+1}\Omega$;

{\em Step 5:} compute QR decomposition $Y = QR$ with $Q$ having orthonormal columns;

{\em Step 6:} compute the $(k+p)\times N$ matrix $B=Q'A$;

{\em Step 7:} compute eigenvectors $\tilde s_1,\dots,\tilde s_{k+p}$ of the $(k+p)\times(k+p)$ matrix $BB'$;

{\em Step 8:} compute $N\times 1$ vectors $s_j = B'\tilde s_j$ for $j=1,\dots,k+p$;

{\em Step 9:} compute $\hat\lambda_j = \textrm{sign}(s_j'As_j)(\|B'B s_j\|/\|s_j\|)^{1/2}$ for $j=1,\dots,k+p$;

{\em Step 10:} order values $\hat\lambda_1,\dots,\hat\lambda_{k+p}$ into $\hat\lambda_{(1)},\dots,\hat\lambda_{(k+p)}$ so that $|\hat\lambda_{(1)}|\geq\dots\geq|\hat\lambda_{(k+p)}|$;

{\em Step 11:} return $\hat\lambda_{(1)},\dots,\hat\lambda_{(k)}$.
\medskip

The result of this algorithm is $k$ values $\hat\lambda_{(1)},\dots,\hat\lambda_{(k)}$. These are estimators of $k$ largest in absolute value eigenvalues $\lambda_{(1)},\dots,\lambda_{(k)}$ of the matrix $A$. As follows from results in \cite{HMT11}, these estimators are consistent as $N\to\infty$ under conditions to be discussed below even though they are based on a realization of the random matrix $\Omega$. Specifically, Corollary 10.10 in \cite{HMT11} shows that
$$
\mathbb E[\|A - QQ'A\|]\leq\left(1 + \sqrt{\frac{k}{p-1}} +\frac{e\sqrt{N(k+p)}}{p}\right)^{\frac{1}{q+p+1}}|\lambda_{(k+1)}|,
$$
where $\lambda_{(k+1)}$ is the $(k+1)$th largest in absolute value eigenvalue of $A$. Therefore, by Markov's inequality,
\begin{equation}\label{eq: hmt}
\|A-QQ'A\| = O_P(|\lambda_{(k+1)}|).
\end{equation}
In turn, by the triangle inequality, the fact that $Q'Q$ is the identity matrix, and the definition $B=Q'A$,
\begin{align*}
\|A'A - B'B\| 
& = \|A'A - A'QQ'QQ'A\|\\
&\leq \|A\|\|A - QQ'A\| + \|QQ'A\|\|A-QQ'A\| \leq 2\|A\|\|A-QQ'A\|
\end{align*}
Thus, given that $\hat\lambda_{(1)}^2,\dots,\hat\lambda_{(k)}^2$ are $k$ largest eigenvalues of the matrix $B'B$ by construction (see Steps 8,9, and 10 in the algorithm above), it follows from Weyl's inequality that $\hat\lambda_{(1)}^2,\dots,\hat\lambda_{(k)}^2$ are consistent estimators of $\lambda_{(1)}^2,\dots,\lambda_{(k)}^2$ as long as $\|A\|=O_P(1)$ and $|\lambda_{(k+1)}|=o_P(1)$ as $N\to\infty$. Hence, $\hat\lambda_{(j)}\to\lambda_{(j)}$ for all $j=1,\dots,k$ under the same conditions by the Davis-Kahane theorem. 

The described algorithm can be applied in Section \ref{sec: estimator} to calculate the spectral estimator $\tilde\beta$ with $A = A^b$ and $k=2G M + 2$. In this case, the aforementioned conditions $\|A\|=O_P(1)$ and $|\lambda_{(k+1)}|=o_P(1)$ are satisfied under Assumptions \ref{as: spectral v} and \ref{as: spectral vz} by the proof of Theorem \ref{thm: consistency spectral estimator}.

\section{Relation between Assumptions 3.1--3.12 and Assumptions 4.1--4.12}\label{sec: additional details on dynamic model}
In this section, we explain what conditions one has to impose on top of Assumptions 3.1--3.12 to obtain Assumptions 4.1--4.12. To start with, note that Assumptions 4.1 and 4.5--4.9 coincide with Assumptions 3.1 and 3.5--3.9. Also, Assumption 4.11 follows immediately from Assumption 3.11 if we define $\mathring{\mathcal B} = [-1,1]\times\mathcal B$, as proposed in the main text. In addition, Assumptions 3.10 and 3.12 depend on $x_{it}$ but their versions corresponding to the dynamic model, i.e. Assumptions 4.10 and 4.12, are discussed in \cite{BM15}, where some interpretations as well as low-level conditions are provided. We thus only need to discuss Assumptions 4.2, 4.3, and 4.4.

To this end, suppose that Assumptions 3.1, 3.2, 3.3, and 3.4 are satisfied and, in addition, suppose that the noise variables $v_{it}$ satisfy \eqref{eq: model with predetermined v's}. Moreover, suppose that there exists a constant $C_y>0$ such that $\|y_{i0}\|_{\psi_2}\leq C_y$ for all $i=1,\dots,N$. Then for all $u=(u_1,\dots,u_T)'\in\mathcal S^T$, we have
$$
\left\| \sum_{t=1}^T u_t\sum_{r=0}^{t-2}\theta^rv_{it-r-1} \right\|_{\psi_2} = \left\| \sum_{t=1}^{T-1}v_{it}\sum_{r=t}^{T-1} u_{r+1}\theta^{r-t} \right\|_{\psi_2} \leq \frac{C_1}{1-\theta}
$$
by Assumption \ref{as: spectral v}(i) and Lemma \ref{lem: weird progression} and, similarly,
$$
\left\| \sum_{t=1}^T u_t\sum_{r=0}^{t-2}\theta^rz_{it-r-1}'\beta \right\|_{\psi_2} = \left\| \sum_{t=1}^{T-1}z_{it}'\beta\sum_{r=t}^{T-1} u_{r+1}\theta^{r-t} \right\|_{\psi_2} \leq \frac{d_zC_2\|\beta\|_{\infty}}{1-\theta}
$$
by Assumption \ref{as: spectral v}(ii) and Lemma \ref{lem: weird progression}, where $d_z$ is the dimension of $z_{it}$. Also, 
$$
\left\|\sum_{t=1}^T u_t \theta^{t-1}y_{i0}\right\|_{\psi_2} \leq \frac{\|y_{i0}\|_{\psi_2}}{1-\theta}\leq \frac{C_y}{1-\theta}.
$$
Hence, by the triangle inequality,
$$
\left\|\sum_{t=1}^T u_t z_{it}^y\right\|_{\psi_2} \leq \frac{C_y+C_1+d_zC_2\|\beta\|_{\infty}}{1-\theta},
$$
and so Assumption 4.2 is satisfied.

Next,
\begin{align*}
\mathbb E\left[ \left(\sum_{i=1}^N\sum_{t=1}^T v_{it}z_{it}^y\right)^2 \right]
& = \sum_{i=1}^N\sum_{t=1}^T \mathbb E\left[(v_{it}z_{it}^y)^2\right]
\leq \sum_{i=1}^N\sum_{t=1}^T\sqrt{\mathbb E[v_{it}^4]\mathbb E[(z_{it}^y)^4]}\leq CNT
\end{align*}
for some constant $C>0$ by \eqref{eq: model with predetermined v's}, the triangle inequality, Assumption \ref{as: spectral v}, the fact that $\|y_{i0}\|_{\psi_2}\leq C_y$ for all $i=1,\dots,N$, and (2.15) in \cite{V18}. Thus, Assumption 4.3 is satisfied as well. Finally, regarding Assumption 4.4, observe that the convergence $(NT)^{-1}\sum_{i=1}^N\sum_{t=1}^T \mathring z_{it}\mathring z_{it}' = \mathring{\Sigma} + O_P(1/\sqrt{NT})$ holds by a law of large numbers as long as the dependence of $(z_{it},v_{it})$'s across $t$ is not too strong. Also, by Assumption 3.4, the matrix $\mathring\Sigma$ has at most one zero eigenvalue, whereas Assumption 4.4 requires that $\mathring\Sigma$ has no zero eigenvalues, so that $\mathring\Sigma$ is invertible. Although it seems difficult to provide low-level conditions implying invertibility of $\mathring\Sigma$, as it depends on the auto-covariance structure of the random processes $(z_{it}',v_{it})'$, $t=1,\dots,T$, we note that $\mathring\Sigma$ can be consistently estimated (as it is done in the process of constructing the spectral estimator), and so the invertibility condition is testable.

\section{Motivating Example}\label{sec: motivating example}
To motivate equation (3), we provide a specific example in the context of agricultural production and environmental economics. Suppose a production process has an input that is potentially polluting, $x$, such as a pesticide. In addition, suppose we want to assess the incidence of the use of this substance in the environment's pollution, such as water or soil. See, for example, \cite{nicholls1988factors}. The incidence of polluting substances on the environment depends on the intensity of use, but also on the characteristics of the environment, such as soil permeability, rainfall, etc. 

Suppose we have the following model to quantify the effect of the use of input $x$ on some measure of pollution:
\begin{eqnarray*}
    y_{it} = x_{it}\beta + \alpha_{g_it} + \varepsilon_{it}
\end{eqnarray*}
where $y_{it}$  is a measure of pollution, $x_{it}$ is the quantity of pesticide, $\alpha_{g_it}$ are unobserved characteristics of the region, $g_i$, such as soil permeability, rain incidence, and other characteristics that influence the presence of chemicals in the environment, and $\varepsilon$ is an unobserved shock. The farmer chooses to purchase pesticide in the local market as to maximize expected profits:
\begin{equation}
    x_{it} = \arg \max_{x}  \mathbb{E} \left[\pi_i(x,\alpha_{g_it},z_{it},\nu_{it}) |  \alpha_{g_i1}, \ldots \alpha_{g_iT}, z_{i1}, \ldots, z_{iT}\right] - C(x, p^{x}_{g_it})/p_{g_it}
    \label{optim}
\end{equation}
where $\pi_i(x,\alpha_{gt},z_{it},\varepsilon_{it})$ is the production function of farmer $i$, and $C(x, p^{x}_{g_it})$ is the cost associated to purchase $x$ when the price of the input is $p^{x}_{g_it}$, and $p_{g_it}$ is the price of the produced good in the local market of farmer $i$. The production process depends on the use of observable inputs, such as $x$, but also soil and rain characteristics specific to the region, that affect both the prevalence of pesticides in the environment but also the agricultural process, and $z_{it}$, a variable that affects production, such as quality of inputs, observed to the farmer but not to the econometrician, which does not affect the environment. Finally $\nu$ is an idiosyncratic shock in the production process, unforeseeable by both the farmer and the econometrician. 

Assume $\pi_i(x,\alpha_{g_it},z_{it},\nu_{it}) = \tilde{\pi}_i(x,\alpha_{g_it},z_{it}) + \nu_{it}$, and $\mathbb{E} \left[\nu_{it} |  \alpha_{g_i1}, \ldots \alpha_{g_iT}, z_{i1}, \ldots, z_{iT}\right] = 0$. In addition, let $\tilde{\pi}_i'(x,\alpha_{g_it},z_{it}) = \beta_i x + \alpha_{g_it} + z_{it}$, and $C'(x, p^{x}_{g_it}) = p^x_{g_it}$ be the partial derivative with respect to the first argument, respectively for $\tilde{\pi}_i$ and $C$. The F.O.C. in (\ref{optim}) implies that $x_{it}$ is defined as:
\begin{eqnarray*}
    \beta_i x_{it} + \alpha_{gt} + z_{it} - \frac{p^{x}_{gt}}{p_{gt}} = 0,
\end{eqnarray*}
which implies: $x_{it} = \frac{1}{\beta_i} \left(\frac{p^{x}_{gt}}{p_{gt}} - \alpha_{gt}\right) - \frac{1}{\beta_i} z_{it}$,
which takes the form of equation (3) with $M = 2$.

\newpage


\begin{landscape}

     \begin{table}
\caption{ Mean Absolute Error (MAE) and misclassification results: $\sigma^2 = 1$,  $M=1$}
     \hspace*{.2cm}
\begin{tabular}{cccccccccccc}

\hline\hline
 & & &  &  & \multirow{2}{*}{$G=2$} &    & &\tabularnewline
 & & &  &  &  &    & &\tabularnewline
\hline\hline 
\multirow{2}{*}{$ \ \ T\ \ $} & {\multirow{2}{*}{ $\ \ N \ \ $}} & \multicolumn{1}{c|}{} & \multicolumn{7}{c|}{ Mean Absolute Error} & \multicolumn{2}{c}{Misclassification}\tabularnewline
 & & \multicolumn{1}{c|}{} & $ \ $S$ \ $ & \ P-S \ &  LS & Pen NN  & $ \ \ \ $I-GFE$ \ \ \ $ & $ \ \ \ $GFE$ \ \ \ $ &\multicolumn{1}{c|}{$ \ \ $Oracle$ \ \ $} & $ \ $S$ \ $ & $ \ \ \ $GFE$ \ \ $\tabularnewline
\hline 
  \multirow{3}{*}{20} & 100 & & 0.035 &  0.018 &   0.059 &  0.151 &  0.016 &  0.034 &  0.014 &  0.009 &  0.058 \cr
    & 200 & &  0.024 &  0.009 &   0.019 &  0.150 &  0.009 &  0.009 &  0.008 &  0.002 &  0.003 \cr
    & 400 & & 0.016 &  0.007 &   0.010 &  0.151 &  0.007 &  0.008 &  0.006 &  0.003 &  0.006 \cr
   \multirow{3}{*}{50} & 100 & & 0.024 &  0.007 &  0.012 &  0.149 &  0.007 &  0.007 &  0.007 &  0.000 &  0.000 \cr
    & 200 & & 0.015 &  0.006 &  0.008 &  0.150 &  0.006 &  0.006 &  0.006 &  0.000 &  0.000 \cr
    & 400 & & 0.010 &  0.004 &  0.006 &  0.149 &  0.004 &  0.004 &  0.004 &  0.000 &  0.000 \cr
   \multirow{3}{*}{100} & 100 & & 0.014 &  0.006 &  0.010 &  0.153 &  0.006 &  0.006 &  0.006 &  0.000 &  0.000 \cr
    & 200 & & 0.008 &  0.004 &  0.006 &  0.151 &  0.004 &  0.004 &  0.004 &  0.000 &  0.000 \cr
    & 400 & & 0.004 &  0.002 &  0.004 &  0.152 &  0.002 &  0.002 &  0.002 &  0.000 &  0.000 \cr
    
\hline\hline
 & & &  &  & \multirow{2}{*}{$G=7$} &  &   &\tabularnewline
 & & &  &  &  &  &  & &\tabularnewline
\hline\hline 
\multirow{2}{*}{$ \ \ T\ \ $} & {\multirow{2}{*}{ $\ \ N \ \ $}} & \multicolumn{1}{c|}{} & \multicolumn{7}{c|}{ Mean Absolute Error} & \multicolumn{2}{c}{Misclassification}\tabularnewline
 & & \multicolumn{1}{c|}{} & $ \ $S$ \ $ & \ P-S \ &  LS &  Pen NN  & $ \ \ \ $I-GFE$ \ \ \ $ & $ \ \ \ $GFE$ \ \ \ $ &\multicolumn{1}{c|}{$ \ \ $Oracle$ \ \ $} & $ \ $S$ \ $ & $ \ \ \ $GFE$ \ \ $\tabularnewline
\hline 
\multirow{3}{*}{20} & 100 & &  0.116 &  0.109 &   0.144 &  0.151 &  0.023 &  0.144 &  0.011 &  0.346 &  0.705 \cr
    & 200 & &  0.057 &  0.085 &   0.149 &  0.152 &  0.029 &  0.149 &  0.010 &  0.234 &  0.744 \cr
    & 400 & & 0.025 &  0.062 &   0.142 &  0.148 &  0.055 &  0.147 &  0.006 &  0.139 &  0.764 \cr
   \multirow{3}{*}{50} & 100 & & 0.034 &  0.014 &   0.142 &  0.147 &  0.007 &  0.147 &  0.007 &  0.006 &  0.651 \cr
    & 200 & & 0.015 &  0.008 &   0.143 &  0.148 &  0.006 &  0.137 &  0.006 &  0.001 &  0.645 \cr
    & 400 & & 0.009 &  0.004 &   0.062 &  0.150 &  0.004 &  0.137 &  0.004 &  0.000 &  0.622 \cr
   \multirow{3}{*}{100} & 100 & & 0.022 &  0.006 &  0.148 &  0.151 &  0.006 &  0.110 &  0.006 &  0.000 &  0.364 \cr
    & 200 & & 0.009 &  0.003 &   0.083 &  0.152 &  0.003 &  0.112 &  0.003 &  0.000 &  0.289 \cr
    & 400 & & 0.006 &  0.003 &   0.006 &  0.150 &  0.003 &  0.078 &  0.003 &  0.000 &  0.203
   \end{tabular}
\end{table}
\end{landscape}

\newpage


\begin{landscape}

     \begin{table}

\caption{ Mean Absolute Error (MAE)  and misclassification results: $\sigma^2 = 4$,  $M=1$}
     \hspace*{.2cm}
\begin{tabular}{cccccccccccc}

\hline\hline
 & & &  &  & \multirow{2}{*}{$G=2$} &    & &\tabularnewline
 & & &  &  &  &  &   & \tabularnewline
\hline\hline 
\multirow{2}{*}{$ \ \ T\ \ $} & {\multirow{2}{*}{ $\ \ N \ \ $}} & \multicolumn{1}{c|}{} & \multicolumn{7}{c|}{ Mean Absolute Error} & \multicolumn{2}{c}{Misclassification}\tabularnewline
 & & \multicolumn{1}{c|}{} & $ \ $S$ \ $ & \ P-S \ &  LS   & Pen NN  & $ \ \ \ $I-GFE$ \ \ \ $ & $ \ \ \ $GFE$ \ \ \ $ &\multicolumn{1}{c|}{$ \ \ $Oracle$ \ \ $} & $ \ $S$ \ $ & $ \ \ \ $GFE$ \ \ $\tabularnewline
\hline 

   \multirow{3}{*}{20} & 100 & &  0.026 &  0.005 &   0.018 &  0.153 &  0.005 &  0.005 &  0.005 &  0.000 &  0.000 \cr
    & 200 & & 0.015 &  0.003 &   0.014 &  0.159 &  0.003 &  0.003 &  0.003 &  0.000 &  0.000 \cr
    & 400 & & 0.009 &  0.002 &   0.009 &  0.157 &  0.002 &  0.002 &  0.002 &  0.000 &  0.000 \cr
   \multirow{3}{*}{50} & 100 & & 0.015 &  0.003 &   0.014 &  0.156 &  0.003 &  0.003 &  0.003 &  0.000 &  0.000 \cr
    & 200 & & 0.009 &  0.002 &    0.009 &  0.157 &  0.002 &  0.002 &  0.002 &  0.000 &  0.000 \cr
    & 400 & & 0.007 &  0.001 &    0.007 &  0.144 &  0.001 &  0.001 &  0.001 &  0.000 &  0.000 \cr
   \multirow{3}{*}{100} & 100 & & 0.010 &  0.002 &    0.007 &  0.140 &  0.002 &  0.002 &  0.002 &  0.000 &  0.000 \cr
    & 200 & & 0.007 &  0.001 &    0.006 &  0.063 &  0.001 &  0.001 &  0.001 &  0.000 &  0.000 \cr
    & 400 & & 0.004 &  0.001 &    0.005 &  0.011 &  0.001 &  0.001 &  0.001 &  0.000 &  0.000 \cr
    
\hline\hline
 & & &  &  & \multirow{2}{*}{$G=7$} &    &\tabularnewline
 & & &  &  &  &   &\tabularnewline
\hline\hline 
\multirow{2}{*}{$ \ \ T\ \ $} & {\multirow{2}{*}{ $\ \ N \ \ $}} & \multicolumn{1}{c|}{} & \multicolumn{7}{c|}{ Mean Absolute Error} & \multicolumn{2}{c}{Misclassification}\tabularnewline
 & & \multicolumn{1}{c|}{} & $ \ $S$ \ $ & \ P-S \ &  LS &  Pen NN  & $ \ \ \ $I-GFE$ \ \ \ $ & $ \ \ \ $GFE$ \ \ \ $ &\multicolumn{1}{c|}{$ \ \ $Oracle$ \ \ $} & $ \ $S$ \ $ & $ \ \ \ $GFE$ \ \ $\tabularnewline
\hline 
   \multirow{3}{*}{20} & 100 & & 0.055 &  0.005 &    0.099 &  0.155 &  0.004 &  0.055 &  0.004 &  0.000 &  0.154 \cr
    & 200 & & 0.029 &  0.003 &    0.035 &  0.158 &  0.003 &  0.073 &  0.003 &  0.000 &  0.216 \cr
    & 400 & & 0.016 &  0.002 &    0.013 &  0.158 &  0.002 &  0.060 &  0.002 &  0.000 &  0.195 \cr
   \multirow{3}{*}{50} & 100 & & 0.026 &  0.003 &    0.014 &  0.157 &  0.003 &  0.049 &  0.003 &  0.000 &  0.135 \cr
    & 200 & & 0.012 &  0.002 &    0.009 &  0.157 &  0.002 &  0.048 &  0.002 &  0.000 &  0.138 \cr
    & 400 & & 0.007 &  0.002 &    0.006 &  0.157 &  0.002 &  0.071 &  0.002 &  0.000 &  0.219 \cr
   \multirow{3}{*}{100} & 100 & & 0.023 &  0.002 &    0.011 &  0.156 &  0.002 &  0.048 &  0.002 &  0.000 &  0.122 \cr
    & 200 & & 0.009 &  0.001 &    0.007 &  0.161 &  0.001 &  0.062 &  0.001 &  0.000 &  0.165 \cr
    & 400 & & 0.005 &  0.001 &    0.004 &  0.158 &  0.001 &  0.036 &  0.001 &  0.000 &  0.101
   \end{tabular}
\end{table}
\end{landscape}

\newpage


\begin{landscape}

 \begin{table}

\caption{ Mean Absolute Error (MAE) and misclassification results: $\sigma^2 = 1$,  $M=2$}
     \hspace*{.2cm}
\begin{tabular}{cccccccccccc}

\hline\hline
 & & &  &  & \multirow{2}{*}{$G=2$} &   &  &\tabularnewline
 & & &  &  &  &  &  & & \tabularnewline
\hline\hline 
\multirow{2}{*}{$ \ \ T\ \ $} & {\multirow{2}{*}{ $\ \ N \ \ $}} & \multicolumn{1}{c|}{} & \multicolumn{7}{c|}{ Mean Absolute Error} & \multicolumn{2}{c}{Misclassification}\tabularnewline
 & & \multicolumn{1}{c|}{} & $ \ $S$ \ $ & \ P-S \ &  LS &  Pen NN  & $ \ \ \ $I-GFE$ \ \ \ $ & $ \ \ \ $GFE$ \ \ \ $ &\multicolumn{1}{c|}{$ \ \ $Oracle$ \ \ $} & $ \ $S$ \ $ & $ \ \ \ $GFE$ \ \ $\tabularnewline
\hline 
\multirow{3}{*}{20} & 100 & &  0.048 &  0.019 &    0.016 &  0.171 &  0.015 &  0.016 &  0.014 &  0.015 &  0.004 \cr
    & 200 & & 0.028 &  0.009 &    0.011 &  0.169 &  0.008 &  0.008 &  0.008 &  0.004 &  0.005 \cr
    & 400 & & 0.019 &  0.007 &    0.008 &  0.169 &  0.007 &  0.007 &  0.006 &  0.003 &  0.003 \cr
\multirow{3}{*}{50} & 100 & & 0.022 &  0.008 &    0.010 &  0.168 &  0.008 &  0.008 &  0.008 &  0.000 &  0.000 \cr
    & 200 & & 0.016 &  0.007 &    0.007 &  0.168 &  0.007 &  0.007 &  0.007 &  0.000 &  0.000 \cr
    & 400 & & 0.009 &  0.003 &    0.006 &  0.168 &  0.003 &  0.003 &  0.003 &  0.000 &  0.000 \cr
\multirow{3}{*}{100} & 100 & & 0.013 &  0.005 &   0.008 &  0.168 &  0.005 &  0.005 &  0.005 &  0.000 &  0.000 \cr
    & 200 & & 0.011 &  0.004 &   0.005 &  0.169 &  0.004 &  0.004 &  0.004 &  0.000 &  0.000 \cr
    & 400 & & 0.006 &  0.003 &   0.004 &  0.167 &  0.003 &  0.003 &  0.003 &  0.000 &  0.000 \cr

\hline\hline
 & & &  &  & \multirow{2}{*}{$G=7$} &   & & &\tabularnewline
 & & &  &  &  &  &  & &\tabularnewline
\hline\hline 
\multirow{2}{*}{$ \ \ T\ \ $} & {\multirow{2}{*}{ $\ \ N \ \ $}} & \multicolumn{1}{c|}{} & \multicolumn{7}{c|}{ Mean Absolute Error} & \multicolumn{2}{c}{Misclassification}\tabularnewline
 & & \multicolumn{1}{c|}{} & $ \ $S$ \ $ & \ P-S \ &  LS &  Pen NN  & $ \ \ \ $I-GFE$ \ \ \ $ & $ \ \ \ $GFE$ \ \ \ $ &\multicolumn{1}{c|}{$ \ \ $Oracle$ \ \ $} & $ \ $S$ \ $ & $ \ \ \ $GFE$ \ \ $\tabularnewline
\hline 
\multirow{3}{*}{20} & 100 & &  0.104 &  0.069 &   0.146 &  0.167 &  0.015 &  0.166 &  0.012 &  0.169 &  0.639 \cr
    & 200 & &  0.075 &  0.057 &    0.108 &  0.169 &  0.012 &  0.159 &  0.009 &  0.144 &  0.641 \cr
    & 400 & & 0.038 &  0.046 &    0.040 &  0.170 &  0.014 &  0.153 &  0.006 &  0.103 &  0.608 \cr
   \multirow{3}{*}{50} & 100 & & 0.050 &  0.013 &    0.094 &  0.169 &  0.009 &  0.081 &  0.009 &  0.004 &  0.204 \cr
    & 200 & & 0.024 &  0.006 &    0.011 &  0.168 &  0.005 &  0.014 &  0.005 &  0.002 &  0.019 \cr
    & 400 & & 0.014 &  0.004 &    0.006 &  0.167 &  0.004 &  0.004 &  0.004 &  0.000 &  0.000 \cr
   \multirow{3}{*}{100} & 100 & & 0.025 &  0.006 &    0.011 &  0.166 &  0.006 &  0.017 &  0.006 &  0.000 &  0.019 \cr
    & 200 & & 0.012 &  0.004 &    0.006 &  0.168 &  0.004 &  0.004 &  0.004 &  0.000 &  0.000 \cr
    & 400 & & 0.007 &  0.003 &    0.003 &  0.168 &  0.003 &  0.003 &  0.003 &  0.000 &  0.000
   \end{tabular}
\end{table}
\end{landscape}

\newpage


\begin{landscape}

 \begin{table}

\caption{ Mean Absolute Error (MAE) and misclassification results: $\sigma^2 = 4$,  $M=2$}
     \hspace*{.2cm}
\begin{tabular}{cccccccccccc}

\hline\hline
 & & &  &  & \multirow{2}{*}{$G=2$} &  &  & &\tabularnewline
 & & &  &  &  &  &  & &\tabularnewline
\hline\hline 
\multirow{2}{*}{$ \ \ T\ \ $} & {\multirow{2}{*}{ $\ \ N \ \ $}} & \multicolumn{1}{c|}{} & \multicolumn{7}{c|}{ Mean Absolute Error} & \multicolumn{2}{c}{Misclassification}\tabularnewline
 & & \multicolumn{1}{c|}{} & $ \ $S$ \ $ & \ P-S \ &   LS &  Pen NN  & $ \ \ \ $I-GFE$ \ \ \ $ & $ \ \ \ $GFE$ \ \ \ $ &\multicolumn{1}{c|}{$ \ \ $Oracle$ \ \ $} & $ \ $S$ \ $ & $ \ \ \ $GFE$ \ \ $\tabularnewline
\hline 
 \multirow{3}{*}{20} & 100 & & 0.036 &  0.005 &    0.006 &  0.167 &  0.005 &  0.005 &  0.005 &  0.000 &  0.000 \cr
    & 200 & &  0.017 &  0.003 &    0.004 &  0.169 &  0.003 &  0.003 &  0.003 &  0.000 &  0.000 \cr
    & 400 & & 0.012 &  0.003 &    0.004 &  0.167 &  0.003 &  0.003 &  0.003 &  0.000 &  0.000 \cr
   \multirow{3}{*}{50} & 100 & & 0.016 &  0.003 &    0.004 &  0.122 &  0.003 &  0.003 &  0.003 &  0.000 &  0.000 \cr
    & 200 & & 0.012 &  0.002 &    0.003 &  0.087 &  0.002 &  0.002 &  0.002 &  0.000 &  0.000 \cr
    & 400 & & 0.006 &  0.001 &    0.002 &  0.038 &  0.001 &  0.001 &  0.001 &  0.000 &  0.000 \cr
   \multirow{3}{*}{100} & 100 & & 0.009 &  0.002 &    0.003 &  0.030 &  0.002 &  0.002 &  0.002 &  0.000 &  0.000 \cr
    & 200 & & 0.007 &  0.002 &    0.002 &  0.009 &  0.002 &  0.002 &  0.002 &  0.000 &  0.000 \cr
    & 400 & & 0.005 &  0.001 &    0.001 &  0.001 &  0.001 &  0.001 &  0.001 &  0.000 &  0.000 \cr
    
\hline\hline
 & & &  &  & \multirow{2}{*}{$G=7$} &  &  & &\tabularnewline
 & & &  &  &  &  &  & &\tabularnewline
\hline\hline 
\multirow{2}{*}{$ \ \ T\ \ $} & {\multirow{2}{*}{ $\ \ N \ \ $}} & \multicolumn{1}{c|}{} & \multicolumn{7}{c|}{ Mean Absolute Error} & \multicolumn{2}{c}{Misclassification}\tabularnewline
 & & \multicolumn{1}{c|}{} & $ \ $S$ \ $ & \ P-S \ &  LS &  Pen NN  & $ \ \ \ $I-GFE$ \ \ \ $ & $ \ \ \ $GFE$ \ \ \ $ &\multicolumn{1}{c|}{$ \ \ $Oracle$ \ \ $} & $ \ $S$ \ $ & $ \ \ \ $GFE$ \ \ $\tabularnewline
\hline 
  \multirow{3}{*}{20} & 100 & & 0.103 &  0.006 &    0.007 &  0.164 &  0.004 &  0.004 &  0.004 &  0.003 &  0.000 \cr
    & 200 & & 0.063 &  0.003 &    0.005 &  0.169 &  0.003 &  0.003 &  0.003 &  0.000 &  0.000 \cr
    & 400 & & 0.028 &  0.002 &    0.004 &  0.169 &  0.002 &  0.002 &  0.002 &  0.000 &  0.000 \cr
   \multirow{3}{*}{50} & 100 & & 0.036 &  0.003 &    0.005 &  0.171 &  0.003 &  0.003 &  0.003 &  0.000 &  0.000 \cr
    & 200 & & 0.015 &  0.002 &    0.003 &  0.169 &  0.002 &  0.002 &  0.002 &  0.000 &  0.000 \cr
    & 400 & & 0.011 &  0.001 &    0.002 &  0.168 &  0.001 &  0.001 &  0.001 &  0.000 &  0.000 \cr
   \multirow{3}{*}{100} & 100 & & 0.021 &  0.002 &    0.003 &  0.173 &  0.002 &  0.002 &  0.002 &  0.000 &  0.000 \cr
    & 200 & & 0.010 &  0.001 &    0.002 &  0.167 &  0.001 &  0.001 &  0.001 &  0.000 &  0.000 \cr
    & 400 & & 0.005 &  0.001 &    0.001 &  0.166 &  0.001 &  0.001 &  0.001 &  0.000 &  0.000
   \end{tabular}
\end{table}
\end{landscape}

\end{document}